%% file: main.tex
\documentclass{article}
\usepackage[utf8]{inputenc}
\usepackage{amsmath}
\usepackage{dsfont}
\usepackage{comment}
\usepackage{amssymb}
\usepackage{amsthm}
\usepackage{thm-restate}
\usepackage{mathtools}
\usepackage {tikz}
\usepackage[margin=1in]{geometry}
\usepackage{systeme}
\usepackage{colortbl}
\usepackage{nicefrac}
\usepackage{multicol}
\usepackage{multirow}
\usepackage{xspace}
\usepackage{enumitem}
\usepackage{hyperref}
\usepackage{mdframed}
\usepackage{ stmaryrd }
\usepackage{makecell}
\usepackage{bm}
\usepackage[capitalize,nameinlink]{cleveref}
\usepackage{hyperref}
\hypersetup{colorlinks,linkcolor={blue},citecolor={blue},urlcolor={red}} 
\newtheorem{theorem}{Theorem}[section]

\newtheorem{lemma}[theorem]{Lemma}
\newtheorem{claim}[theorem]{Claim}
\theoremstyle{definition}
\newtheorem{definition}{Definition}[section]
\theoremstyle{definition}

\newtheorem{remark}[theorem]{Remark}

\DeclareMathOperator{\ex}{ex}
\DeclareMathOperator{\supp}{supp}

\DeclareMathOperator{\poly}{poly}
\DeclareMathOperator{\abs}{abs}
\DeclareMathOperator{\vol}{vol}

\newcommand{\R}{\mathbb{R}}
\newcommand{\E}{\mathbb{E}}
\newcommand{\Z}{\mathbb{Z}}
\newcommand{\N}{\mathbb{N}}

\newcommand{\bmu}{\boldsymbol{\mu}}

\newcommand{\bb}{\mathbf{b}}
\DeclareMathOperator{\bc}{\mathbf{c}}
\DeclareMathOperator{\bd}{\mathbf{d}}

\newcommand{\bn}{\boldsymbol{n}}

\newcommand{\bp}{\boldsymbol{p}}

\newcommand{\br}{\boldsymbol{r}}

\newcommand{\bv}{\boldsymbol{v}}
\newcommand{\bw}{\boldsymbol{w}}
\newcommand{\bx}{{\boldsymbol{x}}}

\DeclareMathOperator{\bF}{\mathbf{F}}
\DeclareMathOperator{\bbF}{\bar{\mathbf{F}}}

\DeclareMathOperator{\bM}{\mathbf{M}}
\DeclareMathOperator{\bbM}{\bar{\mathbf{M}}}

\newcommand{\tdeg}{\widetilde{\deg}}
\newcommand{\tvol}{\widetilde{\vol}}
\DeclareMathOperator{\diag}{\text{diag}}

\newcommand\inner[2]{\left\langle #1, #2 \right\rangle}
\newcommand{\psip}{\cev{\psi}}
\newcommand{\bmup}{\cev{\bmu}}
\newcommand{\cbF}{\cev{\bF}}

\usepackage{tikz}
\usepackage{scalerel}
\usetikzlibrary{calc}
\usetikzlibrary{decorations.pathmorphing}
\newcommand{\sqcev}[1]{%
  \ThisStyle{%
    \setbox0=\hbox{$\SavedStyle#1$}%
    \tikz[baseline=(X.base)]{
      \node[inner sep=0pt] (X) {\copy0};
      \draw[<-, decorate, decoration={
          snake,
          amplitude=0.1mm,
          segment length=1mm
        }]
      ([yshift=1.5ex]$(X.west)!0.15!(X.east)$) 
        -- 
        ([yshift=1.5ex]$(X.west)!0.85!(X.east)$);
    }%
  }%
}

\newcommand{\tcbF}{\sqcev{\bF}}
\newcommand{\tbmup}{\sqcev{\bmu}}

\usepackage{float}
\floatstyle{ruled}
\newfloat{algorithm}{tbp}{loa}
\providecommand{\algorithmname}{Algorithm}
\floatname{algorithm}{\protect\algorithmname}
\newcommand{\tab}{.\hskip.1in}

\makeatletter
\DeclareRobustCommand{\cev}[1]{%
  \mathpalette\do@cev{#1}%
}
\newcommand{\do@cev}[2]{%
  \fix@cev{#1}{+}%
  \reflectbox{$\m@th#1\vec{\reflectbox{$\fix@cev{#1}{-}\m@th#1#2\fix@cev{#1}{+}$}}$}%
  \fix@cev{#1}{-}%
}
\newcommand{\fix@cev}[2]{%
  \ifx#1\displaystyle
    \mkern#23mu
  \else
    \ifx#1\textstyle
      \mkern#23mu
    \else
      \ifx#1\scriptstyle
        \mkern#22mu
      \else
        \mkern#22mu
      \fi
    \fi
  \fi
}

\usepackage{amsbsy}
\newcommand*\Bell{\ensuremath{\boldsymbol\ell}}
\newcommand*\Bf{\ensuremath{\boldsymbol f}}
\newcommand*\Bc{\ensuremath{\boldsymbol c}}
\newcommand*\Bn{\ensuremath{\boldsymbol n}}
\newcommand*\Bp{\ensuremath{\boldsymbol p}}

\newcommand{\link}{\textsc{Link}}
\newcommand{\cut}{\textsc{Cut}}
\newcommand{\fmin}{\textsc{FindMin}}
\newcommand{\add}{\textsc{Add}}
\newcommand{\froot}{\textsc{FindRoot}}

\usepackage{todonotes}

\usepackage{setspace}
\usepackage{authblk}

\title{Improved Directed Expander Decompositions
}
\author[1]{Henry Fleischmann}
\author[1]{George Z. Li}  
\author[1]{Jason Li}

\affil[1]{Carnegie Mellon University $\{\texttt{hfleisch}, \texttt{gzli}, \texttt{jmli}\}$\texttt{@cs.cmu.edu}}
\date{}

\begin{document}

\maketitle

\begin{abstract}

We obtain faster expander decomposition algorithms for directed graphs, matching the guarantees of Saranurak and Wang (SODA 2019) for expander decomposition on undirected graphs. Our algorithms are faster than prior work and also generalize almost losslessly to capacitated graphs. In particular, we obtain the first directed expander decomposition algorithm for capacitated graphs in near-linear time with optimal dependence on $\phi$.

To obtain our result, we provide the first implementation and analysis of the non-stop cut-matching game for directed, capacitated graphs. All existing directed expander decomposition algorithms instead temporarily add ``fake edges'' before pruning them away in a final cleanup step. Our result shows that the natural undirected approach applies even to directed graphs. The difficulty is in its analysis, which is technical and requires significant modifications from the original setting of undirected graphs.

\end{abstract}
\pagenumbering{gobble}

\clearpage
\tableofcontents

\newpage 
\pagenumbering{arabic}
\input{intro}

\input{technical-overview}

\input{organization}

\input{preliminaries}

\input{cut-matching}

\input{weak-expander-decomp}

\input{strong-expander-decomposition}

\section*{Acknowledgments}
This material is based upon work supported by the National Science Foundation Graduate
Research Fellowship Program under Grant No.\ DGE2140739. Any opinions, findings,
and conclusions or recommendations expressed in this material are those of the author(s)
and do not necessarily reflect the views of the National Science Foundation.

\bibliographystyle{alpha}
\bibliography{references}

\appendix

\end{document}

%% file: intro.tex
\section{Introduction}

A (strong) \emph{expander decomposition} of a graph is a partition into well-connected components, called \emph{expanders}, with relatively few edges in between the components. More formally, an expander decomposition of a graph $G=(V,E)$ is a partition $V_1,\ldots,V_k$ of the vertices such that each induced graph $G[V_i]$ is a $\phi$-expander, and there are at most $\tilde O(\phi m)$ edges in between the components.\footnote{In this paper, $\tilde O(\cdot)$ notation suppresses polylogarithmic factors in $n$, the number of vertices of the graph.} Expander decompositions were introduced by~\cite{kannan2004clusterings} and popularized by Spielman and Teng's undirected Laplacian solvers~\cite{spielman2004nearly}. Since then they have become the main driving force behind recent breakthroughs on undirected graphs, including (undirected) maximum flow~\cite{kelner2014almost,bernstein2022deterministic,abboud2023all}, dynamic graph problems~\cite{nanongkai2017dynamic,goranci2021expander,chuzhoy2021decremental,bernstein2022deterministic,chen2024almost,kyng2024dynamic,el2025fully}, and deterministic minimum cut~\cite{kawarabayashi2018deterministic,li2020deterministic,henzinger2024deterministic}. Since expander decompositions have become a key primitive in modern graph algorithms, the problem of computing an expander decomposition itself has been extensively studied~\cite{spielman2004nearly,nanongkai2017dynamic,henzinger2020local}, culminating in the near-optimal algorithm of Saranurak and Wang~\cite{saranurak2019expander}.\footnote{Recent results~\cite{agassy2022expander,agassy2025expanderdecompositionnonuniformvertex} improve the guarantees of~\cite{saranurak2019expander} by a logarithmic factor. However, to do this, they exhibit a nonstop cut-matching game using a spectral cut player~\cite{orecchia2008partitioning}. These spectral techniques do not appear to easily generalize to the directed setting, not least because of the lack of symmetric adjacency matrices. Moreover, these improvements only directly yield ``cut-based'' weak-expander decompositions, which are too weak for the applications of our work.}

More recently, expander decompositions have been generalized to \emph{directed} graphs that consider directed acyclic graphs (DAGs) as a sort of base case. More formally, a directed expander decomposition of a graph $G=(V,E)$ is a partition $V_1,\ldots,V_k$ and an acyclic subgraph $D$ of $G$, such that each induced subgraph $G[V_i]$ is a (directed) $\phi$-expander and there are at most $\tilde O(\phi m)$ edges in $G-D$ in between the components. To understand why a DAG must be excluded from the inter-component edges, consider the case when the graph $G=(V,E)$ itself is a DAG. The only possible partition into (directed) expanders is to make each vertex its own singleton component, since any induced graph with more than one vertex is not strongly connected and therefore cannot be an expander. In other words, expander decompositions on DAGs must be trivial, with all $m$ edges in between the components, so the only way to bound the number of inter-component edges is to exclude edges from a DAG (which, in this case, is all the edges).

Fortunately, DAGs are well-structured enough as a base case, and directed expander decompositions were key to recent breakthroughs for problems on directed graphs, including (directed) maximum flow~\cite{chuzhoy2024maximum,bernstein2024maximum} and dynamic shortest path~\cite{bernstein2020deterministic}. However, all of these algorithms have extra $n^{o(1)}$ factors in their running times, partially because, at the time, computing a directed expander decomposition itself required $m^{1+o(1)}$ time. Hence, a more recent line of work focused on speeding up directed expander decompositions~\cite{bernstein2020deterministic,DBLP:conf/soda/HuaKGW23,SP24}, the last of which obtained a near-optimal algorithm, but only for uncapacitated graphs and with large polylogarithmic factors in their overhead.\footnote{There is a reduction from capacitated to uncapacitated as shown in~\cite{van2024almost},  but this incurs an extra $1/\phi$ factor in the running time. At a high level, the reduction first removes all edges whose capacity is less than a $\phi/m$ fraction of the total capacity, which is at most $\phi$-fraction in total. Then, scale the remaining edge capacities appropriately, and round them to integers. The resulting graph can be viewed as an uncapacitated multi-graph with $O(m/\phi)$ multi-edges.}

In this work, we continue the development of directed expander decompositions, obtaining a faster algorithm for directed expander decomposition that also works for capacitated graphs. Our approach is also more faithful to Saranurak-Wang's algorithm for undirected graphs. For uncapacitated graphs, we even \emph{match} their polylogarithmic overhead: the following statement has the exact same parameters as in Saranurak-Wang, except generalized to the directed case.

\begin{theorem}[Uncapacitated Expander Decomposition]
Given a directed, uncapacitated graph $G=(V,E)$ and a parameter $\phi$, there is a randomized algorithm that computes a partition $V_1,\ldots,V_k$ and an acyclic subgraph $D$ such that each induced subgraph $G[V_i]$ is a $\phi$-expander and the number of inter-component edges in $G-D$ is $O(\phi m\log^3n)$. The algorithm runs in time $O(m\log^4n/\phi)$.
\end{theorem}

In comparison, the (static) algorithm of~\cite{SP24} has $\Omega(\phi/\log^{12}n)$-expanders for induced subgraphs,  $O(\phi m\log^{19}n)$ inter-component edges, and runs in time $O(m\log^{20}n/\phi)$. 

For capacitated graphs, we obtain the first directed expander decomposition algorithm with near-optimal factors and optimal dependence on $\phi$. We introduce a few more logarithmic factors due to technical reasons, but still greatly improve upon the bounds of~\cite{SP24}.

\begin{theorem}[Capacitated Expander Decomposition]
Given a directed, capacitated graph $G=(V,E)$ with edge capacities in $\{1,2,\ldots,W\}$ and a parameter $\phi$, there is a randomized algorithm that computes a partition $V_1,\ldots,V_k$ and an acyclic subgraph $D$ such that each induced subgraph $G[V_i]$ is a $\phi$-expander and the total capacity of inter-component edges in $G-D$ is $O(\phi m\log^8(nW))$. The algorithm runs in time $O(m/\phi+m\log^7(nW))$.
\end{theorem}

Finally, we consider the simpler task of computing a \emph{weak-expander decomposition}, where the induced subgraphs $G[V_i]$ are not required to be $\phi$-expanders, but should still satisfy an expander-like property. Weak-expander decompositions were recently shown to be powerful enough for approximate maximum flow for undirected graphs~\cite{li2025congestion,approxflow} and (exact) maximum flow for directed graphs~\cite{BBLST25}. In all three results, observing that a weak-expander decomposition suffices was the key insight to obtaining polylogarithmic factors in the running time, as opposed to $n^{o(1)}$ factors. We note that our result is explicitly used in the latter result, which requires a weak-expander decomposition primitive. One of our key technical contributions, \textit{expander grafting}, is applied in the second approximate maximum flow result.
\begin{theorem} [Weak-Expander Decomposition] \label{thm: informal near ex decomp}
Given a directed capacitated graph $G = (V, E)$ with edge capacities in $\{1,2, \ldots, W\}$ and a parameter $\phi$, there is a randomized algorithm that computes a partition $V = V_1, V_2, \cdots V_\ell$ and an acyclic subgraph $D$ such that each induced subgraph $G[V_i]$ is a $\phi$-near expander and the total capacity of inter-component edges in $G - D$ is $O(\phi m \log^3 (nW))$. The algorithm runs in time $O(m \log^3 (nW)/\phi)$.
\end{theorem}

In fact, we state our weak-expander decomposition bounds in their full generality, in the context of \emph{vertex weightings}. This context is required for the application in~\cite{BBLST25}. Since the the full statement is fairly technical, we defer the reader to~\Cref{thm: near ex decomp} in \Cref{sec: weak exp}.

%% file: technical-overview.tex
\section{Technical Overview}

At a high level, we follow Saranurak-Wang's analysis of the \emph{non-stop cut-matching game}. The cut-matching game is an iterative procedure first introduced by~\cite{khandekar2009graph} as a method of approximating the sparsest cut of an undirected graph. Since their main focus was sparsest cut, the original cut-matching game terminated upon finding a single sparse cut. A ``non-stop'' version was developed by~\cite{racke2014computing} that continues the iterative process even after a sparse cut has been found, and this version was adopted by Saranurak-Wang's expander decomposition algorithm for undirected graphs.

However, in the directed graph setting, the non-stop cut matching game poses many technical challenges. Prior work on directed graph expander decompositions~\cite{bernstein2020deterministic,DBLP:conf/soda/HuaKGW23} avoid the non-stop version, and work with the original version of~\cite{khandekar2009graph} adapted to directed graphs by~\cite{louis2010directed}. In the event of a sparse cut, they add ``fake'' edges to force the cut-matching game to continue to completion (or until a balanced sparse cut is found). Since the fake edges are not actual edges, they are then ``pruned'' from the graph. The main contribution of~\cite{SP24} is a near-optimal pruning procedure that achieves near-optimal running time, which features a new dynamic flow primitive called \emph{push-pull-relabel} that we also use in our paper.

This approach succeeded in obtaining a near-optimal  algorithm~\cite{SP24}, but with high polylogarithmic factors and only for uncapacitated graphs. As mentioned before, there is a reduction from capacitated graphs to uncapacitated graphs, losing an additional factor of $1/\phi$. We remark that previous works on directed expander decompositions operate in the dynamic setting and obtain additional useful properties (e.g., maintaining that the decompositions are iterative refinements). Our work instead focuses on the static setting with the goal of extending the high-level undirected algorithm to directed graphs; this suffices for our main max flow application~\cite{bernstein2024maximum, BBLST25}.  As the static algorithm in~\cite{saranurak2019expander} can be made dynamic by (now) standard means, we expect that our algorithm can be as well, but that is beyond the scope of the current work.

In this work, we provide the first implementation and analysis of the non-stop cut-matching game for directed, capacitated graphs. In the case of uncapacitated directed graphs, we even match the guarantees of Saranurak-Wang, suggesting that directed graphs are no harder for expander decomposition algorithms. (For capacitated directed graphs, we lose extra factors due to technical reasons but see no fundamental barrier to matching Saranurak-Wang as well.)

For weak-expander decomposition, the non-stop cut-matching game is all we need, and a simple recursion establishes the desired bounds. (In comparison, \cite{SP24} would still need their pruning procedure for weak-expander decomposition due to the presence of fake edges.) For (strong) expander decomposition, we require a similar pruning procedure as in~\cite{SP24}. We adapt their push-pull-relabel primitive to capacitated graphs, which involves ``regularizing'' the vertex degrees to control the running time.

For the rest of the technical overview, we present our techniques in more detail.

\subsection{The Cut-Matching Game for Directed Graphs}
Our starting point for studying the cut-matching game on directed graphs is the work of Louis~\cite{louis2010directed}. In that work, Louis extends the algorithm of~\cite{khandekar2009graph} to the directed setting. We emphasize certain aspects that differ from~\cite{khandekar2009graph}. Roughly, the algorithm is as follows.\footnote{Technically, the algorithm in \cite{louis2010directed} only approximates the \textit{expansion} of the graph $G$ and is only written for unweighted graphs. The expansion of a graph differs from its conductance in that the denominator terms are the number of vertices in the sets as opposed to their respective volumes. It is straightforward to extend Louis's algorithm to approximating conductance for weighted graphs, so we do so in our discussion here.} Let $G = (V,E, \bc)$ be a directed graph of order $n$ and size $m$, and let $\phi > 0$ be the ``expansion'' parameter. Our goal is to either:
\begin{enumerate}
    \item \textbf{Certify expansion:} Certify that $G$ is a $O(\phi/\log^2 n)$-expander.
    \item \textbf{Find a sparse cut:} Find $S \subseteq V$ so that $\Phi_G(S) := \frac{\min(\delta_G(S, \overline{S}), \delta_G(\overline{S}, S))}{\min(\vol_G(S), \vol_G(\overline{S}))} \leq \phi$. 
\end{enumerate}

The algorithm proceeds over $O(\log^2 n)$ rounds. Throughout, we maintain an implicit all-to-all multicommodity flow between the vertices, which embeds into $G$ with low congestion. We represent this by a matrix $\bF \in \R_{\geq 0}^{V \times V}$, where the row $\bF(u)$ is the vector encoding the flow from $u$ to each other vertex in $V$. We initialize $\bF$ as $\diag(\deg_G)$, i.e., each vertex $u \in V$ starts with $\deg_G(u)$ of its own commodity. It turns out that it suffices to show that the flow vectors of all $u \in V$ are very close to uniformly spreading out their commodity in order to certify expansion (i.e., the entry of $\bF(u)$ corresponding to $v$ is roughly $\frac{\deg_G(v)}{\vol_G(V)} \cdot \deg_G(u)$).

In each round, the \textit{cut player} computes a (weighted) \textcolor{purple}{bisection} $L \sqcup R = V$, where the flow vectors of $u \in L$ are ``far'' from those of $v \in R$. Subsequently, the \textit{matching player} attempts to embed (weighted, directed) matchings from $L$ to $R$ \textcolor{purple}{and} from $R$ to $L$, each with low congestion. 

In the case that they fail to embed either of those matchings, they find a sparse cut $S$ and the algorithm terminates. Otherwise, the flow vectors are updated according to the embedded matchings, uniformizing them by \textcolor{purple}{``averaging''} matched flow vectors. Notably, by the choice of cut found by the cut player, these averaged flow vectors were previously far from each other. The progress made towards uniformity in each round is governed by a potential function. By standard concentration bounds of subgaussian random variables, the potential decreases by a $1 - \Omega(1/\log n)$ factor in each round in expectation. The potential also has the property that, upon being sufficiently small, it certifies expansion of $G$. This is then achieved after $O(\log^2 n)$ rounds.

In the undirected setting, the cut player computes a very unbalanced cut (where $L$ is much smaller than $R$), and the matching player only attempts to embed a matching from $L$ into $R$ with low congestion. Then the flow vectors of the matched vertices are averaged. Importantly, the averaging is bidirectional: each sends some portion of its flow vector to the other. This is in sharp contrast to the directed case where $u$ being matched to $v$ means that $\bF(u)$ will receive some portion of $\bF(v)$, but $v$ may not receive any of $\bF(u)$. Indeed, $u$ can send flow to vertices that $v$ can send to via its embedded path, but the converse may not be true. 

This basic cut-matching game routine inspires a naive algorithm for computing an expander decomposition (in both the directed and undirected settings). Run the cut-matching game. In the case of finding a sparse cut, stop and recursively run the routine again on both of sides of the cut. Upon ever certifying expansion (or reaching a singleton node), the algorithm terminates. Unfortunately, this algorithm is inefficient. Each round of cut-matching will take $\Omega(m)$ time, since it involves solving a flow problem for the matching step, and the most basic cut-matching game yields no guarantees about the balancedness of the sparse cut found when not certifying expansion. Hence, the depth of the recursion could be as large as $n$, resulting in an $\Omega(mn)$ time algorithm. 

The obvious solution to this inefficiency is to somehow ensure that, upon termination of each round of the cut-matching game, the components we recurse into are reasonably balanced. This is exactly what is ensured by the \textit{nonstop cut-matching game} introduced in~\cite{racke2014computing} and leveraged in~\cite{saranurak2019expander}, both in the case of undirected graphs.
The essential idea is simple: upon failing to embed a matching in the matching step, the algorithm finds a sparse cut. In the case of implementing the matching step with a max-flow oracle call, this sparse cut is precisely the min-cut resultant from attempting to solve the weighted matching flow instance. As such, the source on vertices not in the sparse cut is routed, and we make some progress towards uniformizing the flow vectors corresponding to vertices not in the cut. (In general, we will not necessarily call a max-flow oracle, but this property remains roughly true for other approximate flow routines, like bounded height push-relabel or unit flow.) Then, instead of stopping and outputting a sparse cut, the algorithm continues, restricting the set of vertices being cut and matched to those not removed in a sparse cut thus far. If at any point the collection of sparse cuts removed is large, e.g., they account for a polylogarithmic fraction of the total volume in $G$, we can safely terminate early and recurse onto each discovered cut. Otherwise, after $O(\log^2 n)$ rounds, we will still certify expansion for the remaining vertices (albeit, a weaker form of expansion, since remaining vertices may have been matched to vertices that have since been removed). Since we do not need to recurse onto subgraphs we have certified expansion for, this ensures a much more agreeable polylogarithmic recursion depth. This is the strategy employed to great effect in~\cite{saranurak2019expander}, yielding an undirected expander decomposition algorithm with polylogarithmic factors far more reasonable than the more than $20$ present in the current state of the art for directed expander decompositions~\cite{SP24}.

\paragraph{Problem 1: directed matchings.} The first and most direct challenge in extending the directed cut-matching game of~\cite{louis2010directed} to a nonstop cut-matching game is the fact that the matching player computes two directed matchings in each step. When both matchings embed with low congestion, there is no issue: the flow vectors $\bF(u)$ all send and receive the same amount of flow, so the sum of their entries remains constant. However, what happens during matching steps where $u \in V$ does not belong to any of the sparse cuts found, but still is matched more in one direction than the other? The $v$\textsuperscript{th} entry of $\bF(u)$ represents the amount of flow sent from $u$ to $v$. The sum of these entries should stay constant throughout, since $u$ started with $\deg_G(u)$ of its own commodity. However, if $u$ is unevenly matched, this invariant will be disrupted. For example, if $u$ sends all of its flow in one matching but receives none in the other matching, it will accumulate additional flow without sending any and the sum of the entries of $\bF(u)$ will increase.

This is deeply problematic for the potential analysis showing that we make progress towards certifying expansion in each round. In the undirected setting, this was a nonissue because the matchings computed by the matching player were undirected, so either a vertex was matched, or it was unmatched; no vertices were matched only in one direction and not the other.

\paragraph{Problem 2: the cut player finds bisections.}
This issue is slightly more subtle. In each round, the cut player must find bisections, rather than arbitrary partitions, since otherwise Problem 1 is exacerbated even further: we want to avoid partially matched vertices as much as possible. 

But, in the undirected setting, the cuts are explicitly chosen to be unbalanced, ensuring that in the matching step the amount of total source is dominated by the total sink. This is a desirable property due to a hiccup in the analysis that we glossed over in our sketch of the algorithm. Again, supposing that we use a max-flow oracle in our implementation of the matching step, if we find a sparse cut, it corresponds to the min-cut from the flow problem. Then we have that all source outside of the cut is routed, so we make some progress matching the remaining vertices. However, what if the cut we find is extremely unbalanced and the small side of the cut is outside of the min-cut? 

In that case, if we remove the min-cut side of the sparse cut, we lose all the benefits of our nonstop cut-matching game. We will immediately trigger the early termination condition, and we have no useful guarantees on our recursion depth. On the other hand, if we remove the other side of the sparse cut, we lose out on the property that the source outside of the cut was routed. This is an essential property of the matching step in prior work (e.g., see Lemma B.6 of~\cite{saranurak2019expander}). These problems are compounded further when using faster approximate max-flow routines like bounded height push relabel.

To reiterate, this complication is avoided in the undirected setting by finding unbalanced cuts and ensuring that the total source is dominated by the total sink in the matching player flow instances. Since all of the sink must be saturated in the min-cut, this means that the min-cut cannot ever be too large.

\paragraph{Problem 3: existing potential analyses only certify in-expansion.}
Existing analyses of the cut-matching game following the~\cite{khandekar2009graph} framework (e.g.,~\cite{louis2010directed, saranurak2019expander, li2025congestion},) all employ essentially the same potential function strategy. These ideas extend to weighted graphs, but for simplicity we consider the unweighted setting for this discussion. Let $\bF$ be the matrix representing the implicit multicommodity flow defined by the cut-matching game, and let $A$ be the set of vertices which are not removed in matching steps. Then, prior works consider a potential function which measures how close the rows $\bF$ corresponding to vertices in $A$ are to their average. They show that this potential drops in expectation in each round, so that after $O(\log^2 n)$ rounds of cut-matching, with high-probability, the potential certifies that the rows of $\bF$ corresponding to $A$ are approximately equal up to scaling. That is, each vertex in $A$ sends roughly the same amount of its commodity to each vertex (up to scaling by its degree). 

In the undirected setting, this suffices to certify that $A$ is a near-expander. We sketch this argument. Let $S \subseteq A$ with at most half of the volume. As discussed in Problem 1, we hope to ensure that $\sum_{u \in V} \bF(u,v) = \deg(v)$. Assuming this still holds, then there are two cases. If $\sum_{u \in A} \bF(u,v) \geq \deg(v)/2$ then a constant fraction of that flow arrives from $A \setminus S$. This is because $A \setminus S$ is more than half of the volume of $A$ and all vertices in $A$ send roughly the same amount of flow to each other vertex (up to scaling by its degree). Otherwise, at least $\deg(v)/2$ flow arrives into $S$ from outside of $A$ (from the removed vertices). In total, this shows that $\Omega(\vol(S))$ can flow into $S$ from $V \setminus S$ with low congestion, certifying that $A$ is indeed a near-expander. 

The above argument still works in the directed setting, but notably only certifies \textit{in-expansion}. To certify out-expansion and hence directed expansion, we also need to consider the flow leaving $S$, something the current potential analyses is ill-adapted to control. Note that flow-based mixing arguments such as in~\cite{li2025congestion} suffer from the same issue of only certifying expansion in one direction. 

\paragraph{New technique: grafting.} Partially matched vertices are at the core of the first problem. One simple way to ensure that they are removed is to delete them whenever they arise. Let us temporarily ignore the issue of how much we end up deleting upon deleting partially matched vertices. Now, suppose that, lumping nodes deleted in sparse cuts and other deleted nodes together, we terminate our nonstop cut-matching algorithm without triggering the early termination condition. Assuming all of the relevant analysis works out, in the best case the result is a sequence of sparse cuts, a subset $D$ of vertices deleted externally to sparse cuts, and a subset $A$ certified as (near) expanding in $G$. Unlike in the undirected setting, we do not have a handle on the edges crossing the boundaries between these sets: we have no bound at all on the edges crossing between $D$ and $A$. 

We handle this via a final round of post-processing, attempting to \textit{graft} the vertices in $D$ onto $A$, extending the certified (near)-expansion of $A$ to $A \cup D$.\footnote{We introduce the term grafting here since the gardening analogy is both descriptive and in contrast to the well-established operation of trimming near-expanders into expanders.} Indeed, assuming we can ensure that $D$ is small relative to $A$, we can attempt a final ``matching'' step, attempting to route source from $D$ into $A$ and from $D$ into $A$ in the reversed graph of $G$ (note that the latter operation is somewhat different than a typical matching step). See~\cref{fig: graft}.

\begin{figure}
    \centering
    \input{diagrams/graft-diagram}
    \caption{An illustration of a post-processing grafting step with $D$ the set of deleted nodes and $A$ the certified expanding set. The orange flow paths correspond to the embedding of a matching from $D$ to $A$. The purple flow paths correspond to the embedding of a matching from $D$ to $A$ in the reversed graph.}
    \label{fig: graft}
\end{figure}
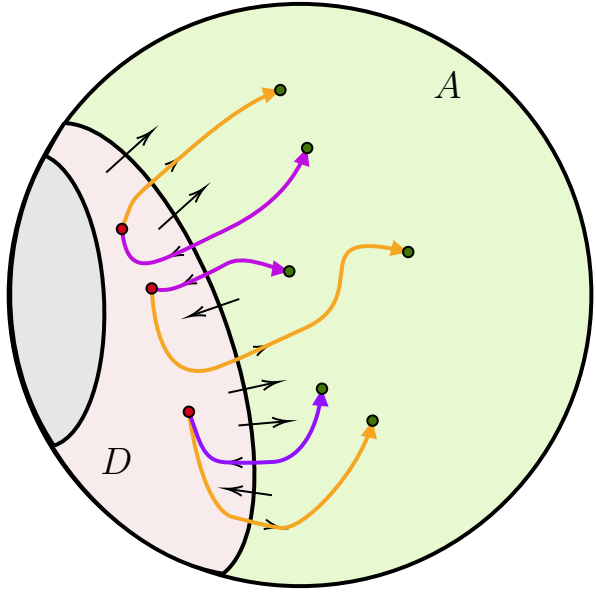

Grafting also solves one of the core problems associated with the cut player finding bisections. In the case of finding extremely unbalanced sparse cuts in the matching step, the problem with cutting out the small side of the cut was the loss of the structural property that source on the remaining vertices was routed. But, by allowing sufficiently unbalanced sparse cuts, we can ensure that almost all of the source on the remaining vertices was routed. Then, we can delete vertices who failed to route their source, adding them to $D$, and grafting them back into $A$ at termination. 

There is one major caveat to this idea: by the end of the grafting post-processing step, we only want nodes in sparse cuts or the certified near-expander. Our standard matching steps no longer guarantee that we only delete nodes in sparse cuts. Nonetheless, by  tweaking our early termination condition, we ensure that these final flow instances have far less source than sink, recovering the crucial property of the matching instances in prior work. We can then  remove the computed min-cut  as our sparse cut as in~\cite{saranurak2019expander}. (In the case of faster approximate max flow routines, we instead appeal to dynamic flow, similar to Lemma B.6 of~\cite{saranurak2019expander}.)

\paragraph{New technique: partial vertex deletions.}
We now revisit the issue of deleting too many vertices not belonging to sparse cuts. This is a very real problem: upon deleting almost all of the graph, we have no choice but to trigger the early termination condition. In this case, grafting cannot help us. Moreover, if only a tiny fraction was deleted in sparse cuts, we again have no reasonable bounds on the recursion depth of our expander decomposition algorithm: one of our recursive cases will include all vertices deleted outside of sparse cuts. Additionally, in an effort to cull vertices which are not fully matched in both directions in each matching step, we might delete vertices with very large degree which were only a single unit of flow from being fully matched.

In the case of unweighted graphs, it is possible to handle this issue by appealing to the split-edge graphs considered in~\cite{racke2014computing, saranurak2019expander} and computing (unweighted) matchings in each matching step. Emulating weighted graphs with parallel edges extends this to graphs with small edge weights. However, for general weighted graphs, the runtime blowup is untenable. We would rather somehow retain the vertex view and find weighted matchings (more akin to Appendix A of~\cite{li2025congestion}, for example).

Our solution is to emulate the split-edge analysis of~\cite{saranurak2019expander} by supporting partial vertex deletions. That is, for each $u \in V$, we have two copies of $u$ named $u_\circ$ and $u_\times$, corresponding to the \textit{active} and \textit{deleted} portions of $u$, respectively. We maintain the invariant that $u_\circ$ is fully matched in each matching step, if necessary retroactively modifying the computed matchings to ensure this. When the $u_\circ$ portion of $u$ corresponds to less than half of the degree of $u$, we fully delete $u$, adding it to $D$. This means that, at termination, although we technically only certify near expansion of the active vertices, we can easily extend this to certify near expansion of all non-deleted vertices at the cost of a factor of two in the expansion.

We design custom cut and matching steps to support partial vertex deletions. The potential function (and its analysis) are also significantly modified from prior work.

\paragraph{New technique: column-based matchings}
To address the final problem, we adapt the cut-matching game to simultaneously ensure that the rows and columns of the implicit multicommodity flow matrix tend towards their weighted average. To this end, we double the number of rounds of cut-matching and spend the latter half of the rounds finding cuts based on the columns as opposed to the rows. To measure the effect of these additional matching steps, we introduce a second potential function based on the columns of the implicit multicommodity flow matrix.

We show that our new potential analysis easily also accommodates reducing this ``column'' potential and that taking steps to reduce one of the potentials does not cause increases in the other. Then, when both potentials are small, we are able to certify both in and out-expansion. 

\paragraph{Other important technical considerations.}
A number of more granular technical issues arise in trying to extend the nonstop cut-matching game analysis from the undirected setting. 

First, the potential analysis of~\cite{saranurak2019expander} (Lemma B.10) crucially relies on the observation that the flow vectors of unmatched vertices remain constant. However, for the directed setting, partially matched vertices completely disrupt this analysis. Even if we were somehow able to ensure that all remaining vertices are completely matched, they might be partially matched to vertices that are removed in sparse cuts. Consequently, we give a novel argument for lower bounding the decrease in potential in each round of the algorithm (\cref{lem: general dec lemma}). One nice secondary implication of this analysis is that it can be used to recover a much simpler proof of convergence for the undirected, weighted setting, avoiding the matrix algebra used in~\cite{li2025congestion}. 

Second, the existing analysis certifying that progress is made on each iteration follows by only considering the progress made for vertices on the small side of the cut found by the cut player (the relevant technical lemma is Lemma 3.3 of~\cite{racke2014computing} or Lemma B.5 of~\cite{saranurak2019expander}). This is at the core of the potential analysis. A priori, this lemma does not extend to bisections. Nonetheless, we prove a new version of this technical lemma supporting bisections (\cref{lem:apd-finding-S}). One interesting feature of the lemma is that it suffices for our analysis to show that there is a subset of one of the two sides of the cut certifying progress; it need not be a subset of a fixed side, and the algorithm does not need to know where it is. 

Additionally, the non-stop cut-matching game in~\cite{saranurak2019expander} outputs a single sparse cut and a near-expanding subset when the early termination condition is not triggered. This is not because they find one single sparse cut over the course of their algorithm; rather, in undirected graphs, the union of sparse cuts is sparse. However, in directed graphs, since a cut is sparse if there are few edges in one of the two directions crossing the cut, the union of sparse cuts is not necessarily sparse. As such, we output a sequence of sparse cuts instead of a singular sparse cut. This detail creates an issue in the matching step, since we now need to solve two directed matching problems, both on the same initial graph. That means that if both find a sparse cut, it is not a priori clear how to add the two sparse cuts to our sequence of sparse cuts. Even if they are disjoint, removing one from the graph first might significantly change the sparsity of the second. We handle this by casework, constructing compatible sparse cuts to add to our sequence. In finding the compatible sparse cuts, we end up deleting additional vertices coming from the initial sparse cuts, which we add to $D$ and defer to our post-processing grafting step.

\subsection{Expander Decompositions for Directed Graphs}

With the non-stop cut-matching game in hand, we now outline how to obtain a strong expander decomposition, following the approach of~\cite{saranurak2019expander}. First, we run the non-stop cut-matching game. If we have the early termination case, we recurse on each component of the sparse cut. We make progress since each component decreases in size by a sufficient multiplicative factor. In the case where we find a near expander $A$, we try to find a real expander $A'\subseteq A$. We do this via a trimming algorithm, which iteratively removes vertices from $A$ until we can certify that $A$ is a real expander. We give details below.

\paragraph{Flow instance: certifying expansion.} In the case where the cut-matching game finds a near-expander, we continue to follow the approach of \cite{saranurak2019expander}. Given a $\phi$-near expander $A$ in $G$, they show that they can efficiently find a large $\Omega(\phi)$-expander $A'\subseteq A$. To this end, they design a flow instance on $G[A]$ such that if the flow instance is feasible, then the near expander is a standard expander. In the case of unweighted directed graphs, past work~\cite{DBLP:conf/soda/HuaKGW23} designed a flow instance which certifies out-expansion of the subgraph. By considering the flow instance on $G[A]$ and on the same graph with edges reversed, it suffices to certify the feasibility of two flow instances in order to certify expansion in directed graphs.

To define the flow instance, they use a witness $(W,\Pi)$ that $A$ is $\phi$-near expanding in $G$. Informally, let $W$ be a graph on $A$ which is near expanding and $\Pi$ be an embedding of the edges in $W$ into flow paths in $G$ with low congestion. The flow instance is defined as follows: for each edge flow path $P$ in the embedding $\Pi$ between endpoints $u$ and $v$, if $P$ crosses the boundary of $A$, we add 100 source to the endpoints $u$ and $v$. We then add $\deg_G(v)$ sink to each vertex $v\in A$ and set the capacity of each edge to be $200/\phi$. They show that in the setting of directed expander pruning, a similar flow instance certifies $\Omega(\phi/\text{poly}\log{n})$-out expansion (implicitly in Lemma 3.1 of~\cite{DBLP:conf/soda/HuaKGW23}, and also informally in the introduction). We show that for our static expander decomposition, this flow instance certifies $\Omega(\phi)$-out expansion (\cref{sec:flow}).

\paragraph{Trimming.}
Now that we have our two flow instances, for which finding a feasible flow certifies expansion, our goal is to iteratively remove vertices, starting with $A_0=A$, until the flow instance is feasible on $A_t$. In undirected graphs, \cite{saranurak2019expander} showed that assuming an exact max flow oracle, one can find an expander in one step of trimming. Specifically, if the flow is infeasible for $A$, they take the cut $S$ certifying infeasibility of the flow instance and define $A'=A-S$. They show that the flow instance is always feasible on $A'$ by using the structure of the flow instance and the max-flow min-cut theorem. 

Without assuming a max flow oracle, \cite{saranurak2019expander} give an algorithm for trimming using a dynamic flow algorithm based on the Push-Relabel framework~\cite{goldberg1988pushrel}. The idea is to replace the max flow oracle with a bounded-height push-relabel algorithm. When bounded-height push-relabel does not find a feasible flow, they show that there is a sufficiently sparse level cut $S_t$ which can be found efficiently and then update $A_{t+1} = A_t\setminus S_t$. Even though the trimming no longer succeeds in one iteration, so we may need to solve many max flow instances, the instances are highly structured. We can avoid starting from scratch each iteration by reusing information via a dynamic flow algorithm. Specifically, we can use the flow and labels from the flow on $A_t$ to warm start the push-relabel computation for the flow instance on $A_{t+1}$. In total, the dynamic flow algorithm for these specialized sequence of flow instances takes the same amount of time as one call to bounded-height push-relabel.

Trimming becomes significantly more difficult for directed graphs, in part due to the need to certify feasibility for two flow instances instead of one. For example, there is no easy way to design a fast trimming algorithm even if we assume a fast max flow oracle. Suppose we follow the approach of \cite{saranurak2019expander}. Let $S_0$ be the cut certifying infeasibility for the first flow instance and set $A_1=A-S_0$. We can show that the first flow instance is feasible on $A_1$, again using the max-flow min-cut theorem and the structure of the flow instance. But we no longer have the property that $S_0$ is the minimum cut in the second flow instance, so we may not have that the second flow instance is feasible on $A_1$. Thus, we will need to iterate: let $S_1$ be the cut certifying infeasibility for the second flow instance and set $A_2=A_1-S_1$. We would now be able to prove that the second flow instance is feasible $A_2$, but there are no guarantees for the first flow instance anymore. It is unclear if this process converges fast enough for our desired runtime\footnote{With the correct setting of the parameters, this approach can give an expander decomposition with a $\tilde{O}(2^{\sqrt{\log{n}}})$ loss in the runtime and the number of intercluster edges. This is already novel for weighted graphs but is too slow for us.}.

If we instead take the approach based on dynamic flows, we run into a different issue, also caused by the fact that there are two flow instances instead of one. After we find that one of the flow instance is infeasible, we cut out a set $S_t$ from $A_t$ in both flow instances. The difficulty arises now because in the other flow instance, there may have been a flow path $P$ which started in $S_t$ and ended in $A_{t+1}$. When we warm start the flow instance on $A_{t+1}$, we restrict the flow from $A_t$ to $A_{t+1}$. If we look at the flow path $P$, the source where the flow originated was deleted but in the warm start of $A_{t+1}$, the end of the flow path is still there. This means that $P$, restricted to $A_{t+1}$, could be sending flow out of a vertex in $A_{t+1}$ which does not have any excess flow to send out. This induces \textit{negative excess} on the start of the next flow instance, which standard push-relabel can not handle.

\paragraph{Push-Pull-Relabel.}

Towards dealing with the negative excess, \cite{SP24} generalize the dynamic push relabel algorithm used in \cite{saranurak2019expander} to a new dynamic flow algorithm, which they call Push-Pull-Relabel. Specifically, they introduce a new subroutine \textsc{PullRelabel}, which is analogous to \textsc{PushRelabel} subroutine from \cite{goldberg1988pushrel} but now moves around negative excess instead of positive excess by sending flow around. By using both \textsc{PushRelabel} and \textsc{PullRelabel}, they are able to dynamically solve the flow problems in a similar way to~\cite{saranurak2019expander} in (unweighted) directed graphs.

Unfortunately, the analysis of Push-Pull-Relabel in~\cite{SP24} gives weak guarantees for capacitated graphs. Specifically, the runtime given in their analysis scales with the total weight of all the edges instead of with the number of edges. In \Cref{sec:ppl}, we give an alternate analysis of a modified Push-Pull-Relabel algorithm. Our analysis adopts the lens of the classic Push-Relabel algorithm rather than Unit Flow. Informally, the final bound in the runtime is 
\begin{align}\label{eq:regularized-demand}
    \sum_{v\in V}\frac{\Bn_T(v)}{\nabla(v)}\cdot\widetilde{\deg}_G(v),    
\end{align}
where $\Bn_T$ is some vector satisfying $\Bn_T(V)\le \deg_G(V)$, $\nabla(v)=\deg_G(v)$ is the sink at $v$ for the flow instance described above, and $\widetilde{\deg}_G(v)$ is the unweighted degree of vertex $v$. Informally, the vector $\Bn_T$ measures how the negative excess spreads in the graph via the \textsc{PullRelabel} operations, but the only property we will need in this discussion is that $\Bn_T(V)\le\deg_G(V)$.

Since our new running time bound depends inversely on the weighted degree of each node, if the weighted degrees were uniform, we would get a better running time bound. Specifically, if the weight of each edge were the same, then we would have $\deg_G(v)=\deg_G(V)\cdot\widetilde{\deg}_G(v)/2m$. Plugging this into \Cref{eq:regularized-demand}, this would give a runtime of $O(m)$ since $\Bn_T(V)\le\deg_G(V)$. But in general, the weighted degrees could be arbitrarily non-uniform: it is possible that $\deg_G(v)$ is a constant for some $v$ where $\Bn_T(v)=\Omega(\deg_G(V))$. In this case, our runtime bound would match the naive generalization  of~\cite{SP24}. 

To counteract this, we define a regularized vertex weighting which is closer to uniform:
\begin{align*}
\mathbf{d}(v)=\deg_G(v)+\frac{\widetilde{\deg}_G(v)}{2m}\cdot \deg_G(V).
\end{align*}
We will compute an expander decomposition with respect to this vertex weighting instead (see \cref{sec: prelims} for the formal definition of this). Observe that in this vertex weighting, $\bd(v)\ge\deg_G(v)$, so a $\phi$-expander with respect to the vertex weighting is stronger than a standard $\phi$-expander (with respect to the vertex weighting $\deg_G$). Furthermore, we have that $\bd(V)=2\deg_G(V)$, so the number of intercluster edges in this stronger expander decomposition is also not too much larger.
Most importantly, when certifying expansion with respect to vertex weighting $\bd(v)$, we set the sink $\nabla(v)=\bd(v)$ in the flow instance. By our choice of $\bd(v)$, we have $\nabla(v)\ge \deg_G(V)\cdot\widetilde{\deg}_G(v)/2m$, which is the same expression as we had if the edges had uniform capacities. Plugging this into the \Cref{eq:regularized-demand}, we can bound the runtime by $O(m)$ as before.

\paragraph{The Flow Decomposition Barrier.}

Unfortunately, there is another complication we have thus far glossed over. Naively defining the flow instance from~\cite{DBLP:conf/soda/HuaKGW23} requires explicit maintenance of a path decomposition of the witness embedding. This is because we need to find the set of embedding paths which cross the boundary of $A$, and $A$ is updated throughout the trimming algorithm. In an unweighted graph, this is possible since the total space required for this is $O(m)$, as each edge can only be used once in the flow embedding. But, for a weighted graph, the flow decomposition barrier implies that $\Omega(mn)$ space is required for even storing the explicit path decomposition of a flow. This is much too slow, and is a significant barrier to computing directed expander decompositions for weighted graphs. 

To overcome this barrier we introduce a few new ideas. First, although we cannot maintain an explicit flow path decomposition for the embedding paths of the witness, we show that, at any given state of the algorithm, we can find the set of embedding paths crossing the boundary of $A$ in $\tilde{O}(m)$ time. The approach is simple: when computing the path decompositions in cut-matching using link-cut trees, additionally record a transcript of the sequence of operations made. When we need to determine embedding paths that cross the boundary of $A$, we will first ``tag'' edges crossing the boundary of $A$ using an auxiliary weight function. Then, we merely re-compute the same path decomposition originally computed during cut-matching, with the same sequence of operations. The only difference is that, immediately prior to adding a path to the path decomposition, we use the new edge tags to efficiently determine whether any edge in the path crosses the boundary.

While this algorithm for finding embedding paths crossing the boundary of $A$ is much more efficient than storing explicit flow path decompositions, it is too slow to run more than a polylogarithmic number of total times. To address this issue, we need to ensure that we only need to increment sources a polylogarithmic total of times during the trimming algorithm. To this end, we do our increase source operations in batches and ensure that the total source in the flow instances decreases by a constant factor between source increases. A key ingredient in bounding the number of rounds of trimming is the idea of incrementing the sinks of vertices in our flow instance, a technique applied to great effect in recent work on parallel expander decompositions~\cite{chen2025parallel}. Intuitively, increasing the sink of vertices ensures that residual source remaining between rounds has a limited influence on the later rounds. Hence, roughly speaking, it suffices to show that the total source decreases sharply on the first round of the algorithm. These ideas are detailed in~\cref{sec: fast-path-decomp}. 

%% file: diagrams/graft-diagram.tex
\tikzset{every picture/.style={line width=0.75pt}} 

\begin{tikzpicture}[x=0.75pt,y=0.75pt,yscale=-0.75,xscale=0.75]

\draw  [fill={rgb, 255:red, 232; green, 249; blue, 210 }  ,fill opacity=1 ][line width=1.5]  (134.3,218.92) .. controls (134.3,110.72) and (222.02,23) .. (330.23,23) .. controls (438.43,23) and (526.15,110.72) .. (526.15,218.92) .. controls (526.15,327.13) and (438.43,414.84) .. (330.23,414.84) .. controls (222.02,414.84) and (134.3,327.13) .. (134.3,218.92) -- cycle ;
\draw  [draw opacity=0][fill={rgb, 255:red, 248; green, 235; blue, 236 }  ,fill opacity=1 ][line width=1.5]  (171.71,103.22) .. controls (206.4,104.09) and (252.24,158.98) .. (279.43,235.18) .. controls (307,312.45) and (305.86,384.9) .. (278.47,406.44) -- (215.38,258.03) -- cycle ; \draw  [line width=1.5]  (171.71,103.22) .. controls (206.4,104.09) and (252.24,158.98) .. (279.43,235.18) .. controls (307,312.45) and (305.86,384.9) .. (278.47,406.44) ;  
\draw [fill={rgb, 255:red, 248; green, 235; blue, 236 }  ,fill opacity=1 ][line width=1.5]    (171.71,103.22) .. controls (74,257) and (192,393) .. (278.47,406.44) ;
\draw  [draw opacity=0][fill={rgb, 255:red, 230; green, 230; blue, 230 }  ,fill opacity=1 ][line width=1.5]  (158.76,125.17) .. controls (179.03,133.83) and (195.89,172.66) .. (198.35,220.04) .. controls (200.98,270.71) and (186.22,313.36) .. (164.7,320.29) -- (154.38,222.32) -- cycle ; \draw  [line width=1.5]  (158.76,125.17) .. controls (179.03,133.83) and (195.89,172.66) .. (198.35,220.04) .. controls (200.98,270.71) and (186.22,313.36) .. (164.7,320.29) ;  
\draw    (199.88,136.47) -- (229.21,109.93) ;
\draw [shift={(230.7,108.59)}, rotate = 137.86] [color={rgb, 255:red, 0; green, 0; blue, 0 }  ][line width=0.75]    (10.93,-3.29) .. controls (6.95,-1.4) and (3.31,-0.3) .. (0,0) .. controls (3.31,0.3) and (6.95,1.4) .. (10.93,3.29)   ;
\draw    (218.96,154.08) -- (248.29,127.54) ;
\draw [shift={(249.77,126.2)}, rotate = 137.86] [color={rgb, 255:red, 0; green, 0; blue, 0 }  ][line width=0.75]    (10.93,-3.29) .. controls (6.95,-1.4) and (3.31,-0.3) .. (0,0) .. controls (3.31,0.3) and (6.95,1.4) .. (10.93,3.29)   ;
\draw    (235.1,174.69) -- (264.44,148.15) ;
\draw [shift={(265.92,146.81)}, rotate = 137.86] [color={rgb, 255:red, 0; green, 0; blue, 0 }  ][line width=0.75]    (10.93,-3.29) .. controls (6.95,-1.4) and (3.31,-0.3) .. (0,0) .. controls (3.31,0.3) and (6.95,1.4) .. (10.93,3.29)   ;
\draw    (276.19,267.09) -- (308.11,253.21) ;
\draw [shift={(309.95,252.41)}, rotate = 156.5] [color={rgb, 255:red, 0; green, 0; blue, 0 }  ][line width=0.75]    (10.93,-3.29) .. controls (6.95,-1.4) and (3.31,-0.3) .. (0,0) .. controls (3.31,0.3) and (6.95,1.4) .. (10.93,3.29)   ;
\draw    (282,284.7) -- (313.8,277.78) ;
\draw [shift={(315.75,277.36)}, rotate = 167.74] [color={rgb, 255:red, 0; green, 0; blue, 0 }  ][line width=0.75]    (10.93,-3.29) .. controls (6.95,-1.4) and (3.31,-0.3) .. (0,0) .. controls (3.31,0.3) and (6.95,1.4) .. (10.93,3.29)   ;
\draw    (288.87,306.71) -- (320.63,303.95) ;
\draw [shift={(322.62,303.77)}, rotate = 175.03] [color={rgb, 255:red, 0; green, 0; blue, 0 }  ][line width=0.75]    (10.93,-3.29) .. controls (6.95,-1.4) and (3.31,-0.3) .. (0,0) .. controls (3.31,0.3) and (6.95,1.4) .. (10.93,3.29)   ;
\draw    (284.93,368.35) -- (315.27,375.24) ;
\draw [shift={(317.22,375.69)}, rotate = 192.8] [color={rgb, 255:red, 0; green, 0; blue, 0 }  ][line width=0.75]    (10.93,-3.29) .. controls (6.95,-1.4) and (3.31,-0.3) .. (0,0) .. controls (3.31,0.3) and (6.95,1.4) .. (10.93,3.29)   ;
\draw    (311.41,353.67) -- (282.57,349.55) ;
\draw [shift={(280.59,349.27)}, rotate = 8.13] [color={rgb, 255:red, 0; green, 0; blue, 0 }  ][line width=0.75]    (10.93,-3.29) .. controls (6.95,-1.4) and (3.31,-0.3) .. (0,0) .. controls (3.31,0.3) and (6.95,1.4) .. (10.93,3.29)   ;
\draw    (312.88,331.66) -- (282.59,331.66) ;
\draw [shift={(280.59,331.66)}, rotate = 360] [color={rgb, 255:red, 0; green, 0; blue, 0 }  ][line width=0.75]    (10.93,-3.29) .. controls (6.95,-1.4) and (3.31,-0.3) .. (0,0) .. controls (3.31,0.3) and (6.95,1.4) .. (10.93,3.29)   ;
\draw    (289.4,221.59) -- (257.53,232.67) ;
\draw [shift={(255.65,233.33)}, rotate = 340.82] [color={rgb, 255:red, 0; green, 0; blue, 0 }  ][line width=0.75]    (10.93,-3.29) .. controls (6.95,-1.4) and (3.31,-0.3) .. (0,0) .. controls (3.31,0.3) and (6.95,1.4) .. (10.93,3.29)   ;
\draw    (280.59,198.11) -- (251.58,211.93) ;
\draw [shift={(249.77,212.78)}, rotate = 334.54] [color={rgb, 255:red, 0; green, 0; blue, 0 }  ][line width=0.75]    (10.93,-3.29) .. controls (6.95,-1.4) and (3.31,-0.3) .. (0,0) .. controls (3.31,0.3) and (6.95,1.4) .. (10.93,3.29)   ;
\draw    (271.79,179.03) -- (242.78,192.85) ;
\draw [shift={(240.97,193.71)}, rotate = 334.54] [color={rgb, 255:red, 0; green, 0; blue, 0 }  ][line width=0.75]    (10.93,-3.29) .. controls (6.95,-1.4) and (3.31,-0.3) .. (0,0) .. controls (3.31,0.3) and (6.95,1.4) .. (10.93,3.29)   ;
\draw [fill={rgb, 255:red, 230; green, 230; blue, 230 }  ,fill opacity=1 ][line width=1.5]    (158.76,125.17) .. controls (129.43,176.1) and (123.56,258.28) .. (164.7,320.29) ;
\draw [color={rgb, 255:red, 245; green, 166; blue, 35 }  ,draw opacity=1 ][line width=1.5]    (210.5,174.49) .. controls (215,162) and (216,158) .. (219.38,153.92) .. controls (223,149) and (234.9,139.29) .. (249.42,126.63) .. controls (269.42,109.2) and (292.32,90.21) .. (313.35,82.27) ;
\draw [shift={(317,81)}, rotate = 162.35] [fill={rgb, 255:red, 245; green, 166; blue, 35 }  ,fill opacity=1 ][line width=0.08]  [draw opacity=0] (11.61,-5.58) -- (0,0) -- (11.61,5.58) -- cycle    ;
\draw [color={rgb, 255:red, 245; green, 166; blue, 35 }  ,draw opacity=1 ][line width=1.5]    (230.5,214.49) .. controls (233.51,264.56) and (251.53,276.83) .. (276.27,266.92) .. controls (301,257) and (286,263) .. (334,241) .. controls (381.28,219.33) and (329.6,173.4) .. (399.71,189.23) ;
\draw [shift={(403,190)}, rotate = 193.5] [fill={rgb, 255:red, 245; green, 166; blue, 35 }  ,fill opacity=1 ][line width=0.08]  [draw opacity=0] (11.61,-5.58) -- (0,0) -- (11.61,5.58) -- cycle    ;
\draw  [fill={rgb, 255:red, 65; green, 117; blue, 5 }  ,fill opacity=1 ] (313.49,81) .. controls (313.49,79.06) and (315.06,77.49) .. (317,77.49) .. controls (318.94,77.49) and (320.51,79.06) .. (320.51,81) .. controls (320.51,82.94) and (318.94,84.51) .. (317,84.51) .. controls (315.06,84.51) and (313.49,82.94) .. (313.49,81) -- cycle ;
\draw  [fill={rgb, 255:red, 65; green, 117; blue, 5 }  ,fill opacity=1 ] (399.49,190) .. controls (399.49,188.06) and (401.06,186.49) .. (403,186.49) .. controls (404.94,186.49) and (406.51,188.06) .. (406.51,190) .. controls (406.51,191.94) and (404.94,193.51) .. (403,193.51) .. controls (401.06,193.51) and (399.49,191.94) .. (399.49,190) -- cycle ;
\draw [color={rgb, 255:red, 245; green, 166; blue, 35 }  ,draw opacity=1 ][line width=1.5]    (255.5,297.49) .. controls (262.09,345.8) and (275.18,366.04) .. (284.93,368.35) .. controls (294.69,370.66) and (300.91,373.05) .. (316.81,375.65) .. controls (332.16,378.17) and (369.28,338.44) .. (378.15,306.9) ;
\draw [shift={(379,303.5)}, rotate = 102.34] [fill={rgb, 255:red, 245; green, 166; blue, 35 }  ,fill opacity=1 ][line width=0.08]  [draw opacity=0] (11.61,-5.58) -- (0,0) -- (11.61,5.58) -- cycle    ;
\draw  [fill={rgb, 255:red, 65; green, 117; blue, 5 }  ,fill opacity=1 ] (375.49,303.5) .. controls (375.49,301.56) and (377.06,299.99) .. (379,299.99) .. controls (380.94,299.99) and (382.51,301.56) .. (382.51,303.5) .. controls (382.51,305.44) and (380.94,307.01) .. (379,307.01) .. controls (377.06,307.01) and (375.49,305.44) .. (375.49,303.5) -- cycle ;
\draw [color={rgb, 255:red, 189; green, 16; blue, 224 }  ,draw opacity=1 ][line width=1.5]    (210.5,174.49) .. controls (212.93,199.61) and (223.09,201.24) .. (241.21,193.5) .. controls (259.34,185.75) and (261.02,184.35) .. (272.41,178.86) .. controls (283.47,173.54) and (319.73,158.91) .. (333.76,123.36) ;
\draw [shift={(335,120)}, rotate = 108.89] [fill={rgb, 255:red, 189; green, 16; blue, 224 }  ,fill opacity=1 ][line width=0.08]  [draw opacity=0] (11.61,-5.58) -- (0,0) -- (11.61,5.58) -- cycle    ;
\draw  [fill={rgb, 255:red, 65; green, 117; blue, 5 }  ,fill opacity=1 ] (331.49,120) .. controls (331.49,118.06) and (333.06,116.49) .. (335,116.49) .. controls (336.94,116.49) and (338.51,118.06) .. (338.51,120) .. controls (338.51,121.94) and (336.94,123.51) .. (335,123.51) .. controls (333.06,123.51) and (331.49,121.94) .. (331.49,120) -- cycle ;
\draw [color={rgb, 255:red, 189; green, 16; blue, 224 }  ,draw opacity=1 ][line width=1.5]    (230.5,214.49) .. controls (235.29,215.83) and (241.21,216.65) .. (251.1,212.32) .. controls (261,208) and (268.06,204.96) .. (280.58,198.05) .. controls (291.84,191.83) and (307.47,199.2) .. (319.23,202.19) ;
\draw [shift={(323,203)}, rotate = 189.46] [fill={rgb, 255:red, 189; green, 16; blue, 224 }  ,fill opacity=1 ][line width=0.08]  [draw opacity=0] (11.61,-5.58) -- (0,0) -- (11.61,5.58) -- cycle    ;
\draw  [fill={rgb, 255:red, 65; green, 117; blue, 5 }  ,fill opacity=1 ] (319.49,203) .. controls (319.49,201.06) and (321.06,199.49) .. (323,199.49) .. controls (324.94,199.49) and (326.51,201.06) .. (326.51,203) .. controls (326.51,204.94) and (324.94,206.51) .. (323,206.51) .. controls (321.06,206.51) and (319.49,204.94) .. (319.49,203) -- cycle ;
\draw [color={rgb, 255:red, 144; green, 19; blue, 254 }  ,draw opacity=1 ][line width=1.5]    (255.5,297.49) .. controls (262.81,317.47) and (264.57,331.1) .. (280.78,331.55) .. controls (297,332) and (296,332) .. (312.04,331.24) .. controls (327.43,330.5) and (340.88,314.37) .. (344.59,285.65) ;
\draw [shift={(345,282)}, rotate = 95.53] [fill={rgb, 255:red, 144; green, 19; blue, 254 }  ,fill opacity=1 ][line width=0.08]  [draw opacity=0] (11.61,-5.58) -- (0,0) -- (11.61,5.58) -- cycle    ;
\draw  [fill={rgb, 255:red, 65; green, 117; blue, 5 }  ,fill opacity=1 ] (341.49,282) .. controls (341.49,280.06) and (343.06,278.49) .. (345,278.49) .. controls (346.94,278.49) and (348.51,280.06) .. (348.51,282) .. controls (348.51,283.94) and (346.94,285.51) .. (345,285.51) .. controls (343.06,285.51) and (341.49,283.94) .. (341.49,282) -- cycle ;
\draw  [fill={rgb, 255:red, 208; green, 2; blue, 27 }  ,fill opacity=1 ] (206.99,174.49) .. controls (206.99,172.55) and (208.56,170.98) .. (210.5,170.98) .. controls (212.44,170.98) and (214.01,172.55) .. (214.01,174.49) .. controls (214.01,176.43) and (212.44,178) .. (210.5,178) .. controls (208.56,178) and (206.99,176.43) .. (206.99,174.49) -- cycle ;
\draw  [fill={rgb, 255:red, 208; green, 2; blue, 27 }  ,fill opacity=1 ] (226.99,214.49) .. controls (226.99,212.55) and (228.56,210.98) .. (230.5,210.98) .. controls (232.44,210.98) and (234.01,212.55) .. (234.01,214.49) .. controls (234.01,216.43) and (232.44,218) .. (230.5,218) .. controls (228.56,218) and (226.99,216.43) .. (226.99,214.49) -- cycle ;
\draw  [fill={rgb, 255:red, 208; green, 2; blue, 27 }  ,fill opacity=1 ] (251.99,297.49) .. controls (251.99,295.55) and (253.56,293.98) .. (255.5,293.98) .. controls (257.44,293.98) and (259.01,295.55) .. (259.01,297.49) .. controls (259.01,299.43) and (257.44,301) .. (255.5,301) .. controls (253.56,301) and (251.99,299.43) .. (251.99,297.49) -- cycle ;

\draw (194,320) node [anchor=north west][inner sep=0.75pt]  [font=\Large] [align=left] {$\displaystyle D$};
\draw (418,67) node [anchor=north west][inner sep=0.75pt]  [font=\Large] [align=left] {$\displaystyle A$};

\end{tikzpicture}

%% file: organization.tex
\section{Organization of the Paper}
In~\cref{sec: cut-matching}, we introduce and analyze our nonstop cut-matching game for weighted, directed graphs.~\cref{sec: weak exp} shows how to apply our nonstop cut-matching game to efficiently construct a weak-expander decomposition of weighted, directed graphs. 
In~\cref{sec: strong exp}, we show how to use our cut-matching game framework to efficiently find strong expander decompositions by trimming near-expanders into expanders. In particular, in~\cref{sec:ppl} we strengthen the push-pull relabel framework from~\cite{SP24}, extending it to weighted graphs.

%% file: preliminaries.tex
\section{Preliminaries} \label{sec: prelims}
\paragraph{Sets.} We use $\N$ to denote the set of positive integers and $\Z_{\geq 0}$ to denote the set of non-negative integers. For $k \in \N$, we use $[k]$ to denote $\{1,2, \ldots, k\}$. For a vector $\bv$, we define $\supp(\bv)$ to be its set of nonzero coordinates. For a sequence of sets $X_1, X_2, \ldots, X_k$, we use $X_{<j}$ to denote $\bigcup_{i = 1}^{j-1} X_i$. We similarly define $X_{\leq j}$, $X_{> j}$, and $X_{\geq j}$. When $S \subseteq X$, we use $\overline{S}$ to denote its complement; that is, $\overline{S} = X \setminus S$.

\paragraph{Functions.} For two functions $f, g : \mathcal{X} \to \R$ we sometimes write $f \leq g$ to denote that, for all $x \in \mathcal{X}$, $f(x) \leq g(x)$. We also use this notation for vectors, interpreting those vectors as functions. For $S \subseteq \mathcal{X}$ and $S$ finite, we also write $f(S)$ as shorthand denoting $\sum_{x \in S} f(x)$. 

\paragraph{Graphs.} We consider directed, capacitated (weighted) graphs $G = (V, E, \bc)$ where $\bc \in \N^E$ and $\bc \in [1,W]$. Unless otherwise specified, we use $n$ to denote the order of $G$ and $m$ to denote its size. Sometimes we write $V_G$ (or $V(G)$), $E_G$ (or $E(G)$), and $\bc_G$ to clarify that they are the parameters of the graph $G$. For $S \subseteq V$, denote the induced subgraph of $G$ as $G[S]$. That is, $G[S]$ is the subgraph of $G$ formed by retaining exactly the vertices in $S$ and edges between vertices in $S$.

For $S, T \subseteq V$, we write $E(S, T)$ to denote the set of edges from a vertex in $S$ to a vertex in $T$. We will almost always be interested in the weighted degree of vertices in $G$. So, for $u \in V$, we denote the weighted degree of $u$ in $G$ as $\deg_G(u) = \sum_{e \in E(u, V\setminus\{u\}) \cup E(V\setminus\{u\}, u)} \bc(e)$ and the unweighted degree $\tdeg_G(u) := |E(u, V \setminus \{u\}) \cup E(V\setminus\{u\}, u)|$. For $S \subseteq V$, $ \vol_G(S) := \deg_G(S)$ and $\tvol_G(S) := \tdeg_G(S)$. For $S, T \subseteq V$, we use $\delta(S,T)$ as shorthand denoting $\bc(E(S,T))$ and $\delta_G(S)$ to denote $\bc(E(S, V \setminus S))$. We also often consider vertex weights $\bd \in \Z_{\geq 0}^V$, with the most common weight function being $\bd = \deg_G$. 

\paragraph{Flow.}
A \textit{demand} is a vector $\bb \in \R^V$ whose entries sum to $0$ (i.e., such that $\bb(V) = 0$). A flow $\Bf$ routes a demand $\bb$ if, for each $v \in V$, the net flow at $v$ in $\Bf$ is $\bb(v)$. We say that $\Bf$ has \textit{congestion} $\kappa$ if the flow through any edge in $\Bf$ is at most $\kappa$ times its capacity. Given a flow $\Bf$, a  \textit{path decomposition} of $\Bf$ is a collection of directed, capacitated paths in $G$ such that, for each $(u,v) \in E$, the flow from $u$ to $v$ in $\Bf$ is the total capacity of paths containing the edge from $u$ to $v$ in the path decomposition. 

\paragraph{Pre-flows.}
A \textit{pre-flow} is a relaxation of the notion of a flow, where we need not obey sink capacity constraints on nodes. We use $\ex_{\Bf}(u), \ex_{\Bf}(u)$ to denote the amount of \textit{excess} and \textit{absorbed} flow at node $u$ in the pre-flow $\Bf$. More precisely, using $\Bf(u)$ to denote the net flow out of $u$ (flow out minus flow in), $\abs_{\Bf}(u) = \min(\nabla(u), \Bf(u))$ and $\ex_{\Bf}(u) = \Bf(u) - \abs_{\Bf}(u)$.

\paragraph{Notions of Expansion.} Let $G = (V,E, \bc)$ be a directed, capacitated graph, and let $\bd \in \R_{\geq 0}^V$ be a vertex weighting. Let $S \subseteq V$. Then, the \textit{conductance of $S$ in $G$ with respect to $\bd$} is
\[
\Phi_{G, \bd}(S) = \frac{\min(\delta_G(S, V \setminus S), \delta_G(V \setminus S, S))}{\min(\bd(S), \bd(V \setminus S))}.
\]
For $G$ an undirected graph and $\bd = \deg_G$, we recover the familiar notion of conductance for undirected graphs. We say that a cut $S$ is \textit{$\phi$-sparse} (in $G$ with respect to $\bd$) if $\Phi_{G, \bd}(S) \leq \phi$. We say that $G$ is a $(\phi, \bd)$-\textit{expander} if, for all $S \subseteq V$, we have $\Phi_{G, \bd}(S) \geq \phi$. For $A \subseteq V$, we say that $A$ is \textit{$\phi$-nearly $\bd$-expanding} in $G$ (or $A$ is a $(\phi, \bd)$-\textit{near expander} in $G$) if, for all $S \subseteq A$, we have 
\[
\frac{\min(\delta_G(S, V \setminus S), \delta_G(V \setminus S, S))}{\min(\bd(S), \bd(A \setminus S))} \geq \phi.
\]
Note that if $A$ is $\phi$-nearly $\bd$-expanding in $G$, then the same holds for all $A' \subseteq A$, since the denominators of the relevant expressions only decrease. We say that $G$ is a $(\phi, \bd)$-\textit{out expander} if for all $S \subseteq V$, we have 
\[
\frac{\delta_G(S, \overline{S})}{\bd(S)} \geq \phi.
\]
We can similarly define in-expanders by switching $S$ and $\overline{S}$ in the numerator of the above expression. When $\bd = \deg_G$, we say $G$ is a $\phi$-expander (respectively, near expander, out expander, or in expander). 

We can also define expansion with respect to (single-commodity) flows. We say that a vertex weighting $\bd \in \R_{\geq 0}^V$ \textit{mixes in} $G$ \textit{with congestion} $\kappa$ if, for all demands $\bb \in \R^V$ with $|\bb| \leq \bd$, we have that $\bb$ is routable in $G$ with congestion at most $\kappa$. In particular, this means that $G$ has no $1/\kappa$-sparse cuts with respect to $\bd$ and is hence a $(1/\kappa, \bd)$-expander.\footnote{This is because otherwise there exists $S \subseteq V$ such that $\delta_G(S, \overline{S}) \leq 1/\kappa \cdot  \min(\bd(S), \bd(\overline{S}))$. If $\bd(S) = \min(\bd(S), \bd(\overline{S}))$, then any demand setting $\bb = \bd$ on $S$ will not be routable with congestion less than $1/\kappa$.} Indeed, $\bd$ mixes in $G$ with congestion $1/\phi$ if and only if $G$ is a $(\phi, \bd)$-expander. We remark that while $\bd|_A$ mixing in $G$ with congestion $\kappa$ similarly implies that $A$ is a $(\phi, \bd)$-near expander in $G$, the converse does not hold in general. Flow-based expansion is stronger than cut-based expansion for near-expanders.

%% file: cut-matching.tex
\section{A Cut-Matching Game} \label{sec: cut-matching}
In this section we give a (non-stop) cut-matching game for directed, vertex-weighted capacitated graphs. This is the key subroutine in our expander decomposition algorithms.~\cref{subsec: flow update} describes how to update the implicit flow matrix while supporting partial vertex deletions.~\cref{subsec: cut player} and~\cref{subsec: matching player} detail the new cut and matching players, respectively.~\cref{sec: cvg analysis} gives a potential function analysis showing that the cut-matching game certifies near expansion after $O(\log^2 n)$ rounds. Finally,~\cref{subsec: grafting} describes our post-processing grafting step which reintegrates nodes deleted outside of sparse cuts.

The following is the main result of this section.
\begin{theorem}[Cut-Matching] \label{thm: cut-matching}
Consider a (directed, capacitated) graph $G = (V, E, \bc)$ of order $n$ and size $m$ with integral edge capacities $\bc(e)$ bounded by $W$, a vertex weight function $\bd \in \N_{\geq 0}^V$, and a parameter $\phi \in (0,1)$. Then, there exists a parameter $T = O(\log n \log nW)$ and a randomized, Monte Carlo algorithm that outputs:
\begin{itemize}
    \item A sequence of cuts $C_1, \ldots, C_k \subseteq V$ such that, for each $j \in [k]$, $\Phi_{G[V \setminus C_{< j}], \bd}(C_j) \leq 3\phi$, $\bd(C_{j}) \leq 2\bd(V)/3$, and $C_i \cap C_j = \emptyset$ for $i \neq j \in [k]$.
\end{itemize}
 We also have either:
    \begin{enumerate}
        \item \textbf{Early termination:} $\bd(C_{\leq k}) > \bd(V)/10^4$. 
        \item \textbf{Finds large near-expander:}  $\bd|_{V \setminus C_{\leq k}}$ mixes in $G$ with congestion $44T/\phi$ with high probability.
    \end{enumerate}
For general vertex weights $\bd$, the algorithm runs in $O(T \cdot F(n, m) + T^2 m)$ time, where $F(n, m)$ is the runtime of solving a max-flow instance of order $n$ and size $m$. For $\bd \geq \deg_G/\kappa$ for $\kappa \geq 1$, the algorithm runs in $O(\kappa  m  T \log (n) \log (nW)/ \phi)$ time. 
\end{theorem}

\begin{remark}
Note that, unlike in Saranurak-Wang~\cite{saranurak2019expander}, our early termination case does not depend on $T$. In their work, the dependence on $T$ is necessary for the expander trimming step of their expander decomposition algorithm. Since we apply~\cref{thm: cut-matching} to compute weak expander decompositions and hence will not need to trim near-expanders, we can afford to remove the dependence on $T$. 
\end{remark}
\begin{remark}\label{rem:early-termination}
The early termination condition can be set to any $\tau \geq 10^4$ so that our guarantees are either:
  \begin{enumerate}
        \item \textbf{Early termination:} $\bd(C_{\leq k}) > \bd(V)/\tau$. 
        \item \textbf{Finds large near-expander:}  $\bd|_{V \setminus C_{\leq k}}$ mixes in $G$ with congestion $44T/\phi$ with high probability.
    \end{enumerate}
Doing so does not incur any loss in the run-time. This is useful for our application to (strong) expander decompositions.
\end{remark}
\begin{remark}
The difference in the running times for general demands $\bd$ and $\bd = \deg_G$ is natural: indeed, when $\bd$ is only supported on only two vertices, showing that it mixes amounts to solving an $s$-$t$ max flow instance. In the case of $\bd = \deg_G$ we can apply fast approximate maximum flow algorithms such as bounded height push-relabel. In the case of uncapacitated graphs and $\bd = \deg_G$, we can even use a related flow variant, Unit Flow~\cite{henzinger2020local}, to remove an additional $O(\log n)$ factor from the running time (see~\cref{rmk: unit flow}). 
\end{remark}

Our setup combines elements of Appendix B of \cite{saranurak2019expander}, Appendix A of \cite{li2025congestion}, and~\cite{louis2010directed}. However, since \cite{saranurak2019expander, li2025congestion} consider only undirected graphs and \cite{louis2010directed} does not consider the \textit{non-stop} cut-matching game, we will have to make significant modifications throughout. As in \cite{li2025congestion}, we avoid working with the subdivision graph of \cite{racke2014computing, saranurak2019expander}. Let $A := \supp(\bd)$ and, as define an \textit{$A$-commodity flow} as a multicommodity flow where each $v \in A$ is a source of $\bd(v)$ of its unique flow commodity. 

Solely for the purposes of analysis, for each $v \in A$, we define $v_\circ$ and $v_\times$ representing the active and deleted portions of $v$, respectively. Let $A^\circ := \{ v_\circ \,:\, v \in A\}$ and similarly define $A^\times := \{ v_\times \,:\, v \in A\}$. Then define $\bar{A} := A^\circ \cup A^\times$. We will consider \textit{flow matrices} $\bbF \in \R_{\geq 0}^{\bar{A} \times \bar{A}}$. By merging the entries of $v_\circ$ and $v_\times$ for each $v \in A$, these naturally induce a flow matrix $\bF \in \R_{\geq 0}^{A \times A}$ with entries $\bF(u,v)=\bbF(u_\circ, v_\circ) + \bbF(u_\circ, v_\times) + \bbF(u_\times, v_\circ) + \bbF(u_\times, v_\times)$. We say that a flow matrix $\bF \in \R_{\geq 0}^{A \times A}$ is \textit{routable with congestion $c$} if there exists an $A$-commodity flow $\Bf$ on $G$ such that, for each $u,v \in A$, $\Bf$ simultaneously routes $\bF(u,v)$ of the commodity of $u$ to $v$ and no edge $e$ has more than $c \cdot \bc(e)$ flow passing through it. We say that $\bbF \in \R_{\geq 0}^{\bar{A} \times \bar{A}}$ is routable with congestion $c$ if the induced flow matrix $\bF \in \R_{\geq 0}^{A \times A}$ is routable with congestion $c$. 

We initialize $\bbF_0 \in \R_{\geq 0}^{\bar{A} \times \bar{A}}$ with $\bbF_0(v_\circ, v_\circ) = \bd(v)$ for each $v_\circ \in A^\circ$ and each other entry $0$. Additionally, we initialize $\bd_0 = \bd$, $A_0 = \supp(\bd_0)$, $k_0 = 0$, and $D_0 = \emptyset$. We will proceed in $T$ rounds. For each $t \in [T]$, $\bbF_t \in \R_{\geq 0}^{\bar{A} \times \bar{A}}$ induces an implicit $A$-commodity flow with congestion at most $t/\phi$. The vertex weighting $\bd_t$ represents the fractions of each vertex that remain active, and $A_t = \supp(\bd_t) \subseteq A$ is the set of vertices with some active portion.  $k_t \le t$ is the number of sparse cuts found by the end of the $t$\textsuperscript{th} iteration. $D_t \subseteq A $ is the set of vertices that have been completely deleted but do not belong to some discovered sparse cut. That is, $\supp(\bd_t) = A_t = A \setminus (D_t \cup C_{\leq k_t})$. 

In each cut step, we will compute a bipartition of $A_{t-1} = L_{t} \sqcup R_{t}$ such that $\bd_{t-1}(L_{t}) = \bd_{t-1}(R_{t})$.\footnote{This may require fractionally assigning some vertex to both sides, but this detail will not be important for our analysis.} In each matching step, we solve two flow problems, corresponding to matching across the cut (bipartition) in both ways. The flow networks considered will be a scaled version of $G[A \setminus C_{\leq k_{t-1}}]$. In particular, note that although the nodes in $D_{t-1}$ are not sources or sinks, they are part of the network we attempt to route flows in. In each flow problem, we will route a partial matching, and find a (possibly empty) cut, and partially delete some additional vertices. After both flows, we will partially delete some additional vertices to ensure that the active portions of each vertex sent and received the same amount of flow. In particular, for $u \in A_t$, $\bd_t(u)$ is the minimum of the amount of flow sent into and out of $u$ in the two flow problems in the $t$\textsuperscript{th} matching step.

At the end of each round, for each $u \in A_t$ with $\bd_t(u) < \bd(u)/2$, we will retroactively delete it entirely, adding $u$ to $D_t$. Then, if $\bd_t(A_t) < 0.99\bd(A)$, we will trigger our early termination condition, outputting the sequence of sparse cuts computed thus far. To ensure that $\bd(C_{\leq k_t})$ is sufficiently large in this case, we will relate $\bd(C_{\leq k_t})$ and $\bd(D_t)$. 

In the case that we never reach the early termination condition, we will certify that $\bd_T$ mixes in $G$ with congestion $T/\phi$. Since we ensure that $\bd_T(u) \geq \bd(u)/2$ for all $u \in A_T$, this implies that $\bd|_{V \setminus (D_T \cup C_{\leq k_T})}$ mixes in $G$ with congestion $2T/\phi$. Finally, we incorporate the vertices in $D_T$ either back into $A_T$ or delete them in sparse cuts by solving two final flow problems. Reincorporating the vertices in $D_T$ incurs at most $2/\phi$ additional congestion. The full algorithm is described in~\cref{alg:cut-matching}. The while loop is the cut-matching game, and the post-processing is the grafting step.

\begin{algorithm}[t]
	\caption{Cut-Matching}
	\label{alg:cut-matching}
	\fbox{
		\parbox{0.97\columnwidth}{
		  \textsc{Cut-Matching}$(G, \bd_0, \phi)$ \\
            \tab $A \leftarrow \supp(\bd)$, $D \leftarrow \emptyset$, $\mathcal{C} \leftarrow \emptyset$, $\bd \leftarrow \bd_0$, $t \leftarrow 0$, $T = O(\log (n) \log (nW))$ \\
            \tab Implicitly set $\bbF$ so that $\bbF(v_\circ, v_\circ) = \bd(v)$ for $v_\circ \in A^\circ$ and is $0$ otherwise\\
            \tab \textbf{While} $\bd(A) \geq 0.99 \bd_0(V)$ and $t \leq T$:\\
            \tab\tab $t \leftarrow t + 1$ \\
            \tab\tab Find a random $\bd$-weighted bipartition $L \sqcup R = A$ (\cref{subsec: cut player}) \\
            \tab\tab\tab Find the bipartition using the rows of $\bbF$ if $t \leq T/2$ and the columns of $\bbF$ otherwise \\
             \tab\tab Attempt to compute $\bd$-weighted matchings from $L$ to $R$  and $R$ to $L$ (\cref{subsec: matching player}) \\
             \tab\tab Update $\bd$ with the new matchings and set $A \leftarrow \supp(\bd)$ (\cref{subsec: flow update}) \\
             \tab\tab Implicitly update $\bbF$ with the new matching (\cref{subsec: flow update})\\
             \tab\tab Append sparse cuts $S_1, S_2$ found when computing the matchings to $\mathcal{C}$ (\cref{rmk: sparse cut issue}) \\
             \tab\tab Add vertices deleted outside of sparse cuts to $D$, including those with $\bd(v) < \bd_0(v)/2$\\
            \tab If $\bd(A) < 0.99 \bd_0(V)$: \\
            \tab\tab \textbf{Return} $\mathcal{C}$ as the sequence of cuts; this is \textbf{early termination} \\
            \tab Attempt to compute $\bd$-weighted matching from $D$ into $A$ (\cref{subsec: grafting}) \\
            \tab Append the sparse cut $S$ found when computing the matching to $\mathcal{C}$; remove $S$ from $A$ and $D$ \\
            \tab If $\bd(A) < 0.99 \bd_0(V)$: \\
             \tab\tab \textbf{Return} $\mathcal{C}$ as the sequence of cuts; this is \textbf{early termination} \\
            \tab Attempt to compute $\bd$-weighted matching from $D$ into $A$ in the \textit{reversed} graph (\cref{subsec: grafting}) \\
            \tab Append the sparse cut $S$ found when computing the matching to $\mathcal{C}$; remove $S$ from $A$ and $D$ \\
            \tab If $\bd(A) < 0.99 \bd_0(V)$: \\
             \tab\tab \textbf{Return} $\mathcal{C}$ as the sequence of cuts; this is \textbf{early termination} \\
             \tab \textbf{Return} $\mathcal{C}$ as the sequence of cuts; this is the \textbf{finds large near-expander} case
	}}
\end{algorithm}

\subsection{Updating the Flow Matrix} \label{subsec: flow update}
Each flow problem solved in the matching step induces some partial matching from $L_{t-1}$ to $R_{t-1}$ and vice versa. We distribute the in-flow and out-flow for each $v \in A_{t-1}$ between $v_\circ$ and $v_\times$ to ensure that active portions of vertices have the same total in-flow and out-flow. In particular, the matching step results in a \textit{matching matrix} $\bbM_t : \bar{A}_{t-1} \times \bar{A}_{t-1} \to \R_{\geq 0}$, where $\bbM_t(a,b)$ is the amount of flow routed from $b$ to $a$ in the matching step. The matching step also sets new values $\bd_t(v) \le \bd_{t-1}(v)$ for all $v\in A_{t-1}$.
\begin{remark} \label{rmk: backwards?}
Although the definition of $\bbM_t(a,b)$ might appear backwards, we define it this way because then $b$ is able to send flow to vertices that $a$ has already sent flow to (since there is a directed path from $b$ to $a$). Then, some fraction of $\bF_{t-1}(a)$ can moved to $\bF_{t-1}(b)$. 
\end{remark}
Given $\bbM_t$, we update the flow matrix $\bbF_t$ as follows. As a first step, for all $u \in \bar{A}$ and $v \in A_{t-1}$, define a temporary $\bbF_{t-1}'$ such that
\[
\bbF_{t - 1}'(u, v_{\circ}) = \frac{\bd_t(v)}{\bd_{t-1}(v)}\bbF_{t-1}(u, v_\circ), \quad \bbF_{t - 1}'(u, v_{\times}) = \left(1 - \frac{\bd_t(v)}{\bd_{t-1}(v)} \right)\bbF_{t-1}(u, v_\circ)+ \bbF_{t-1}(u, v_\times).
\]
Intuitively, this is just retroactively redistributing flow to $v \in A_{t-1}$ in the induced flow $\bF_{t-1}$ between $v_\circ$ and $v_\times$. The purpose of this redistribution is to maintain that the sum of incoming and outgoing flows from $v_\circ \in A_t^\circ$ is $\bd_t(v)$. For $v \not\in A_{t-1}$, define $\bbF_{t-1}'(u,v_\circ) = \bbF_{t-1}(u,v_\circ)$ and $\bbF_{t-1}'(u,v_\times) = \bbF_{t-1}(u,v_\times)$. Next, for all $u \in A_{t-1}$, we define
\begin{align*}
    \bbF_t(u_\circ) &= \frac{\bd_t(u)}{\bd_{t-1}(u)} \bbF'_{t-1}(u_\circ) + \sum_{v \in \bar{A}_{t-1}} \frac{\bbM_t(v, u_\circ)}{2} \cdot \frac{\bbF_{t-1}'(v_\circ)}{\bd_{t-1}(v)} - \sum_{v \in \bar{A}_{t-1}} \frac{\bbM_t(u_\circ, v)}{2} \cdot \frac{\bbF_{t-1}'(u_\circ)}{\bd_{t-1}(u)}, \\
    \bbF_t(u_\times) 
    &= \bbF_{t-1}'(u_\times) + \left(1 - \frac{\bd_t(u)}{\bd_{t-1}(u)}\right)\bbF_{t-1}'(u_\circ)  + \sum_{v \in \bar{A}_{t-1}} \frac{\bbM_t(v, u_\times)}{2} \cdot \frac{\bbF_{t-1}'(v_\circ)}{\bd_{t-1}(v)} - \sum_{v \in \bar{A}_{t-1}} \frac{\bbM_t(u_\times, v)}{2} \cdot \frac{\bbF_{t-1}'(u_\circ)}{\bd_{t-1}(u)}.
\end{align*}
For $u \not\in \bar{A}_{t-1}$, they will not be involved as a source or sink in the matching step, so we set $\bbF_t(u) = \bbF_{t-1}'(u)$. The first term in the definition of $\bbF_t(u_\circ)$ and second term in the definition of $\bbF_t(u_\times)$ amount to redistributing the flow previously sent from $u_\circ$. The first terms involving $\bbM_t$ represent each $u \in \bar{A}_{t-1}$ receiving a $\frac{\bbM_t(v, u)}{2 \bd_{t-1}(v)}$ fraction of the redistributed active flow vector $\bbF_{t-1}'(v_\circ)$ from $v$. The latter terms involving $\bbM_t$ represent each $u \in \bar{A}_{t-1}$ sending a $\frac{\bbM_{t}(u, v)}{2\bd_{t-1}(u)}$ fraction of the redistributed 
active flow vector $\bbF_{t-1}'(u_\circ)$ to $v$.

\begin{remark} \label{rmk: matching step prop}
 Our matching step is designed so that, for each $u_\circ \in A_{t}^\circ$, we have
\[
\sum_{v \in \bar{A}_{t-1}} \bbM_t(v, u_\circ) = \sum_{v \in \bar{A}_{t-1}} \bbM_t(u_\circ, v) = \bd_t(u), \quad \sum_{v \in \bar{A}_{t-1}}\bbM_t(v, u_\times) \leq \bd_{t-1}(u) - \bd_t(u).
\]
This fact is crucial in our convergence analysis. 
\end{remark}

\begin{remark} \label{rmk: alt recur uo}
Applying~\cref{rmk: matching step prop} to the second summation in the recursive expression for $\bbF_t(u_\circ)$ for $u_\circ \in A_{t-1}^\circ$, we get
\[
\bbF_t(u_\circ) = \frac{\bd_t(u)}{2\bd_{t-1}(u)} \bbF'_{t-1}(u_\circ) + \sum_{v \in \bar{A}_{t-1}} \frac{\bbM_t(v, u_\circ)}{2} \cdot \frac{\bbF_{t-1}'(v_\circ)}{\bd_{t-1}(v)}.
\]
\end{remark}

\begin{remark} \label{rmk: other ord flow upd}
We could equivalently define the flow update by first scaling the rows of $\bbF_{t-1}$. For ease of notation, define $\cbF_{t-1} = \bbF_{t-1}^\top$. Then, as before, for all $u \in \bar{A}$ and $v \in A_{t-1}$, define
\[
\cbF_{t-1}'(u, v_\circ) := \frac{\bd_t(v)}{\bd_{t-1}(v)}\cbF_{t-1}(u, v_\circ), \quad \quad \cbF_{t-1}'(u, v_\times) := \left(1 - \frac{\bd_t(v)}{\bd_{t-1}(v)}\right) \cbF_{{t-1}}(u, v_\circ) + \cbF_{t-1}(u, v_\times).
\]
For $v \not\in A_{t-1}$, define $\cbF_{t-1}'(u, v_\circ) = \cbF_{t-1}(u, v_\circ)$ and $\cbF_{t-1}'(u, v_\times) = \cbF_{t-1}(u, v_\times)$. Finally, for all $u \in A_{t-1}$, we then can define
\begin{align*}
    \cbF_t(u_\circ) &= \frac{\bd_t(u)}{\bd_{t-1}(u)} \cbF'_{t-1}(u_\circ) + \sum_{v \in \bar{A}_{t-1}} \frac{\bbM_t(u_\circ, v)}{2} \cdot \frac{\cbF_{t-1}'(v_\circ)}{\bd_{t-1}(v)} - \sum_{v \in \bar{A}_{t-1}} \frac{\bbM_t(v, u_\circ)}{2} \cdot \frac{\cbF_{t-1}'(u_\circ)}{\bd_{t-1}(u)}, \\
    \cbF_t(u_\times) 
    &= \cbF_{t-1}'(u_\times) + \left(1 - \frac{\bd_t(u)}{\bd_{t-1}(u)}\right)\cbF_{t-1}'(u_\circ)  + \sum_{v \in \bar{A}_{t-1}} \frac{\bbM_t(u_\times, v)}{2} \cdot \frac{\cbF_{t-1}'(v_\circ)}{\bd_{t-1}(v)} - \sum_{v \in \bar{A}_{t-1}} \frac{\bbM_t(v, u_\times)}{2} \cdot \frac{\cbF_{t-1}'(u_\circ)}{\bd_{t-1}(u)}.
\end{align*}
For $u \not\in \bar{A}_{t-1}$, we set $\cbF_t(u) = \cbF_{t-1}'(u)$. Most importantly for our later convergence analysis in~\cref{sec: cvg analysis}, for $u_\circ \in A_{t-1}^\circ$ we get 
\[
\cbF_t(u_\circ) = \frac{\bd_t(u)}{2 \bd_{t-1}(u)} \cbF_{t-1}'(u_\circ) + \sum_{v \in \bar{A}_{t-1}} \frac{\bbM_t(u_\circ, v)}{2} \cdot \frac{\bbF_{t-1}'(v_\circ)}{\bd_{t-1}(v)},
\]
similarly to in~\cref{rmk: alt recur uo}.
\end{remark}
It will be important in our convergence analysis that the sum of flows in and out of $u_\circ \in A_t^\circ$ is preserved. We can derive this from the recursive formulation of $\bbF_t$ along with~\cref{rmk: matching step prop}.

\begin{claim} \label{cla: sum of flow}
For all $t \leq T$ and $u_\circ \in A_t^\circ$, we have the following:
\[
\sum_{w \in \bar{A}} \bbF'_{t-1}(u_\circ, w) = \bd_{t-1}(u), \quad \sum_{w \in \bar{A}} \bbF_t(u_\circ, w) = \bd_t(u) = \sum_{w \in \bar{A}} \bbF_t(w, u_\circ),
\]
where the left equality only applies when $t>0$.
\end{claim}
\begin{proof}

We prove both parts of the claim by induction. The base case of $t = 0$ is trivial by the instantiation of the algorithm. Fix $u_\circ \in A_t^\circ$. We have, using the definition of $\bbF'_{t-1}$ and the inductive hypothesis, 
\begin{align*}
\sum_{w \in \bar{A}} \bbF'_{t-1}(u_\circ, w) &= \sum_{w \in A} (\bbF'_{t-1}(u_\circ, w_\circ) + \bbF'_{t-1}(u_\circ, w_\times)) \\
&= \sum_{w \in A} (\bbF_{t-1}(u_\circ, w_\circ)+ \bbF_{t-1}(u_\circ, w_\times)) \\
&= \sum_{w \in \bar{A}} \bbF_{t-1}(u_\circ, w) \\
&= \bd_{t-1}(u).
\end{align*}
Using~\cref{rmk: matching step prop} and~\cref{rmk: alt recur uo},
\begin{align*}
\sum_{w \in \bar{A}} \bbF_t(u_\circ, w) &= \frac{\bd_t(u)}{2\bd_{t-1}(u)} \sum_{w \in \bar{A}} \bbF'_{t-1}(u_\circ, w) + \sum_{v \in \bar{A}_{t-1}}\frac{\bbM_t(v, u_\circ)}{2} \sum_{w \in \bar{A}}  \frac{\bbF_{t-1}'(v_\circ, w)}{\bd_{t-1}(v)} \\
&= \frac{\bd_t(u)}{2\bd_{t-1}(u)} \cdot \bd_{t-1}(u) + \sum_{v \in \bar{A}_{t-1}}\frac{\bbM_t(v, u_\circ)}{2} \cdot 1  \\
&= \bd_t(u)/2 + \bd_t(u)/2 = \bd_t(u).
\end{align*}
Next, we have
\begin{align*}
\sum_{w \in \bar{A}} \bbF_t(w, u_\circ) &= \sum_{w_\circ \in A^\circ} \bbF_t(w_\circ, u_\circ) + \sum_{w_\times \in A^\times} \bbF_t(w_\times, u_\circ) \\
&= \sum_{w \in \bar{A}} \bbF_{t-1}'(w, u_\circ) +  \sum_{w \in \bar{A}} \sum_{v \in \bar{A}_{t-1}} \frac{\bbM_t(v, w)}{2} \cdot \frac{\bbF_{t-1}'(v_\circ, u_\circ)}{\bd_{t-1}(v)} -  \sum_{w \in \bar{A}} \sum_{v \in \bar{A}_{t-1}} \frac{\bbM_t(w, v)}{2} \cdot \frac{\bbF_{t-1}'(w_\circ, u_\circ)}{\bd_{t-1}(w)} \\
&= \bd_t(u) + \sum_{v \in \bar{A}_{t-1}} \sum_{w \in \bar{A}_{t-1}} \frac{\bbM_t(v, w)}{2} \cdot \frac{\bbF_{t-1}'(v_\circ, u_\circ)}{\bd_{t-1}(v)} -  \sum_{w \in \bar{A}_{t-1}} \sum_{v \in \bar{A}_{t-1}} \frac{\bbM_t(w, v)}{2} \cdot \frac{\bbF_{t-1}'(w_\circ, u_\circ)}{\bd_{t-1}(w)} \\
&= \bd_t(u) + \sum_{v \in \bar{A}_{t-1}} \sum_{w \in \bar{A}_{t-1}} \frac{\bbM_t(v, w)}{2} \cdot \frac{\bbF_{t-1}'(v_\circ, u_\circ)}{\bd_{t-1}(v)} -  \sum_{v \in \bar{A}_{t-1}} \sum_{w \in \bar{A}_{t-1}} \frac{\bbM_t(v, w)}{2} \cdot \frac{\bbF_{t-1}'(v_\circ, u_\circ)}{\bd_{t-1}(v)} \\
&= \bd_t(u).
\end{align*}
In the second step, we apply the recursive formulation of $\bbF_t$. In the third step, we use that $\sum_{w \in \bar{A}} \bbF_{t-1}'(w, u_\circ) = \bd_t(u)$ for all $u_\circ \in A_{t-1}^\circ$ by definition of $\bbF_{t-1}'$ and the inductive hypothesis. We also use that we can restrict to $w \in \bar{A}_{t-1}$ in the latter two sums due to the $\bbM_t$ terms. Finally, we reindex in the fourth step to cancel the summation terms.
\end{proof}
Moreover, while the flows in and out of  $u_\times \in A^\times_t$ are not exactly preserved, we can still bound them in a manner similar to \cref{cla: sum of flow}. We will need this additional property in our convergence analysis.
\begin{claim} \label{cla: sum of flow deleted}
For all $t \leq T$ and $u \in A$, we have the following.
\begin{enumerate}
    \item $\sum_{w \in \bar{A}} \bbF_t(u_\times, w) \leq 3(\bd(u)- \bd_t(u))/2.$
    \item $\sum_{w \in \bar{A}} \bbF_t(w, u_\times) = \bd(u) - \bd_t(u).$
\end{enumerate}
\end{claim}
\begin{proof}

As before, we prove both parts of the claim by induction. The base case of $t = 0$ is again immediate by the instantiation of the algorithm. Fix $u_\times \in A^\times$. We have, using the recursive definition of $\bbF_t$, the inductive hypothesis, and \cref{cla: sum of flow},
{
\allowdisplaybreaks
\begin{align*}
\sum_{w \in \bar{A}} \bbF_t(u_\times, w) &= \sum_{w \in \bar{A}} \bbF_{t-1}'(u_\times, w)  + \frac{\bd_{t-1}(u) - \bd_t(u)}{\bd_{t-1}(u)} \sum_{w \in \bar{A}} \bbF'_{t-1}(u_\circ, w) + \sum_{v \in \bar{A}_{t-1}}\frac{\bbM_t(v, u_\times)}{2} \sum_{w \in \bar{A}}  \frac{\bbF_{t-1}'(v_\circ, w)}{\bd_{t-1}(v)} \\
& \quad - \sum_{v \in \bar{A}_{t-1}} \frac{\bbM_t(u_\times, v)}{2} \sum_{w \in \bar{A}} \frac{\bbF_{t-1}'(u_\circ, w)}{\bd_{t-1}(u)} \\
&\leq 3(\bd(u) - \bd_{t-1}(u))/2 + \bd_{t-1}(u) - \bd_t(u)  + \sum_{v \in \bar{A}_{t-1}}\frac{\bbM_t(v, u_\times)}{2} \cdot 1 - \sum_{v \in \bar{A}_{t-1}} \frac{\bbM_t(u_\times, v)}{2} \cdot 1 \\
&\leq 3(\bd(u) - \bd_{t-1}(u))/2 + \bd_{t-1}(u) - \bd_t(u) + (\bd_{t-1}(u) - \bd_t(u))/2 \\
&= 3(\bd(u) - \bd_{t}(u))/2.
\end{align*}
}
The first inequality follows from applying the inductive hypothesis and \cref{cla: sum of flow}. Then, the second inequality uses that $ \sum_{v \in \bar{A}_{t-1}}\bbM_t(v, u_\times) \leq \bd_{t-1}(u) - \bd_t(u)$ by \cref{rmk: matching step prop}.
Next, we have
{
\allowdisplaybreaks
\begin{align*}
\sum_{w \in \bar{A}} \bbF_t(w, u_\times) &= \sum_{w_\circ \in A^\circ} \bbF_t(w_\circ, u_\times) + \sum_{w_\times \in A^\times} \bbF_t(w_\times, u_\times) \\
&= \sum_{w \in \bar{A}} \bbF_{t-1}'(w, u_\times) +\sum_{w \in \bar{A}}  \sum_{v \in \bar{A}_{t-1}}  \frac{\bbM_t(v, w)}{2} \cdot \frac{\bbF_{t-1}'(v_\circ, u_\times)}{\bd_{t-1}(v)} \\
& \quad -  \sum_{w \in \bar{A}} \sum_{v \in \bar{A}_{t-1}} \frac{\bbM_t(w, v)}{2} \cdot \frac{\bbF_{t-1}'(w_\circ, u_\times)}{\bd_{t-1}(w)} \\
&= \sum_{w \in \bar{A}} \left(\left(1 - \frac{\bd_t(u)}{\bd_{t-1}(u)} \right)\bbF_{t-1}(w, u_\circ)+ \bbF_{t-1}(w, u_\times)\right) +\sum_{w \in \bar{A}}  \sum_{v \in \bar{A}_{t-1}}  \frac{\bbM_t(v, w)}{2} \cdot \frac{\bbF_{t-1}'(v_\circ, u_\times)}{\bd_{t-1}(v)} \\
& \quad -  \sum_{w \in \bar{A}} \sum_{v \in \bar{A}_{t-1}} \frac{\bbM_t(w, v)}{2} \cdot \frac{\bbF_{t-1}'(w_\circ, u_\times)}{\bd_{t-1}(w)} \\
&= \left(1 - \frac{\bd_t(u)}{\bd_{t-1}(u)} \right) \bd_{t-1}(u) + (\bd(u) - \bd_{t-1}(u)) + \sum_{v \in \bar{A}_{t-1}} \sum_{w \in \bar{A}_{t-1}} \frac{\bbM_t(v, w)}{2} \cdot \frac{\bbF_{t-1}'(v_\circ, u_\times)}{\bd_{t-1}(v)} \\
& \quad -  \sum_{w \in \bar{A}_{t-1}} \sum_{v \in \bar{A}_{t-1}} \frac{\bbM_t(w, v)}{2} \cdot \frac{\bbF_{t-1}'(w_\circ, u_\times)}{\bd_{t-1}(w)} \\
&= \bd(u) - \bd_t(u) + \sum_{v \in \bar{A}_{t-1}} \sum_{w \in \bar{A}_{t-1}} \frac{\bbM_t(v, w)}{2} \cdot \frac{\bbF_{t-1}'(v_\circ, u_\circ)}{\bd_{t-1}(v)} -  \sum_{v \in \bar{A}_{t-1}} \sum_{w \in \bar{A}_{t-1}} \frac{\bbM_t(v, w)}{2} \cdot \frac{\bbF_{t-1}'(v_\circ, u_\circ)}{\bd_{t-1}(v)} \\
&= \bd(u) - \bd_t(u).
\end{align*}
}
In the second step, we apply the recursive formulation of $\bbF_t$. In the third step, we use the definition of $\bbF_{t-1}'$. In the fourth step, we use the inductive hypothesis and~\cref{cla: sum of flow} in the expansion of $\bbF_{t-1}'$. We also use that we can restrict to $w \in \bar{A}_{t-1}$ in the latter two sums due to the $\bbM_t$ terms. Finally, we reindex in the fourth step to cancel the summation terms.
\end{proof}

\subsection{Cut Player} \label{subsec: cut player}
We will use a similar cut player as~\cite{louis2010directed}. Define $\bF_{t-1}^\circ$ to be the submatrix of $\bbF_{t-1}$, only retaining rows and columns corresponding to $u_\circ \in A_{t-1}^\circ$. Let $\tilde{\bF}_{t-1}^\circ$ be $\bF_{t-1}^\circ$ with the $w_\circ$\textsuperscript{th} column scaled by $1/\sqrt{\bd_{t-1}(w)}$. We can then describe the algorithm used in the cut step in the first $T/2$ rounds. 
\begin{enumerate}
    \item Sample $\br_t$, a random unit vector supported on $A_{t-1}^\circ$.
    \item For each $u \in A_{t-1}$, compute $p_t(u) := \inner{\tilde{\bF}^\circ_{t-1}(u_\circ)/\bd_{t-1}(u)}{\br_t}$. 
    \item Compute bipartition $L_{t} \sqcup R_{t} = A_{t-1}$ and separator value $\eta$, such that $\bd_{t-1}(L_{t-1}) = \bd_{t-1}(R_{t-1})$ and $\max_{u \in L_{t-1}} p_t(u) \leq \eta \leq \min_{u \in R_{t-1}} p_t(u)$. 
\end{enumerate}
\begin{remark}
If necessary, we treat a single node as two nodes for this round and put it partially in both $L_{t-1}$ and $R_{t-1}$ to have a balanced partition. For such a $u$, we will have $p_t(u) = \eta$, so it will turn out not to affect our convergence analysis.    
\end{remark}
\begin{remark}
The scaling of $\bF_{t-1}^\circ$ into  $\tilde{\bF}_{t-1}^\circ$ is relevant in our convergence analysis when we apply a subgaussian concentration lemma and compare to our potential function. 
\end{remark}
\begin{remark} \label{rmk: switching to columns}
For technical reasons (related to generalizing to directed graphs), in the last $T/2$ rounds of cut-matching, we compute $p_t(u)$ analogously, except using the columns on $\bF_{t-1}^\circ$ instead of the rows. In particular, denote $(\bF_{t-1}^\circ)^\top$ by $\cbF_{t-1}^\circ$ and $\cbF_{t-1}^\circ$ with the $w_\circ$\textsuperscript{th} column scaled by $1/\sqrt{\bd_{t-1}(w)}$ by $\tcbF^\circ_{t-1}$. Then, for $t > T/2$, 
\[
p_t(u) := \inner{\tcbF^\circ_{t-1}(u_\circ)/\bd_{t-1}(u)}{\br_t}.
\]
We revisit this technicality in the convergence analysis in~\cref{sec: cvg analysis}.
\end{remark}
We observe that the cut player can be implemented efficiently. 
\begin{lemma} \label{lem: cut player run time}
For each $t \leq T$, we can compute the bipartition $A_{t-1} = L_t \sqcup R_t$ on round $t$ in $O(mt + m \log m)$ time.
\end{lemma}
\begin{proof}
Note that, since $G$ has $m$ edges, any flow on $G$ can be decomposed into $m$ flow paths. Hence, since each $\bbM_s$ is composed of two flows on subgraphs of $G$, each has at most $2m$ nonzero entries. Then, for any vector $\bv$, it takes $O(m)$ time to compute $\bbF_t \bv$ given $\bbF_{t-1} \bv$ by the recursive formulation for $\bbF_t$. As such, padding $\br_t$ with zeroes and scaling the $w_\circ$\textsuperscript{th} entry by $1/\sqrt{\bd_{t-1}(w)}$ for each $w \in A_{t-1}$ to form $\tilde{\br}_t$, we can compute $\bbF_t \tilde{\br}_t$ in $O(mt)$ total time. Then, the entries of $\bbF_t \tilde{\br}_t$ corresponding to $A_{t-1}^\circ$ are precisely the $p_t(u)$ we care about. We can then sort them in $O(m \log m)$ time to compute the required bipartition.
\end{proof}

\begin{remark}
Note that, since we are restricting the support of our random unit vector to $A_{t-1}^\circ$, we have to compute $\bbF_{t} \bv$ from scratch for each cut step. This is in contrast to usual implementations of the cut step for~\cite{khandekar2009graph} based analyses where the computations from prior cut steps can be reused. Nonetheless, this runtime increase (a factor of $T$ for the cut steps) is a lower order term. 
\end{remark}

We will need that there is some set certifying that we make progress on each pair of cut/matching steps. The following technical lemma proves the existence of such a set. This is similar to Lemma 3.3 in \cite{racke2014computing}, except we allow $S$ to be a subset of either part of the partition and the sets are equal-sized. 

\begin{lemma}\label{lem:apd-finding-S}
    Let $X$ be a finite multi-subset of $\R$. Let $(L_{\eta}, R_{\eta})$ be a bisection  of $X$ by some $\eta \in \R$ so that $\max(L_{\eta}) \leq \eta \leq \min(R_{\eta})$. Define $\bar{\mu} = \frac{1}{|X|}\sum_{x \in X} x$. Then, there exists $S \subseteq X$ such that: 
    \begin{enumerate}
        \item $S \subseteq L_{\eta}$ or $S \subseteq R_{\eta}$.
        \item For each $s\in S$, we have $(s-\eta)^2\ge \frac{1}{9}\cdot (s-\bar{\mu})^2$.
        \item $\sum_{s\in S}(s-\bar{\mu})^2\ge \frac{1}{16}\sum_{x\in X}(x-\bar{\mu})^2$.
    \end{enumerate}
\end{lemma}
\begin{remark}
Intuitively either the sum of distances squared to the average is reasonably well-balanced between $L_\eta$ and $R_\eta$ or it is very unbalanced and mostly comes from elements far from both $\eta$ and $\bar{\mu}$.     
\end{remark}

\begin{remark} \label{rmk: use of apd-finding-S}
We will apply \cref{lem:apd-finding-S} with $L_\eta$ as the multiset of $p_t(u)$ for $u \in L_t$, with $p_t(u)$ having multiplicity $\bd_{t-1}(u)$, and $R_\eta$ similarly defined. Then, for $t \leq T/2$, $\bar{\mu} = \inner{\tilde{{\bmu}}_t}{\br_t}$, where $\tilde{{\bmu}}_t$ is just $\bmu_t$  with $w_\circ$\textsuperscript{th} column scaled by $1/\sqrt{\bd_{t-1}(w)}$. (Here $\bmu_t$ is the average of $\bF_{t-1}^\circ(u_\circ)$ over $u_\circ \in A_{t-1}^\circ$. See \cref{sec: cvg analysis}.)
\end{remark}

\begin{proof}[Proof of \cref{lem:apd-finding-S}]
Assume for simplicity that $|X|$ is even. Define $L_{\bar{\mu}}$ and $R_{\bar{\mu}}$ to be the subsets of $X$ less than and at least $\bar{\mu}$, respectively. Note that
\begin{equation*}
\sum_{x \in L_{\bar{\mu}}} \bar{\mu} - x = \sum_{x \in R_{\bar{\mu}}} x - \bar{\mu} =: d.    
\end{equation*}
Suppose that $\bar{\mu} \geq \eta$. 
Now, for the first case, suppose
\begin{equation*}
\sum_{x \in R_{\bar{\mu}}} (x - \bar{\mu})^2 \geq \frac{1}{16} \sum_{x \in X} (x - \bar{\mu})^2.
\end{equation*}
Then, $R_{\bar{\mu}}$ satisfies all the desired properties.
For the second case,  suppose
\begin{equation*}
\sum_{x \in R_{\bar{\mu}}} (x - \bar{\mu})^2 < \frac{1}{16} \sum_{x \in X} (x - \bar{\mu})^2,
\end{equation*}
so that
\begin{equation*}
\sum_{x \in L_{\bar{\mu}}} (x - \bar{\mu})^2 \geq \frac{15}{16} \sum_{x \in X} (x - \bar{\mu})^2.
\end{equation*}
Since $\bar{\mu} \geq \eta$, the elements in $L_{\bar\mu}\setminus L_\eta$ satisfy
\begin{equation*}
\sum_{x\in L_{\bar\mu}\setminus L_\eta}(x-\bar\mu)^2\le\frac{|X|}2(\eta-\bar\mu)^2\le\sum_{x\in L_\eta}(x-\bar\mu)^2.
\end{equation*}
This in turn implies that 
\begin{equation*}
\sum_{x \in L_{\eta}} (x - \bar{\mu})^2 \geq \frac{15}{32} \sum_{x \in X} (x - \bar{\mu})^2.
\end{equation*}
Observe also that we must have 
\begin{equation*} \label{eq: eta bd}
\eta \geq \bar{\mu} - \frac{2d}{|X|},
\end{equation*}
since otherwise 
\begin{equation*}
d = \sum_{x \in L_{\bar{\mu}}} \bar{\mu} - x \geq \sum_{x \in L_{\eta}} \bar{\mu} - x > \frac{2d}{|X|} \cdot \frac{|X|}{2} = d.
\end{equation*}
As a result, for any $x \in X$ with $\bar{\mu} - x = \frac{Cd}{|X|}$ for $C \geq 3$, the second property is satisfied. Indeed,
\begin{align*}
(\eta - x)^2 &\geq \left(\bar{\mu} - \frac{2d}{|X|} -x\right)^2 \\
&= (\bar{\mu} - x)^2 - \frac{4d}{|X|}\cdot (\bar{\mu} - x) + \frac{4d^2}{|X|^2} \\
&= (C^2 -4C + 4) \cdot \frac{d^2}{|X|^2} \\
&\geq \frac{C^2d^2}{9|X|^2}.
\end{align*}
Now, consider the subset of ``intermediate'' elements in $L_{\bar{\mu}}$ which do not satisfy the second property as a direct consequence of the above analysis: 
\begin{equation*}
I := \left\{ x \in X\,:\, 0 < \bar{\mu} - x < \frac{3d}{|X|} \right\}.    
\end{equation*}
Note that
\begin{equation*}
    \sum_{x \in I} (\bar{\mu} - x)^2 \leq \frac{9d^2 |I|}{|X|^2} \leq \frac{9d^2}{|X|}.
\end{equation*}
But, by the Cauchy-Schwarz inequality,
\begin{equation*}
    \sum_{x \in R_{\bar{\mu}}} (\bar{\mu} - x)^2 \geq 
    \frac{\left(\sum_{x \in R_{\bar{\mu}}} (\bar{\mu} - x) \cdot 1\right)^2}{\sum_{x \in R_{\bar{\mu}} }1^2} = \frac{d^2}{|R_{\bar{\mu}}|} \geq \frac{2d^2}{|X|}.
\end{equation*}
So, since $\sum_{x \in I} (x - \bar{\mu})^2$ is at most $4.5 \cdot \sum_{x \in R_{\bar{\mu}}} (x - \bar{\mu})^2$ and $\sum_{x \in R_{\bar{\mu}}} (x - \bar{\mu})^2 < \frac{1}{16} \sum_{x \in X} (x - \bar{\mu})^2$ by assumption, we have
\begin{equation*}
\sum_{x \in L_\eta \cap \{y \,:\ y \leq \bar{\mu} - 3d/|X|\}} (x - \bar{\mu})^2 \geq \frac{5}{32} \sum_{x \in X} (\bar{\mu} - x)^2.
\end{equation*}
Hence, taking the subset of $x \leq \bar{\mu} - 3d/|X|$ in $L_{\eta}$ yields the desired result. The proof in the case of $\bar{\mu} < \eta$ is analogous.
\end{proof}

\subsection{Matching Player} \label{subsec: matching player}

The matching player receives the bisection $A_{t-1} = L_t \sqcup R_t$ computed by the cut player. The matching player tries to embed a matching from $L_t$ to $R_t$ into $G[V\setminus C_{\leq k_{t-1}}]$ and a matching from $R_t$ to $L_t$ into $G[V \setminus C_{\leq k_{t-1}}]$, each with congestion at most $1/\phi$.  When some matching fails to be embedded with low congestion, we find a sparse cut and a partial matching.

To embed a matching from $L_{t}$ to $R_t$ into $G[V \setminus C_{\leq k_{t-1}}]$, consider the flow problem where we set $\Delta_t(u) = \bd_{t-1}(u)$ for each $u \in L_t$ and $\nabla_t(u) = \bd_{t-1}(v)$ for each $v \in R_t$. Each edge $e \in E(G[V \setminus C_{\leq k_{t-1}}])$ has capacity $\bc(e)/\phi$, resulting in the graph $G_t$. We additionally solve the flow problem with $\Delta_t$ and $\nabla_t$ swapped.

Using dynamic trees~\cite{sleator1981data}, we can find path decompositions of each flow into at most $m$ flow paths in $O(m \log nW)$ time. These path decompositions induce a weighted partial matching $\bM_t$ on $A_{t-1}$, where $\bM_t(u,v)$ is the amount of source sent from $v$ to $u$ among the two flows (see \cref{rmk: backwards?} if this seems backwards). We can use these path decompositions to define the weighted partial matching $\bbM_t$ on $\bar{A}_{t-1}$, as well as define $\bd_t(u)$ for $u \in A$. First, for $u \in A$ in the sparse cut(s) resultant from solving the flow problems, we set $\bd_t(u) = 0$ (we ``fully delete $u$''), and add the cut to our sequence of sparse cuts, incrementing $k$. We do this even if $u$ was previously deleted but not in the sequence of sparse cuts. Once deleted, $\bd_{t+1}(u) = \bd_t(u)$ for all subsequent rounds. For  $u \in A$, we set 
\[
\bd_t(u) = \min\left(\sum_{v \in A} \bM_t(u, v), \sum_{v \in A} \bM_t(v, u)\right).
\]
This update can be interpreted as ``partially deleting'' the parts of $u$ that did not both send and receive flow in this step. For technical reasons in our convergence analysis, if this update results in $\bd_t(u) < \bd(u)/2$, we set $\bd_t(u) = 0$ and add $u$ to $D_t$.

It will be important for post-processing to know how deleted vertices were deleted. After describing how to implement our matching step (and find sparse cuts efficiently in the case of failing to match), we will show that $\bd(D_t)$ is not much larger than $\bd(C_{\leq k_t})$.

We can now define $\bbM_t$. 
\begin{enumerate}
    \item If $\bd_t(u) = \sum_{w \in A} \bM_t(u, w)$ and $\bd_t(v) = \sum_{w \in A} \bM_t(w, v)$ we set 
    \[
    \bbM_t(u_\circ, v_\circ) = \bM_t(u,v), \quad \bbM_t(u_\circ, v_\times) = 0, \quad \bbM_t(u_\times, v_\circ) = 0, \quad \bbM_t(u_\times, v_\times) = 0.
    \]
    \item If $\bd_t(u) < \sum_{w \in A} \bM_t(u, w)$ and $\bd_t(v) = \sum_{w \in A} \bM_t(w, v)$, we set 
    \begin{align*}
    &\bbM_t(u_\circ, v_\circ) = \frac{\bd_t(u)}{\bd_{t-1}(u)}\bM_t(u,v),  &&\bbM_t(u_\circ, v_\times) = 0,   \\
   &\bbM_t(u_\times, v_\circ) = \left(1 - \frac{\bd_t(u)}{\bd_{t-1}(u)}\right)\bM_t(u,v), &&\bbM_t(u_\times, v_\times) = 0. 
    \end{align*}
    \item If $\bd_t(u) = \sum_{w \in A} \bM_t(u, w)$ and $\bd_t(v) < \sum_{w \in A} \bM_t(w, v)$, we set 
    \begin{align*}
    &\bbM_t(u_\circ, v_\circ) = \frac{\bd_t(v)}{\bd_{t-1}(v)}\bM_t(u,v),   &&\bbM_t(u_\circ, v_\times) = \left(1 - \frac{\bd_t(v)}{\bd_{t-1}(v)}\right)\bM_t(u,v), \\
    &\bbM_t(u_\times, v_\circ) = 0, &&\bbM_t(u_\times, v_\times) = 0.
    \end{align*} 
    \item If $\bd_t(u) < \sum_{w \in A} \bM_t(u, w)$ and $\bd_t(v) < \sum_{w \in A} \bM_t(w, v)$, we set 
    \begin{align*}
    &\bbM_t(u_\circ, v_\circ) = \frac{\bd_t(v) \bd_t(u) }{\bd_{t-1}(v) \bd_{t-1}(u)}\bM_t(u,v), &&\bbM_t(u_\circ, v_\times) = \frac{(\bd_{t-1}(v) - \bd_t(v)) \bd_t(u) }{\bd_{t-1}(v) \bd_{t-1}(u)}\bM_t(u,v) \\
    &\bbM_t(u_\times, v_\circ) = \frac{\bd_t(v) (\bd_{t-1}(u) - \bd_t(u)) }{\bd_{t-1}(v) \bd_{t-1}(u)}\bM_t(u,v), &&\bbM_t(u_\times, v_\times) = \frac{(\bd_{t-1}(v) -  \bd_t(v))(\bd_{t-1}(u) - \bd_t(u)) }{\bd_{t-1}(v) \bd_{t-1}(u)}\bM_t(u,v).
    \end{align*} 
\end{enumerate} 
We can attach a tagline to each of these cases. In the first case, $u$ sends at least as much flow as it receives, and $v$ receives at least as much as it sends. In the second case, $u$ sends less flow than it receives, and $v$ receives at least as much as it sends. This is flipped in the third case; $u$ sends at least as much flow as it receives, and $v$ receives less flow than it sends. Finally, in the fourth case $u$ sends less flow than it receives, and $v$ receives less flow than it sends. 
\begin{remark}\label{rmk: matching equality}
  Note that $\bM_t(u, v) = \bbM_t(u_\circ, v_\circ) + \bbM_t(u_\circ, v_\times) + \bbM_t(u_\times, v_\circ) + \bbM_t(u_\times, v_\times)$, so we are just dividing the flow from $u$ to $v$ up between their active and deleted portions. Also note that $\bbM_t$ is defined so that for each $u \in A_{t}$, we have 
\[
\sum_{v \in \bar{A}_{t-1}} \bbM_t(v, u_\circ) = \sum_{v \in \bar{A}_{t-1}} \bbM_t(u_\circ, v) = \bd_t(u), \quad \sum_{v \in \bar{A}_{t-1}}\bbM_t(v, u_\times) \leq \bd_{t-1}(u) - \bd_t(u),
\]
as promised by~\cref{rmk: matching step prop}.
\end{remark}

In general, we can solve each flow problem with one call to an exact max flow oracle. But, in the case of $\bd(u)$ corresponding to $\deg_G(u)$, we can instead use bounded height push relabel with height $h = O(\frac{1}{\phi}\log nW)$  to either route the matching or find a sequence of $2.1\phi$-sparse cuts in $O(mh \log n)$ = $O(\frac{m \log n \log nW}{\phi})$ total time. The extra $\log n$ factor comes from a standard application of link-cut trees~\cite{sleator1981data} to push-relabel, as in~\cite{goldberg1988pushrel}. 

Indeed, as we make concrete in \cref{lem: bdd ht push relabel}, we can use bounded height push relabel for the matching step for demands similar to $\deg_G$ as well. In the description of~\cref{lem: bdd ht push relabel} we use additional notation specifically relevant to the Push-Relabel framework and pre-flows (see~\cref{sec: prelims} for the relevant definitions).

\begin{lemma}\label{lem: bdd ht push relabel}
Let $H = (V, E, \bc)$ be a graph with $\bc \in ([1, W] \cap \N)^E$, and let $H_\phi$ be the graph with capacities scaled by $1/\phi$. Let $\bb$ be a vertex weighting on $H$ with $\bb \geq \deg_H/\kappa$, for some $\kappa \geq 1$. Let $\Delta$ and $\nabla$ be source and sink functions on $H_\phi$ with $\Delta, \nabla \leq  \bb$, $\Delta(V) \leq  \nabla(V)$. Let $\Bf$ and $\Bell$ be the pre-flow and vertex levels resultant from termination of push-relabel on flow instance with source $\Delta$ and sink $\nabla$ with maximum height $h = O( \frac{\kappa \log nW}{\phi})$ on $H_\phi$. Then, either
\begin{enumerate}
    \item $\ex_{\Bf}(V) = 0$, and $\Bf$ is a flow is routing $\Delta$ (i.e., no vertices have excess), or
    \item $\ex_{\Bf}(V) > 0$ and there exists $S \subseteq V$ with $\Phi_{\bb}(S) \leq 1.1 \phi$ and $\bb(S) \leq 2\bb(V)/3$.
    Moreover, $S$ can be found in an additional $O(m)$ time, and the flow instance with source nonzero only on $V \setminus S$ is feasible upon removing an additional $2 \cdot \min(\bb(S), \bb(\overline{S}))$ source. 
\end{enumerate}
\end{lemma}
\begin{proof}
The case of $\ex_{\Bf}(V) = 0$ follows from correctness of push-relabel. For each $k \in \N$, we define the \textit{$k$\textsuperscript{th} level set} $L_k := \{u \in V\,:\, \Bell(u) = k\}.$
For the case of $\ex_{\Bf}(V) > 0$, we consider the level cuts $L_{\geq k} := \{u \,:\, \Bell(u) \geq k\}$. 
 We first show that there exists a level $k$ where 
\begin{align*}
\Bf(L_{\geq k}, L_{<k}) \geq \delta_{H_\phi}(L_{\geq k}, L_{<k}) -  \frac{1}{10} \min(\bb(L_{\geq k}), \bb(L_{<k})).
\end{align*}
Let $H_{\text{con}} \subseteq H_\phi$ be the edge subgraph where we retain only edges between consecutive levels. Now, suppose there are $B$ many ``big'' level $k$ cuts, i.e., such that 
\begin{equation}
\delta_{H_\text{con}}(L_{\geq k}) \geq \frac{\phi}{10\kappa } \vol_{H_\text{con}}(L_{\geq k}).  \label{eqtn: sparse lvl}  
\end{equation}
Then, note that we have  
\begin{align*}
\vol_{H_\text{con}}(L_{\geq k - 1})  &\geq \delta_{H_\text{con}}(L_{\geq k}) + \vol_{H_\text{con}}(L_{\geq k})    
\end{align*}
and, trivially, 
\[
\vol_{H_\text{con}}(L_{\geq r}) \geq \vol_{H_\text{con}}(L_{\geq k})
\]
for $r \leq k$. Hence, we get
\[
\left(1 + \frac{\phi}{10\kappa}\right)^B \leq \vol_{H_\text{con}}(V) \leq \vol_{H_\phi}(V) < 2(nW)^2.
\]
Consequently, we have $B \leq \frac{25 \kappa  \log nW}{ \phi }$. In particular, for $h = \frac{100 \kappa \log nW}{ \phi }$, there will be at least $ \frac{75 \kappa \log nW}{2 \phi }$ levels $k$ on which \cref{eqtn: sparse lvl} does not hold. We can similarly show that there are at most $ \frac{25 \kappa \log nW}{ \phi }$ levels such that
\begin{equation*}
\delta_{H_\text{con}}(L_{\geq k}) \geq \frac{\phi }{10\kappa} \vol_{H_\text{con}}(L_{< k}),
\end{equation*}
this time using the identity that 
\[
\vol_{H_\text{con}}(L_{< k + 1})  \geq \delta_{H_\text{con}}(L_{\geq k}) + \vol_{H_\text{con}}(L_{< k}).
\]
Therefore, we get that there exists a level $k > 0$ such that 
\begin{align*}
 \delta_{H_\text{con}}(L_{\geq k}) &< \frac{\phi}{10\kappa} \min(\vol_{H_\text{con}}(L_{< k}),\vol_{H_\text{con}}(L_{\geq k})) \\
 &\leq \frac{\phi}{10\kappa} \min(\vol_{H_\phi}(L_{< k}),\vol_{H_\phi}(L_{\geq k})) \\
 &\leq \frac{1}{10} \min(\bb(L_{< k}),\bb(L_{\geq k})).
\end{align*}
Now, it is an invariant of push-relabel that edges that decrease more than one level from start to end must be saturated, and edges that increase more than one level from start to end carry zero flow. Hence, 
\begin{align*}
\Bf(L_{\geq k}, L_{<k}) &\geq \delta_{H_\phi}(L_{\geq k}, L_{< k}) -  \frac{1}{10 } \min(\bb(L_{\geq k}), \bb(L_{<k})) \\
&\geq \frac{1}{\phi} \delta_{H}(L_{\geq k}, L_{< k}) -  \frac{1}{10 } \min(\bb(L_{\geq k}), \bb(L_{<k})),
\end{align*}
since the only possible non-saturated edges from $L_{\geq k}$ to $L_{<k}$ in $H_\phi$ are the edges accounted for in $\delta_{H_\text{con}}(L_{\geq k})$.

We want to show that $(L_{\geq k}, L_{< k})$ is a sparse cut relative to $\bb$, so we need to upper bound $\Bf(L_{\geq k}, L_{<k})$ by roughly $\min(\bb(L_{\geq k}), \bb(L_{<k}))$. Note that we have
\[
\Bf(L_{\geq k}, L_{<k}) = \Delta(L_{\geq k}) - \abs_{\Bf}(L_{\geq k}) - \ex_{\Bf}(L_{\geq k}) < \Delta(L_{\geq k}) - \nabla(L_{\geq k}),
\]
since, for all $ u \in L_{\geq k}$ we have $\Bell(u) > 0$ so $\abs_{\Bf}(u) = \nabla(u)$. We also have $\ex_{\Bf}(L_{\geq k}) > 0$ since $\ex_{\Bf}(V) > 0$ and all nodes with positive excess have height $h$. We trivially have $ \Delta(L_{\geq k}) \leq \bb(L_{\geq k})$ since $\Delta \leq \bb$. Likewise, 
\[
\Delta(L_{\geq k}) = \Delta(V) - \Delta(L_{< k}) \leq \nabla(V) - \Delta(L_{< k}) \leq \nabla(V) 
\]
using $\nabla(V) \geq \Delta(V)$ and 
\[
\nabla(L_{\geq k}) = \nabla(V) - \nabla(L_{< k}) \geq \nabla(V) - \bb(L_{< k}),
\]
using $\nabla \leq\bb$. Hence, we get
\[
\Bf(L_{\geq k}, L_{<k}) \leq  \min(\bb(L_{\geq k}), \bb(L_{<k})).
\]
This implies that 
\[
\delta_H(L_{\geq k}, L_{< k}) \leq 1.1 \phi \min(\bb(L_{\geq k}), \bb(L_{<k})),
\]
so $\Phi_{\bb}(L_{\geq k}) = \Phi_{\bb}(L_{< k}) \leq 1.1\phi$.

Finally, if $\bb(L_{\geq k}) \leq 2\bb(V)/3$, we output $L_{\geq k}$ as $S$. Otherwise, $\bb(L_{<k}) \leq 2\bb(V)/3$, and we output $L_{< k}$ as $S$. In both cases,  $\Delta|_{\overline{S}}$ is feasible except possibly for flow from sources to vertices with excess. When $S = L_{\geq k}$, at most  $\delta_{H_\text{con}}(L_{\geq k}, L_{< k}) \leq 0.1 \min(\bb(L_{\geq k}), \bb(L_{<k}))$ flow is sent from sources in $\overline{S}$ to vertices with excess.

When $S = L_{< k}$, we know that $\bb(L_{< k}) \leq 2  \bb(V)/3$. Then, since $\nabla \leq \bb$, we know $\nabla(L_{< k}) \leq \bb(L_{< k})$. Moreover, since $\Bell(u)$ only increases when $\abs_{\Bf}(u) = \nabla(u)$, we know that all sink in $L_{\geq k}$ is saturated. Then, since $\Delta(L_{< k}) \leq  \bb(L_{< k})$, all but $\bb(L_{< k})$ is saturated by sources in $L_{\geq k}$. Hence, $\Delta|_{\overline{S}}$ is feasible upon removing at most $2\cdot \min(\bb(L_{\geq k}), \bb(L_{<k})) = 2 \bb(L_{<k})$ additional source, as required.
\end{proof}

\begin{remark}
\cref{lem: bdd ht push relabel} is the analogue of Lemma B.6 of \cite{saranurak2019expander}. Unlike Lemma B.6 of \cite{saranurak2019expander}, we cannot find a restriction of our source function that is feasible by iterating push-relabel using dynamic push-relabel. The issue is that we are bipartitioning $\bd_{t-1}(A_{t-1})$ rather than finding an uneven partition (e.g., see the second property of Lemma 3.3 of~\cite{racke2014computing}, used in \cite{saranurak2019expander, li2025congestion}). We need the extra structure of a bipartition because our matchings are directed. Indeed, otherwise we could not guarantee that each active vertex receives and sends the same amount of flow    (see~\cref{rmk: matching step prop}). As a result, the total source and sink are equal in our flow problems (at least initially) and it's possible that the sparse cut we remove is the side with vertices at level $0$ at termination of bounded height push relabel, which dynamic push-relabel does not allow.
\end{remark}
\begin{remark} \label{rmk: unit flow}
In the case of unweighted graphs, Unit Flow~\cite{henzinger2020local} solves this flow problem in $O(mh)$ time, removing the extra $O(\log n)$ factor in the running time.
\end{remark}

\begin{remark}
In the subsequent analysis involving~\cref{lem: bdd ht push relabel}, we only use the guarantees (1) and (2). Indeed, any flow algorithm with those guarantees could be substituted here instead of bounded-height push-relabel (e.g., these are the only guarantees we are using for an exact max-flow algorithm). 
\end{remark}

In the case $\bd \geq \deg_G /\kappa$, we apply~\cref{lem: bdd ht push relabel} with $H_\phi = G_t$ and $\bb = \bd|_H$. Since $\bd_{t-1} \leq \bd$ , we have $\Delta, \nabla \leq \bb$. By our cut step, we have $\Delta(V) = \nabla(V)$.

In the case of general $\bd$, we use exact max flow, and the mincut $S$ will yield our sparse cut in each flow problem. Indeed, each edge  leaving the cut will be saturated (since it is a mincut), and we get $f(S, V \setminus (C_{\leq k_{t-1}} \cup S)) =  \Delta(S) - \nabla(S)$. By the same proof as in~\cref{lem: bdd ht push relabel}, we then have $f(S, V \setminus (C_{\leq k_{t-1}} \cup S))) \leq \min(\bd(S),  \bd(V \setminus (C_{\leq k_{t-1}} \cup S)))$. But, $\frac{1}{\phi} \delta_{G[V \setminus C_{\leq k_{t-1}}]}(S, V \setminus (C_{\leq k_{t-1}} \cup S)) = f(S, V \setminus (C_{\leq k_{t-1}} \cup S))$, so we get the desired sparsity. Then, we can choose $S$ so that $\bd_{t-1}(S) \leq 2 \bd_{t-1}(V)/3$ and $\Delta|_{V \setminus (C_{\leq k_{t-1}} \cup S)}$ is feasible upon removing at most $2 \bd_{t-1}(S)$ additional source. 

\begin{remark} \label{rmk: sparse cut issue}
  Since both flow instances are on $G_t$, we need to be careful about how to add the sparse cuts we find to our sequence of $C_i$'s. Indeed, if the cuts we find are not disjoint, it may not be the case that the remainder of one is still sparse after removing the other.   
\end{remark}

We address~\cref{rmk: sparse cut issue}. Let $S_1$, $S_2$ be the cuts we find from our flow problems. If adding just one of the two to our sequence of sparse cuts would trigger our earlier termination condition, we do so and terminate. 

Otherwise, let $\bd(S_1) \leq \bd(S_2)$. Then, if $\bd(S_2 \setminus S_1) \geq \bd(S_2)/2$, add $S_1$ and then $S_2 \setminus S_1$ to our sequence of cuts. The sparsity of $S_2 \setminus S_1$ decreases by at most a factor of $2$, so it remains at most $3 \phi$. Finally, if $\bd(S_2 \setminus S_1) <  \bd(S_2)/2$, then $\bd(S_2 \cap S_1) \geq \bd(S_2)/2$. Then we add $S_2 \cap S_1$ to our sequence of sparse cuts (which has sparsity at most $3 \phi$), and add the symmetric difference $S_1 \triangle S_2$ to $D_t$. Note that
\[
\bd(S_1 \triangle S_2) = \bd(S_1 \setminus S_2) + \bd(S_2 \setminus S_1) \le \bd(S_1) +  \bd(S_2) / 2 \le 3 \bd(S_2) / 2 \le 3\bd(S_2 \cap S_1).
\]    

\begin{claim}  \label{cla: D_t bd}
For all $t \leq T$, we have $\bd(D_t) \leq 35\bd(C_{\leq k_t})$. Moreover, $\bd(C_{\leq k_t}) \geq (\bd(V) - \bd_t(V))/36$.
\end{claim}
\begin{proof}
We reduce the active portions of vertices during the matching step in the following ways.
\begin{enumerate}
    \item When computing matchings, we delete some demand to ensure feasibility (\cref{lem: bdd ht push relabel} or using exact max flow). Indeed, upon finding a sparse cut $S$, we need to delete at most $4 \bd(S)$ additional source from the flow instance with source on $V \setminus \left(C_{\leq k_{t-1}} \cup S\right)$ to ensure feasibility.
    \item For $S_1$ and $S_2$ the sparse cuts found in the two matchings, we add at least $\max(\bd(S_1), \bd(S_2))/2$ to $\bd(C_{\leq k_t})$, and delete at most $3\max(\bd(S_1), \bd(S_2))/2$ additional demand. (See the discussion preceding this claim). Combined with the source deleted to ensure feasibility, we delete at most $19\max(\bd(S_1), \bd(S_2))/2$ demand without adding it to  $\bd(C_{\leq k_{t-1}})$ while adding at least $\max(\bd(S_1), \bd(S_2))/2$ demand to $\bd(C_{\leq k_{t-1}})$.
\end{enumerate}
So, in total, the matching step directly contributes at most $19$ times as much to the total deleted demand not in $\bd(C_{\leq k_t})$ as it does to $\bd(C_{\leq k_t})$.

There is then only one way demand can be added to $D_t$: If we ever have $\bd_t(u) < \bd(u)/2$, we set $\bd_t(u) = 0$ and add $u$ to $D_t$. 

In each round $s$,  $\sum_{u \in A} \bd_{s-1}(u) - \bd_s(u)$ is at most the total amount of source removed to ensure feasibility, which is at most $16$ times the contribution to $\bd(C_{\leq k_t})$, by the above. When we have $\bd_t(u) < \bd(u)/2$, then, since $\bd_0(u) = \bd(u)$, we know $u$ must have contributed at least $\bd(u)/2$ to $\sum_{s = 1}^t \sum_{u \in A} \bd_{s-1}(u) - \bd_s(u)$, so the total contribution  to $D_t$ is bounded by $16\bd(C_{\leq k_t})$.
Combining all of the contributions to $\bd(D_t)$ yields the desired result.
\end{proof}

\subsection{Convergence Analysis} \label{sec: cvg analysis}
Next, we define our potential function and lower bound its decrease in each step. 
First, we define the weighted average of
\[
\bmu_t := \sum_{u_\circ \in A_t^\circ} \bF_t^\circ(u_\circ) / \bd_t(A_t).
\]
We can now define our potential function
\[
\psi(t) := \sum_{u_\circ \in A_t^\circ} \bd_t(u) \sum_{v_\circ \in A_t^\circ} \frac{1}{\bd_t(v)}({\bF}^\circ_t(u_\circ, v_\circ)/\bd_t(u) - {\bmu}_t(v_\circ))^2.
\]

The function $\psi(t)$ amounts to taking a weighted sum of the norm squared differences of $\bF^\circ_t(u_\circ)$ and $\bmu_t$, splitting each coordinate corresponding to $u_\circ$ into $\bd_t(u)$ equal coordinates.  We can similarly define $\bmup_t$ and $\psip(t)$ for the columns of $\bF_t^\circ$. (For this, we use the notation $\cbF_t$ to denote $\bbF_t^\top$ and $\cbF_t^\circ$ to denote $[\bbF_t^\circ]^\top$. See~\cref{rmk: other ord flow upd}.) That is,
\[
\bmup_t := \sum_{u_\circ \in A_t^\circ} \cbF_t^\circ (u_\circ) / \bd_t(A_t),
\]
and
\begin{align*}
\psip(t) &:= \sum_{u_\circ \in A_t^\circ} \bd_t(u) \sum_{v_\circ \in A_t^\circ} \frac{1}{\bd_t(v)}({\bF}^\circ_t(v_\circ, u_\circ)/\bd_t(u) - {\bmup}_t(v_\circ))^2 \\
&= \sum_{u_\circ \in A_t^\circ} \bd_t(u) \sum_{v_\circ \in A_t^\circ} \frac{1}{\bd_t(v)}({\cbF}^\circ_t(u_\circ, v_\circ)/\bd_t(u) - {\bmup}_t(v_\circ))^2.    
\end{align*}

\begin{remark} \label{rmk: argmin}
Note that
\[
\bmu_t = \arg \min_{\bx \in \R^{A^\circ}} \sum_{u_\circ \in A_t^\circ} \bd_t(u) \sum_{v_\circ \in A_t^\circ} \frac{1}{\bd_t(v)}(\bF_t^\circ(u_\circ, v_\circ)/\bd_t(u) - \bx(v_\circ))^2.
\]
This can be verified coordinate-wise, setting the derivative equal to $0$. An analogous claim holds for $\bmup_t$ and $\psip$.
\end{remark}

We first show that $\psi(t) + \psip(t)$ being small implies that $\bd_t$ mixes with low congestion in $G$.

\begin{lemma} \label{lem: small pot implies exp}
For any $t \in [T]$, if $\bd_t(A_t) \geq \frac{3}{4}\bd(A)$ and $\psi(t) + \psip(t) \leq 1/(nW)^C$ for large enough constant $C$, then $\bd_t$ mixes in $G$ with congestion $7T/\phi$.
\end{lemma}
\begin{proof}
We begin by showing that a constant fraction of the flow must go between vertices in $A_t^\circ$.
\begin{claim} \label{cla: most flow stays}
$\bF_t^\circ(A_t^\circ, A_t^\circ) \geq \bd(A)/2.$
\end{claim}
\begin{proof}[Proof of claim.]
Using \cref{cla: sum of flow deleted}, we know that 
\[
\bbF_t(A_t^\circ, A^\times) = \bd(A) - \bd_t(A) = \bd(A) - \bd_t(A_t)\leq \bd(A)/4.
\]
On the other hand, from Claim~\ref{cla: sum of flow}, we know
\[
\bbF_t(A_t^\circ, \bar{A}) = \bd_t(A_t) \geq 3\bd(A)/4.
\]
Then, 
\[
\bF_t^\circ(A_t^\circ, A_t^\circ) =  \bbF_t(A_t^\circ, A_t^\circ) = \bbF_t(A_t^\circ, \bar{A}) - \bbF_t(A_t^\circ, A^\times) \geq \bd(A)/2,
\]
as needed.
\end{proof}
Next, we show that the potential being small allows us to approximate each $\bbF_t(u_\circ, v_\circ)$, for $u_\circ, v_\circ \in A_t^\circ$. 
\begin{claim} \label{cla: small pot struct}
For all $u_\circ \in A_t^\circ$, we have 
\[
\left| \sum_{v_\circ \in A_t^\circ} \bF_t^\circ(u_\circ, v_\circ) - \bd_t(u) \cdot \frac{\bF_t^\circ(A_t^\circ, A_t^\circ)}{\bd_t(A_t)}\right| \leq \frac{1}{(nW)^{C/2 - 2}}
\]
In particular, we have
\[
\sum_{v_\circ \in A_t^\circ} \bF_t^\circ(u_\circ, v_\circ) \geq \bd_t(u)/3.
\]
We also have 
\[
\left| \sum_{v_\circ \in A_t^\circ} \bF_t^\circ(v_\circ, u_\circ) - \bd_t(u) \cdot \frac{\bF_t^\circ(A_t^\circ, A_t^\circ)}{\bd_t(A_t)}\right| \leq \frac{1}{(nW)^{C/2 - 2}}
\]
and 
\[
\sum_{v_\circ \in A_t^\circ} \bF_t^\circ( v_\circ, u_\circ) \geq \bd_t(u)/3.
\]
\end{claim}
\begin{proof}[Proof of claim.]
We prove the first two claims; the latter two follow by analogous arguments, using that $\psip(t) \leq 1/(nW)^C$. For all $u_\circ$, $v_\circ \in A_t^\circ$, we have
\begin{align*}
     \frac{1}{\bd_t(v)}\left(\frac{\bF_t^\circ(u_\circ, v_\circ)}{\bd_t(u)} - {\bmu}_t(v_\circ)\right)^2 
 =  \frac{1}{\bd_t(v)}\left(\frac{\bbF_t(u_\circ, v_\circ)}{\bd_t(u)} - {\bmu}_t(v_\circ)\right)^2 &\leq \psi(t) \leq 1/(nW)^C,
\end{align*}
since $\psi(t) \leq 1/(nW)^C$ by assumption. Consequently, we have 
\[
\left| \bF_t^\circ(u_\circ, v_\circ) - \bd_t(u){\bmu}_t(v_\circ) \right| \leq \frac{1}{(nW)^{C/2 - 1}}.
\]
Observe that, by definition of $\bmu_t$, we have
\[
\sum_{v_\circ \in A_t^\circ} \bmu_t(v_\circ) =  \frac{\sum_{v_\circ \in A_t^\circ}\sum_{u_\circ \in A_t^\circ} \bF_t^\circ(u_\circ, v_\circ)}{\bd_t(A_t)} = \frac{\bF_t^\circ(A_t^\circ, A_t^\circ)}{\bd_t(A_t)} \geq 1/2,
\]
with the final inequality from Claim~\ref{cla: most flow stays}.
Finally, 
\begin{align*}
\left| \sum_{v_\circ \in A_t^\circ} \bF_t^\circ(u_\circ, v_\circ) - \bd_t(u) \cdot \frac{\bF_t^\circ(A_t^\circ, A_t^\circ)}{\bd_t(A_t)}\right| &= \left|  \sum_{v_\circ \in A_t^\circ} \bF_t^\circ(u_\circ, v_\circ) - \bd_t(u) {\bmu}_t(v_\circ) \right| \\
&\leq \sum_{v_\circ \in A_t^\circ}  \left| \bF_t^\circ(u_\circ, v_\circ) - \bd_t(u) {\bmu}_t(v_\circ) \right| \\
&\leq \frac{1}{(nW)^{C/2 - 2}}.
\end{align*}
Then, the fact that $\sum_{v_\circ \in A_t^\circ} \bF_t^\circ(u_\circ, v_\circ) \geq \bd_t(u)/3$ follows from the fact that $\bd_t(u) \cdot \frac{\bF_t^\circ(A_t^\circ, A_t^\circ)}{\bd_t(A_t)} \geq \bd_t(u)/2.$ 
\end{proof}
We can now apply \cref{cla: most flow stays} and \cref{cla: small pot struct} to prove \cref{lem: small pot implies exp}. First, decompose the $A$-commodity flow routing $\bF_t^\circ$ with congestion $T/\phi$ into single commodity flows $\Bf_{u,v}$ for $u_\circ, v_\circ \in A_t^\circ$, with $\Bf_{u,v}$ sending $\bF_t^\circ(u_\circ, v_\circ)$ flow from $u$ to $v$. 

Next, let $\bb$ be any vertex demand satisfying $|\bb| \leq \bd_t$ with $\sum_{u \in A_t} \bb(u) = 0$. For each pair $u,v \in A_t$, we route a flow. If $\bb(u) \geq 0$, we define
\[
\Bf_{u,v}' := \frac{\Bf_{u,v}}{\sum_{w_\circ \in A_t^\circ} \bF^\circ_t(u_\circ, w_\circ)} \cdot |\bb(u)|.
\]
Otherwise, if $\bb(u) < 0$, we define
\[
\Bf_{u,v}' := \frac{\Bf_{v,u}}{\sum_{w_\circ \in A_t^\circ} \bF^\circ_t(w_\circ, u_\circ)} \cdot |\bb(u)|.
\]
We then route $\Bf_{u,v}'$ simultaneously for all $u,v \in A_t$.  
Now, each $u \in A_t$ with $\bb(u) \geq 0$ sends $|\bb(u)|$ total demand and each $\Bf_{u,v}'$ has absolute value at most
\[
\frac{\bF^\circ_t(u_\circ, v_\circ)}{\sum_{w_\circ \in A_t^\circ} \bF^\circ_t(u_\circ, w_\circ)} \cdot \bd_t(u) \leq 3\bF_t^\circ(u_\circ, v_\circ),
\]
by the first part of~\cref{cla: small pot struct}. Similarly, for each $u \in A_t$ with $\bb(u) < 0$, $u$ receives $|\bb(u)|$ total demand and each $\Bf_{u,v}'$ has absolute value at most
\[
\frac{\bF^\circ_t(v_\circ, u_\circ)}{\sum_{w_\circ \in A_t^\circ} \bF^\circ_t(w_\circ, u_\circ)} \cdot \bd_t(u) \leq 3\bF_t^\circ(v_\circ, u_\circ).
\]
This bound uses the second part of~\cref{cla: small pot struct}.

This implies that all of these flows $\Bf'_{u,v}$ can be routed simultaneously with congestion at most $6T/\phi$. All that remains is to route the leftover demand. Vertex $v \in A_t$ receives 
\begin{align*}
\frac{\bF_t^\circ(u_\circ, v_\circ) \bb(u)}{\sum_{w_\circ \in A_t^\circ} \bF^\circ_t(u_\circ, w_\circ)}.
\end{align*}
demand from $u \in A_t$ with $\bb(u) \geq 0$ and 
\[
\frac{\bF_t^\circ(v_\circ, u_\circ) \bb(u)}{\sum_{w_\circ \in A_t^\circ} \bF^\circ_t(w_\circ, u_\circ)}
\]
demand from $u \in A_t$ with $\bb(u) < 0$. By~\cref{cla: small pot struct}, $\sum_{w_\circ \in A_t^\circ} \bF^\circ_t(w_\circ, u_\circ)$ and $\sum_{w_\circ \in A_t^\circ} \bF^\circ_t(u_\circ, w_\circ)$ are approximately equal. Moreover, since $\bF_t^\circ(u_\circ, u_\circ)$ approximately equals both $\bd_t(u) \bmu_t(u_\circ)$ and $\bd_t(u) \bmu_t'(u_\circ)$, we get that $\bmu_t'(u_\circ)$ and $\bmu_t(u_\circ)$ for all $u_\circ \in A_t$. In particular, we then get that $\bF_t^\circ(u_\circ, v_\circ)$ is approximately equal to $\bF_t^\circ(v_\circ, u_\circ)$. Hence, we can approximate the total demand received by vertex $v \in A_t$ by 
\begin{align*}
\sum_{u \in A_t} \frac{\bF_t^\circ(u_\circ, v_\circ) \bb(u)}{\sum_{w_\circ \in A_t^\circ} \bF^\circ_t(u_\circ, w_\circ)}.
\end{align*}
By Claim~\ref{cla: small pot struct}, we can further approximate the denominator as $\bd_t(u) \cdot \frac{\bF_t^\circ(A_t^\circ, A_t^\circ)}{\bd_t(A_t)}$ and each term $\bF^\circ_t(u_\circ, v_\circ)$ in the numerator as $\bd_t(u) {\bmu}_t(v_\circ)$. 
Hence, we can approximate this sum as 
\[
\frac{{\bmu}_t(v_\circ)  \bd(A_t)}{\bF^\circ_t(A_t^\circ, A_t^\circ)}\sum_{u \in A_t} \bb(u),
\]
which is $0$ since $\sum_{u_\circ \in A_t} \bb(u) = 0$. For $C$ sufficiently large, we can ensure that the true value is within $1/n^2$ of the approximated value of $0$, meaning we can trivially route the remaining demand, incurring additional congestion at most $n^2 \cdot 1/n^2 = 1$. Then, the final congestion is at most $6T/\phi + 1 \leq 7T/\phi$, as required.
\end{proof}

\begin{remark}
This argument also trivially shows that $\bd$ mixes with congestion at most $7$ on the graph formed by the union of the matchings from each matching step. (This graph is formed as follows. In each matching step, after taking a path decomposition of the matching, we add an weighted edge between the endpoints of each path with weight equal to the flow along the path.) This is because instead of routing using flow paths, we can route using the edges in the union of matchings graph corresponding to those flow paths. Consequently, we get no added congestion from overlapping flow paths. 
\end{remark}

We can now state our key lemma, lower bounding the decrease in $\psi(t)$ and $\psip(t)$ in each step.
\begin{lemma} \label{lem: general dec lemma}
For all $t$, we have
\begin{align*}
    {\psi}(t-1) - {\psi}(t) &\geq \frac{1}{8} \sum_{u_\circ \in A_t^\circ} \bd_t(u) \sum_{w_\circ \in A_{t-1}^\circ} \frac{1}{\bd_{t-1}(w)} \left(\frac{\bbF_{t-1}(u_\circ, w_\circ)}{\bd_{t-1}(u)}  -  \sum_{v \in \bar{A}_{t-1}} \frac{\bbM_t(v, u_\circ)}{\bd_t(u)} \frac{\bbF_{t-1}(v_\circ, w_\circ)}{\bd_{t-1}(v)} \right)^2 \\
    & \quad + \frac{1}{2} \sum_{u_\circ \in A_{t-1}^\circ \setminus A_{t}^\circ} \bd_{t-1}(u) \sum_{w_\circ \in A_{t-1}^\circ} \frac{1}{\bd_{t-1}(w)} \left(\frac{\bbF_{t-1}(u_\circ, w_\circ)}{\bd_{t-1}(u)} - {\bmu}_{t-1}(w_\circ)\right)^2.
\end{align*}
We also have 
\begin{align*}
    {\psip}(t-1) - {\psip}(t) &\geq \frac{1}{8} \sum_{u_\circ \in A_t^\circ} \bd_t(u) \sum_{w_\circ \in A_{t-1}^\circ} \frac{1}{\bd_{t-1}(w)} \left(\frac{\cbF_{t-1}(u_\circ, w_\circ)}{\bd_{t-1}(u)}  -  \sum_{v \in \bar{A}_{t-1}} \frac{\bbM_t( u_\circ, v)}{\bd_t(u)} \frac{\cbF_{t-1}(v_\circ, w_\circ)}{\bd_{t-1}(v)} \right)^2 \\
    & \quad + \frac{1}{2} \sum_{u_\circ \in A_{t-1}^\circ \setminus A_{t}^\circ} \bd_{t-1}(u) \sum_{w_\circ \in A_{t-1}^\circ} \frac{1}{\bd_{t-1}(w)} \left(\frac{\cbF_{t-1}(u_\circ, w_\circ)}{\bd_{t-1}(u)} - {\bmup}_{t-1}(w_\circ)\right)^2.
\end{align*}
\end{lemma}
\begin{remark} \label{rmk: dont undo progress}
Note that the decrease in potential is non-negative in each step, regardless of how we choose our cuts in the cut step. This is important to ensure that we do not undo our progress in the potential $\psi(t)$ while making progress on potential $\psip(t)$.
\end{remark}
\begin{proof}
 We prove the first part of the claim. The second part follows analogously, using the recurrence for $\cbF_t$ given in~\cref{rmk: other ord flow upd} instead as opposed to the one in~\cref{rmk: alt recur uo}.
 
 We begin by upper bounding $\psi(t)$. Let $\bmu_{t-1}' \in \R^{A_{t-1}^\circ}$ be a vector to be defined later. By \cref{rmk: argmin}, we have 
\[
\bmu_t = \arg \min_{\bx} \sum_{u_\circ \in A_t^\circ} \bd_t(u) \sum_{v_\circ \in A_t^\circ} \frac{1}{\bd_t(v)}(\bF_t^\circ(u_\circ, v_\circ)/\bd_t(u) - \bx(v_\circ))^2,
\]
so for any choice of $\bmu_{t-1}' \in \R^{A_{t-1}^\circ}$,
\begin{align*}
    \psi(t) &= \sum_{u_\circ \in A_t^\circ} \bd_t(u) \sum_{v_\circ \in A_t^\circ} \frac{1}{\bd_t(v)}(\bF_t^\circ(u_\circ, v_\circ)/\bd_t(u) - \bmu_t(v_\circ))^2 \\
    &\leq \sum_{u_\circ \in A_t^\circ} \bd_t(u) \sum_{v_\circ \in A_t^\circ} \frac{1}{\bd_t(v)}(\bF_t^\circ(u_\circ, v_\circ)/\bd_t(u) - \bmu'_{t-1}(v_\circ))^2 \\
    &= \sum_{u_\circ \in A_t^\circ} \bd_t(u) \sum_{v_\circ \in A_t^\circ} \frac{1}{\bd_t(v)}(\bbF_t(u_\circ, v_\circ)/\bd_t(u) - \bmu'_{t-1}(v_\circ))^2.
\end{align*}
We wish to write our upper bound on $\psi(t)$ in a more convenient form. Observe that, using~\cref{rmk: alt recur uo}, we have, for any $u_\circ, w_\circ \in A_t^\circ$, 
\begin{align*}
&\left(\frac{\bbF_t(u_\circ, w_\circ)}{\bd_t(u)} - \bmu_{t-1}'(w_\circ)\right)^2 \\
&= \left(\frac{1}{2}\left(\frac{\bbF_{t-1}'(u_\circ, w_\circ)}{\bd_{t-1}(u)} - \bmu_{t-1}'(w_\circ)\right) + \frac{1}{2}\left( \sum_{v \in \bar{A}_{t-1}} \frac{\bbM_t(v, u_\circ)}{\bd_t(u)} \frac{\bbF_{t-1}'(v_\circ, w_\circ)}{\bd_{t-1}(v)} - \bmu_{t-1}'(w_\circ)\right)\right)^2\\
&=\frac{1}{4}\left(\frac{\bbF_{t-1}'(u_\circ, w_\circ)}{\bd_{t-1}(u)} - \bmu_{t-1}'(w_\circ)\right)^2  + \frac{1}{4} \left( \sum_{v \in \bar{A}_{t-1}} \frac{\bbM_t(v, u_\circ)}{\bd_t(u)} \frac{\bbF_{t-1}'(v_\circ, w_\circ)}{\bd_{t-1}(v)} - \bmu_{t-1}'(w_\circ)\right)^2 \\
&\quad+ \frac{1}{2} \left(\frac{\bbF_{t-1}'(u_\circ, w_\circ)}{\bd_{t-1}(u)} - \bmu_{t-1}'(w_\circ)\right) \left( \sum_{v \in \bar{A}_{t-1}} \frac{\bbM_t(v, u_\circ)}{\bd_t(u)} \frac{\bbF_{t-1}'(v_\circ, w_\circ)}{\bd_{t-1}(v)} - \bmu_{t-1}'(w_\circ)\right).\\
\end{align*} 
Moreover, note that 
{
\allowdisplaybreaks
\begin{align*}
&\left(\frac{\bbF_{t-1}'(u_\circ, w_\circ)}{\bd_{t-1}(u)}  -  \sum_{v \in \bar{A}_{t-1}} \frac{\bbM_t(v, u_\circ)}{\bd_t(u)} \frac{\bbF_{t-1}'(v_\circ, w_\circ)}{\bd_{t-1}(v)} \right)^2 \\
&= \left(\left(\frac{\bbF_{t-1}'(u_\circ, w_\circ)}{\bd_{t-1}(u)}  - \bmu_{t-1}'(w_\circ) \right) - \left( \sum_{v \in \bar{A}_{t-1}} \frac{\bbM_t(v, u_\circ)}{\bd_t(u)} \frac{\bbF_{t-1}'(v_\circ, w_\circ)}{\bd_{t-1}(v)} - \bmu_{t-1}'(w_\circ) \right)\right)^2   \\
&= \left(\frac{\bbF_{t-1}'(u_\circ, w_\circ)}{\bd_{t-1}(u)} - \bmu_{t-1}'(w_\circ)\right)^2 + \left( \sum_{v \in \bar{A}_{t-1}} \frac{\bbM_t(v, u_\circ)}{\bd_t(u)} \frac{\bbF_{t-1}'(v_\circ, w_\circ)}{\bd_{t-1}(v)} - \bmu_{t-1}'(w_\circ)\right)^2 \\
&\quad- 2\left(\frac{\bbF_{t-1}'(u_\circ, w_\circ)}{\bd_{t-1}(u)} - \bmu_{t-1}'(w_\circ)\right) \left( \sum_{v \in \bar{A}_{t-1}} \frac{\bbM_t(v, u_\circ)}{\bd_t(u)} \frac{\bbF_{t-1}'(v_\circ, w_\circ)}{\bd_{t-1}(v)} - \bmu_{t-1}'(w_\circ)\right).
\end{align*}
}
That is, this expression shares the same terms as $\left(\frac{\bbF_t(u_\circ, w_\circ)}{\bd_t(u)} - \bmu_{t-1}'(w_\circ)\right)^2$ with different coefficients.
Hence, expanding out the definition of $\bbF_{t-1}'$ and setting $\bmu_{t-1}'(w_\circ)= \frac{\bmu_{t-1}(w_\circ) \bd_t(w)}{\bd_{t-1}(w)}$, we continue our upper bound on $\psi(t)$ by 
{
\allowdisplaybreaks
\begin{align*}
&\sum_{u_\circ \in A_t^\circ} \bd_t(u) \sum_{w_\circ \in A_t^\circ} \frac{1}{\bd_t(w)}\left(\frac{\bbF_t(u_\circ, w_\circ)}{\bd_t(u)} - \bmu_{t-1}'(w_\circ)\right)^2   \\
&= - \frac{1}{4} \sum_{u_\circ \in A_t^\circ} \bd_t(u) \sum_{w_\circ \in A_t^\circ} \frac{1}{\bd_t(w)} \left(\frac{\bbF_{t-1}'(u_\circ, w_\circ)}{\bd_{t-1}(u)}  -  \sum_{v \in \bar{A}_{t-1}} \frac{\bbM_t(v, u_\circ)}{\bd_t(u)} \frac{\bbF_{t-1}'(v_\circ, w_\circ)}{\bd_{t-1}(v)} \right)^2 \\
&\quad+\frac{1}{2} \sum_{u_\circ \in A_t^\circ} \bd_t(u) \sum_{w_\circ \in A_t^\circ} \frac{1}{\bd_t(w)} \left(\frac{\bbF_{t-1}'(u_\circ, w_\circ)}{\bd_{t-1}(u)} - \bmu_{t-1}'(w_\circ)\right)^2 \\
&\quad+ \frac{1}{2} \sum_{u_\circ \in A_t^\circ} \bd_t(u) \sum_{w_\circ \in A_t^\circ} \frac{1}{\bd_t(w)} \left( \sum_{v \in \bar{A}_{t-1}} \frac{\bbM_t(v, u_\circ)}{\bd_t(u)} \frac{\bbF_{t-1}'(v_\circ, w_\circ)}{\bd_{t-1}(v)} - \bmu_{t-1}'(w_\circ)\right)^2 \\
&=- \frac{1}{4} \sum_{u_\circ \in A_t^\circ} \bd_t(u) \sum_{w_\circ \in A_t^\circ} \textcolor{red}{\frac{\bd_t(w)}{\bd_{t-1}(w)^2}} \left(\frac{\textcolor{red}{\bbF_{t-1}(u_\circ, w_\circ)}}{\bd_{t-1}(u)}  -  \sum_{v \in \bar{A}_{t-1}} \frac{\bbM_t(v, u_\circ)}{\bd_t(u)} \frac{\textcolor{red}{\bbF_{t-1}(v_\circ, w_\circ)}}{\bd_{t-1}(v)} \right)^2 \\
&\quad+\frac{1}{2} \sum_{u_\circ \in A_t^\circ} \bd_t(u) \sum_{w_\circ \in A_t^\circ} \textcolor{red}{\frac{\bd_t(w)}{\bd_{t-1}(w)^2}} \left(\frac{\textcolor{red}{\bbF_{t-1}(u_\circ, w_\circ)}}{\bd_{t-1}(u)} - \frac{\bmu_{t-1}'(w_\circ)\textcolor{red}{\bd_{t-1}(w)}}{\textcolor{red}{\bd_{t}(w)}}\right)^2 \\
&\quad+ \frac{1}{2} \sum_{u_\circ \in A_t^\circ} \bd_t(u) \sum_{w_\circ \in A_t^\circ} \textcolor{red}{\frac{\bd_t(w)}{\bd_{t-1}(w)^2}} \left( \sum_{v \in \bar{A}_{t-1}} \frac{\bbM_t(v, u_\circ)}{\bd_t(u)} \frac{\textcolor{red}{\bbF_{t-1}(v_\circ, w_\circ)}}{\bd_{t-1}(v)} - \frac{\bmu_{t-1}'(w_\circ)\textcolor{red}{\bd_{t-1}(w)}}{\textcolor{red}{\bd_{t}(w)}}\right)^2. \\
&= -\frac{1}{4} \sum_{u_\circ \in A_t^\circ} \bd_t(u) \sum_{\textcolor{red}{w_\circ \in A_{t-1}^\circ}} \frac{\bd_t(w)}{\bd_{t-1}(w)^2} \left(\frac{\bbF_{t-1}(u_\circ, w_\circ)}{\bd_{t-1}(u)}  -  \sum_{v \in \bar{A}_{t-1}} \frac{\bbM_t(v, u_\circ)}{\bd_t(u)} \frac{\bbF_{t-1}(v_\circ, w_\circ)}{\bd_{t-1}(v)} \right)^2 \\
&\quad+\frac{1}{2} \sum_{u_\circ \in A_t^\circ} \bd_t(u) \sum_{\textcolor{red}{w_\circ \in A_{t-1}^\circ}} \frac{\bd_t(w)}{\bd_{t-1}(w)^2} \left(\frac{\bbF_{t-1}(u_\circ, w_\circ)}{\bd_{t-1}(u)} - \textcolor{red}{{\bmu}_{t-1}(w_\circ)}\right)^2 \\
&\quad+ \frac{1}{2} \sum_{u_\circ \in A_t^\circ} \bd_t(u) \sum_{\textcolor{red}{w_\circ \in A_{t-1}^\circ}} \frac{\bd_t(w)}{\bd_{t-1}(w)^2} \left( \sum_{v \in \bar{A}_{t-1}} \frac{\bbM_t(v, u_\circ)}{\bd_t(u)} \frac{\bbF_{t-1}(v_\circ, w_\circ)}{\bd_{t-1}(v)} - \textcolor{red}{{\bmu}_{t-1}(w_\circ)}\right)^2 \\
&\leq -\frac{1}{4} \sum_{u_\circ \in A_t^\circ} \bd_t(u) \sum_{w_\circ \in A_{t-1}^\circ} \frac{\bd_t(w)}{\bd_{t-1}(w)^2} \left(\frac{\bbF_{t-1}(u_\circ, w_\circ)}{\bd_{t-1}(u)}  -  \sum_{v \in \bar{A}_{t-1}} \frac{\bbM_t(v, u_\circ)}{\bd_t(u)} \frac{\bbF_{t-1}(v_\circ, w_\circ)}{\bd_{t-1}(v)} \right)^2 \\
&\quad+\frac{1}{2} \sum_{u_\circ \in A_t^\circ} \bd_t(u) \sum_{w_\circ \in A_{t-1}^\circ} \textcolor{red}{\frac{1}{\bd_{t-1}(w)}} \left(\frac{\bbF_{t-1}(u_\circ, w_\circ)}{\bd_{t-1}(u)} - {\bmu}_{t-1}(w_\circ)\right)^2 \\
&\quad+ \frac{1}{2} \sum_{u_\circ \in A_t^\circ} \bd_t(u) \sum_{w_\circ \in A_{t-1}^\circ} \frac{\bd_t(w)}{\bd_{t-1}(w)^2} \left( \sum_{v \in \bar{A}_{t-1}} \frac{\bbM_t(v, u_\circ)}{\bd_t(u)} \frac{\bbF_{t-1}(v_\circ, w_\circ)}{\bd_{t-1}(v)} - {\bmu}_{t-1}(w_\circ)\right)^2 
\end{align*}
}
We \textcolor{red}{mark} the changes in consecutive steps for clarity. Lastly, we apply Jensen's inequality to upper bound the third term in the final expression by a more useful quantity. Note that 
{
\allowdisplaybreaks
\begin{align*}
&\sum_{u_\circ \in A_t^\circ} \bd_t(u) \sum_{w_\circ \in A_{t-1}^\circ} \frac{\bd_t(w)}{\bd_{t-1}(w)^2} \left( \sum_{v \in \bar{A}_{t-1}} \frac{\bbM_t(v, u_\circ)}{\bd_t(u)} \frac{\bbF_{t-1}(v_\circ, w_\circ)}{\bd_{t-1}(v)} - \bmu_{t-1}(w_\circ)\right)^2 \\
&\leq  \sum_{w_\circ \in A_{t-1}^\circ} \frac{\bd_t(w)}{\bd_{t-1}(w)^2} \sum_{u_\circ \in A_t^\circ} \sum_{v \in \bar{A}_{t-1}} \frac{\bbM_t(v, u_\circ)}{\bd_t(u)} \bd_t(u) \left( \frac{\bbF_{t-1}(v_\circ, w_\circ)}{\bd_{t-1}(v)} - \bmu_{t-1}(w_\circ)\right)^2 \\
&= \sum_{w_\circ \in A_{t-1}^\circ} \frac{\bd_t(w)}{\bd_{t-1}(w)^2}  \sum_{v \in \bar{A}_{t-1}} \left( \sum_{u_\circ \in A_t^\circ} \bbM_t(v, u_\circ) \right) \left( \frac{\bbF_{t-1}(v_\circ, w_\circ)}{\bd_{t-1}(v)} - \bmu_{t-1}(w_\circ)\right)^2 \\
&\leq \sum_{v_\circ \in A_{t-1}^\circ} \bd_{t-1}(v) \sum_{w_\circ \in A_{t-1}^\circ} \frac{1}{\bd_{t-1}(w)} \left( \frac{\bbF_{t-1}(v_\circ, w_\circ)}{\bd_{t-1}(v)} - \bmu_{t-1}(w_\circ)\right)^2.
\end{align*}
}
We apply Jensen's inequality in the first step, using that, since $\sum_{v \in \bar{A}_{t-1}} \bbM_t(v, u_\circ) = \bd_t(u)$ by \cref{rmk: matching equality}, we can treat the $\bbM_t(v, u_\circ)/ \bd_t(u)$ terms as probabilities and the entire summation as an expectation. The second inequality uses that $\bd_t(w) \leq \bd_{t-1}(w)$ and that $\sum_{v \in \bar{A}_{t-1}} \bbM_t(v, u_\circ) \leq \bd_{t-1}(v)$.

Combining all of our upper bounds for ${\psi}(t)$, we get
\begin{align*}
    {\psi}(t-1) - {\psi}(t) &\geq \frac{1}{4} \sum_{u_\circ \in A_t^\circ} \bd_t(u) \sum_{w_\circ \in A_{t-1}^\circ} \frac{\bd_t(w)}{\bd_{t-1}(w)^2} \left(\frac{\bbF_{t-1}(u_\circ, w_\circ)}{\bd_{t-1}(u)}  -  \sum_{v \in \bar{A}_{t-1}} \frac{\bbM_t(v, u_\circ)}{\bd_t(u)} \frac{\bbF_{t-1}(v_\circ, w_\circ)}{\bd_{t-1}(v)} \right)^2 \\
    &\quad+ \frac{1}{2} \sum_{u_\circ \in A_{t-1}^\circ \setminus A_{t}^\circ} \bd_{t-1}(u) \sum_{w_\circ \in A_{t-1}^\circ} \frac{1}{\bd_{t-1}(w)} \left(\frac{\bbF_{t-1}(u_\circ, w_\circ)}{\bd_{t-1}(u)} - {\bmu}_{t-1}(w_\circ)\right)^2 \\
    &\geq \frac{1}{8} \sum_{u_\circ \in A_t^\circ} \bd_t(u) \sum_{w_\circ \in A_{t-1}^\circ} \frac{1}{\bd_{t-1}(w)} \left(\frac{\bbF_{t-1}(u_\circ, w_\circ)}{\bd_{t-1}(u)}  -  \sum_{v \in \bar{A}_{t-1}} \frac{\bbM_t(v, u_\circ)}{\bd_t(u)} \frac{\bbF_{t-1}(v_\circ, w_\circ)}{\bd_{t-1}(v)} \right)^2 \\
    &\quad+ \frac{1}{2} \sum_{u_\circ \in A_{t-1}^\circ \setminus A_{t}^\circ} \bd_{t-1}(u) \sum_{w_\circ \in A_{t-1}^\circ} \frac{1}{\bd_{t-1}(w)} \left(\frac{\bbF_{t-1}(u_\circ, w_\circ)}{\bd_{t-1}(u)} - {\bmu}_{t-1}(w_\circ)\right)^2,
\end{align*}
with the last inequality using that $\bd_t(w) = 0$ or $\bd_t(w) \geq \bd_{t-1}(w)/2$. This yields the desired result.
\end{proof}
Finally, we show that \cref{lem: general dec lemma} implies that the expectation of $\psi(t)$ is a large factor smaller than the expectation of $\psi(t-1)$ for $t \leq T/2$ and the expectation of $\psip(t)$ is a large factor smaller than the expectation of $\psip(t-1)$ for $t > T/2$. To do so, we will use \cref{lem:apd-finding-S}, applied as described in \cref{rmk: use of apd-finding-S}, and two auxiliary lemmas. 

We begin with a standard lemma regarding the concentration of inner products with random unit vectors. For $u \in A_{t-1}$ and $t \leq T/2$, we define 
\[
p_t(m_u) := \inner{\sum_{v \in \bar{A}_{t-1}} \frac{\bbM_t(v, u_\circ)}{\bd_t(u)} \frac{\tilde{\bF}^\circ_{t-1}(v_\circ)}{\bd_{t-1}(v)}}{\br_t}.
\]
Likewise, define $\bar{\mu}_t := \inner{\tilde{\bmu}_t}{\br_t}$, where  $\tilde{\bmu}_t$ is $\bmu_t$ with $w_\circ$\textsuperscript{th} entry scaled by  $1/\sqrt{\bd_{t-1}(w)}$. For $t > T/2$, we instead define 
\[
p_t(m_u) := \inner{\sum_{v \in \bar{A}_{t-1}} \frac{\bbM_t(u_\circ, v)}{\bd_t(u)} \frac{\tcbF^\circ_{t-1}(v_\circ)}{\bd_{t-1}(v)}}{\br_t}
\]
and $\bar{\mu}_t := \inner{\tbmup_t}{\br_t}$, where  $\tbmup_t$ is $\bmup_t$ with $w_\circ$\textsuperscript{th} entry scaled by  $1/\sqrt{\bd_{t-1}(w)}$

\begin{lemma}[Analogous to Lemmas 3.5, 3.6 of \cite{racke2014computing}] \label{lem:subgaussian concentration}
For each $t \leq T/2$, for all $u\in A_t$, we have 
\begin{align*}
\E[(p_t(u) - p_t(m_u))^2] &= \frac{1}{|A_{t-1}|} \sum_{w_\circ \in A_{t-1}^\circ} \frac{1}{\bd_{t-1}(w)} \left(\frac{\bbF_{t-1}(u_\circ, w_\circ)}{\bd_{t-1}(u)}  -  \sum_{v \in \bar{A}_{t-1}} \frac{\bbM_t(v, u_\circ)}{\bd_t(u)} \frac{\bbF_{t-1}(v_\circ, w_\circ)}{\bd_{t-1}(v)} \right)^2, \\
\E[(p_t(u) - \bar{\mu}_t)^2] &= \frac{1}{|A_{t-1}|}\sum_{w_\circ \in A_{t-1}^\circ} \frac{1}{\bd_{t-1}(w)} \left(\frac{\bbF_{t-1}(u_\circ, w_\circ)}{\bd_{t-1}(u)} - {\bmu}_{t-1}(w_\circ)\right)^2.
\end{align*}
Similarly, for $t > T/2$, for all $u \in A_t$ we have 
\begin{align*}
\E[(p_t(u) - p_t(m_u))^2] &= \frac{1}{|A_{t-1}|} \sum_{w_\circ \in A_{t-1}^\circ} \frac{1}{\bd_{t-1}(w)} \left(\frac{\cbF_{t-1}(u_\circ, w_\circ)}{\bd_{t-1}(u)}  -  \sum_{v \in \bar{A}_{t-1}} \frac{\bbM_t(u_\circ, v)}{\bd_t(u)} \frac{\cbF_{t-1}(v_\circ, w_\circ)}{\bd_{t-1}(v)} \right)^2, \\
\E[(p_t(u) - \bar{\mu}_t)^2] &= \frac{1}{|A_{t-1}|}\sum_{w_\circ \in A_{t-1}^\circ} \frac{1}{\bd_{t-1}(w)} \left(\frac{\cbF_{t-1}(u_\circ, w_\circ)}{\bd_{t-1}(u)} - {\bmup}_{t-1}(w_\circ)\right)^2.
\end{align*}
Moreover, each is at most $C \log n$ times their expectation (for constant $C > 0$) with high probability in $n$.
\end{lemma}
We need one additional technical lemma. (Here $\eta$ comes from our application of~\cref{lem:apd-finding-S}. See~\cref{rmk: use of apd-finding-S}.)
\begin{lemma}\label{lem: matched on other side of eta}
For all $t \leq T$ and $u \in A_{t-1}$, we have 
\[
(p_t(u) - p_t(m_u))^2 \geq (p_t(u) - \eta)^2.
\]
\end{lemma}
\begin{proof}
The proof is analogous for $t > T/2$, so we assume $t \leq T/2$. Suppose that $p_t(u) \geq \eta$. We show that $p_t(m_u) \leq \eta$. By linearity of inner products, 
\begin{align*}
p_t(m_u) &= \sum_{v \in \bar{A}_{t-1}} \frac{\bbM_t(v, u_\circ)}{\bd_t(u)} \inner{\frac{\tilde{\bF}_{t-1}^\circ(v_\circ)}{\bd_{t-1}(v)}}{\br_t}  \\
&=  \sum_{v \in \bar{A}_{t-1}} \frac{\bbM_t(v, u_\circ)}{\bd_t(u)} p_t(v) \\
&\leq \sum_{v \in \bar{A}_{t-1}} \frac{\bbM_t(v, u_\circ)}{\bd_t(u)} \eta \\
&= \frac{\bd_t(u)}{\bd_t(u)} \cdot \eta = \eta.
\end{align*}
The inequality follows from the fact that we only match vertices in $L_t$ to those in $R_t$ and vice versa. So, since $p_t(u) \geq \eta$, for all $v$ such that $\bbM_t(v, u_\circ) > 0$, $p_t(v) \leq \eta$. Finally, by definition of $\bd_t(u)$, we have $\sum_{v \in \bar{A}_{t-1}} \bbM_t(v, u_\circ) = \bd_t(u)$. The proof for $p_t(u) < \eta$ is analogous.
\end{proof}

We are now ready to prove that the potential decreases in each round.

\begin{lemma} \label{lem: exp decrease}
    For all $1 \leq t \leq T/2$, we have 
    \begin{equation*}
     \mathbb{E}[{\psi}(t)]\le (1-\Omega(1/\log{n}))\cdot \mathbb{E}[\psi(t-1)] + O(1/\poly n).
    \end{equation*}
    For $t > T/2$, we have 
    \begin{equation*}
     \mathbb{E}[{\psip}(t)]\le (1-\Omega(1/\log{n}))\cdot \mathbb{E}[\psip(t-1)] + O(1/\poly n).
     \end{equation*}
\end{lemma}
\begin{proof}
We prove the first part of the claim. The proof of the second part of claim is analogous, instead using the second parts of~\cref{lem: general dec lemma},~\cref{lem:subgaussian concentration}.

First, we apply~\cref{lem:apd-finding-S} in our setting by considering our ($\bd_{t-1}$-weighted) bipartition $V_{t-1} = L_{t} \sqcup R_t$ from the cut step, replacing each $p_t(u)$ with $\bd_{t-1}(u)$ many copies of $p_t(u)$. The original bipartition then corresponds to an (unweighted) bisection of the  multiset with the same $\eta$. Additionally, by linearity of inner products, $\bar{\mu}_t$ remains the average of the elements of the set.

Then, applying~\cref{lem:apd-finding-S} directly, we get that there exists $S \subseteq L_t$ or $S \subseteq R_t$ such that $(p_t(u) - \eta)^2 \geq \frac{1}{9} (p_t(u) - \bar{\mu}_t)^2$ for all $u \in S$, and 
 \[
 \sum_{u \in S} \bd_{t-1}(u) (p_t(u) - \bar{\mu}_t)^2 \geq \frac{1}{32} \sum_{u \in A_{t-1}} \bd_{t-1}(u)(p_t(u) - \bar{\mu}_t)^2.
    \]
Now, applying~\cref{lem: general dec lemma}, ~\cref{lem:subgaussian concentration}, and~\cref{lem: matched on other side of eta} we get, with high probability, 
{
\allowdisplaybreaks
\begin{align*}
     {\psi}(t-1) - {\psi}(t) &\geq \frac{1}{8} \sum_{u_\circ \in A_t^\circ} \bd_t(u) \sum_{w_\circ \in A_{t-1}^\circ} \frac{1}{\bd_{t-1}(w)} \left(\frac{\bbF_{t-1}(u_\circ, w_\circ)}{\bd_{t-1}(u)}  -  \sum_{v \in \bar{A}_{t-1}} \frac{\bbM_t(v, u_\circ)}{\bd_t(u)} \frac{\bbF_{t-1}(v_\circ, w_\circ)}{\bd_{t-1}(v)} \right)^2 \\
    &\quad+ \frac{1}{2} \sum_{u_\circ \in A_{t-1}^\circ \setminus A_{t}^\circ} \bd_{t-1}(u) \sum_{w_\circ \in A_{t-1}^\circ} \frac{1}{\bd_{t-1}(w)} \left(\frac{\bbF_{t-1}(u_\circ, w_\circ)}{\bd_{t-1}(u)} - {\bmu}_{t-1}(w_\circ)\right)^2 \\
    &\geq \frac{|A_{t-1}|}{8C \log n} \sum_{u_\circ \in A_t^\circ} \bd_t(u) (p_t(u) - p_t(m_u))^2 + \frac{|A_{t-1}|}{2C \log n} \sum_{u_\circ \in A_{t-1}^\circ \setminus A_{t}^\circ} \bd_{t-1}(u) (p_t(u) - \bar{\mu}_t)^2 \\
    &\geq \frac{|A_{t-1}|}{16C \log n} \sum_{u_\circ \in A_t^\circ} \bd_{t-1}(u) (p_t(u) - \eta)^2 + \frac{|A_{t-1}|}{2C \log n} \sum_{u_\circ \in A_{t-1}^\circ \setminus A_{t}^\circ} \bd_{t-1}(u) (p_t(u) - \bar{\mu}_t)^2 \\
    &\geq \frac{|A_{t-1}|}{16C \log n} \sum_{u \in A_t \cap S} \bd_{t-1}(u) (p_t(u) - \eta)^2 + \frac{|A_{t-1}|}{2C \log n} \sum_{u \in (A_{t-1} \setminus A_{t}) \cap S} \bd_{t-1}(u) (p_t(u) - \bar{\mu}_t)^2 \\
    &\geq \frac{|A_{t-1}|}{144C \log n} \sum_{u \in A_t \cap S} \bd_{t-1}(u) (p_t(u) - \bar{\mu}_t)^2 + \frac{|A_{t-1}|}{2C \log n} \sum_{u \in (A_{t-1} \setminus A_{t}) \cap S} \bd_{t-1}(u) (p(u) - \bar{\mu}_t)^2 \\
     &\geq \frac{|A_{t-1}|}{144C \log n} \sum_{u \in S} \bd_{t-1}(u) (p_t(u) - \bar{\mu}_t)^2 \\  
     &\geq \frac{|A_{t-1}|}{2304C \log n} \sum_{u_\circ \in A_{t-1}^\circ} \bd_{t-1}(u) (p_t(u) - \bar{\mu}_t)^2. 
\end{align*}
}
Next, note that we have, by \cref{lem:subgaussian concentration},
\begin{align*}
&\frac{|A_{t-1}|}{2304C \log n} \sum_{u_\circ \in A_{t-1}^\circ} \bd_{t-1}(u) \E\left[(p_t(u) - \bar{\mu}_t)^2\right] \\
&= \frac{1}{2304C \log n} \sum_{u_\circ \in A_{t-1}^\circ} \bd_{t-1}(u) \sum_{w_\circ \in A_{t-1}^\circ}  \left(\frac{\bbF_{t-1}(u_\circ, w_\circ)}{\bd_{t-1}(u)} - {\bmu}_{t-1}(w_\circ)\right)^2 \\
&= \frac{1}{2304C \log n}\psi(t-1).
\end{align*}
Hence, conditioning on the high probability events above, 
    \begin{equation*}
        \E[\psi(t - 1)-\psi(t)] \geq \Omega(\psi(t - 1)/\log n).
    \end{equation*}
    Thus, we can conclude
     \begin{equation*}
        \E[\psi(t- 1)-\psi(t)] \geq \Omega(\psi(t - 1)/\log n - 1/\poly(n)).
    \end{equation*}
\end{proof}
\begin{remark}
    \cref{lem: exp decrease}, combined with~\cref{rmk: dont undo progress}, immediately implies that $\psi(T) + \psip(T) \leq 1/(nW)^C$ for $ T = O(\log n \log nW)$ with high probability. The high probability guarantee follows from Markov's inequality.
\end{remark}

\subsection{Post-processing} \label{subsec: grafting}
At termination of the algorithm, each cut $C_j$ in the sequence of cuts $C_1, \ldots, C_k$ computed is disjoint and satisfies $\Phi_{G[V \setminus C_{< j}], \bd}(C_j) < 3 \phi$, by our matching step. We also guarantee in the matching step that $\bd(C_j) \leq 2\bd(V)/3$.

If we terminated the algorithm via the early termination condition, we know that, at termination, we had $\bd_t(A_t) < 0.99\bd(A)$. Hence, by \cref{cla: D_t bd}, we then have 
\[
\bd(C_{\leq k_t}) > (\bd(V) - \bd_t(V))/36 =(\bd(A) - \bd_t(A_t))/36 \geq \bd(V)/10^4,
\]
as required for the early termination case.

In the case of termination after $T$ total rounds, we know that $\bd_T(A_T) \geq 0.99\bd(A)$, since the early termination condition was not reached. Hence, \cref{lem: small pot implies exp} applies at termination, and $\bd_T$ mixes in $G$ with congestion at most $7T/\phi$. However, the set $D_T$ is not part of the sequence of sparse cuts computed thus far, so we need to do some additional post-processing. As such, we perform one final ``matching step.'' 

We will solve two flow problems. The first has source $\Delta(u) = \bd(u)$ for $u \in D_T$ and sink $\nabla(u) = \bd(u)$ for $u \in A_T$. We try to solve this flow problem on $G_T$, the same graph considered in the matching step. In case of general $\bd$, we can solve this flow problem with one call to a max flow oracle. Let $\Bf$ be the resultant flow and $S \subseteq V(G_T)$ be the cut found. Then, we know that 
\[
\Bf(S, \overline{S}) = \delta_{G_T}(S, \overline{S}) = \frac{1}{\phi} \delta_{G[V \setminus C_{\leq k_T}]}(S, \overline{S}),
\]
since all edges from $S$ to $\overline{S}$ must be saturated, since it is a min-cut. On the other hand, 
\[
\Bf(S, \overline{S}) = \Delta(S) - \nabla(S) \leq \min(\bd(S), \bd(\overline{S})).
\]
Here, $\Delta(S) \leq \bd(S)$ is clear and $\Delta(S) = \Delta(V) - \Delta(\overline{S})$, $\nabla(S) = \nabla(V) - \nabla(\overline{S})$ and $\nabla(V) \geq \Delta(V)$ together imply $\Delta(S) - \nabla(S) \leq \bd(\overline{S})$. Together, this implies that $S$ is a sparse cut. Moreover, we trivially must have $\Delta(S) - \nabla(S) > 0$. But then, since $\Delta(V) = \bd(D_T) \leq 0.02 \bd(A)$, we have $\bd(S) < 0.04 \bd(A)$. Hence, we can add $S$ to our sequence of sparse cuts. By the max-flow min-cut theorem, then the flow problem restricted to the graph without $S$ is feasible.

Let $D_T'$, $A_T'$, $G_T'$ be the vertex sets updated according to the cut found in the first flow problem. (As usual, if we now would satisfy the early termination condition, we again terminate and output our sequence of sparse cuts.) For the second flow problem, we set $\Delta(u) = \bd(u)$ for $u \in D_T'$ and sink $\nabla(u) = \bd(u)$ for $u \in A_T'$. We try to solve this problem on $G_T'$ with all edges reversed. By an analogous analysis as above, we will find a sparse cut $S'$ with $\bd(S) < 0.04 \bd(A)$ such that the flow problem restricted to $G_T'$ with $S'$ removed is feasible. We add this sparse cut to our sequence of sparse cuts.

In the case of $\bd \geq \deg_G/\kappa$, we can avoid appealing to a max flow oracle to solve the above flow problems by using dynamic push relabel (as in Lemma B.6 of \cite{saranurak2019expander}, for example). Since this application of dynamic push-relabel is well-studied in the literature (and is encompassed by push-pull relabel in this paper), we give a high-level discussion. Consider the first flow instance. Applying~\cref{lem: bdd ht push relabel}, we will find a sparse (level) cut $S = L_{\geq k}$ such that $\bd(S) < 0.04 \bd(A)$ and at most $\phi \bd(S)/10$ total flow is sent from $\bd(\bar{S})$ to $\bd(S)$. Then, we inject at most $\bd(S)/10$ additional source at the boundary in the remaining graph and repeat applying bounded height push relabel. Since we inject source at most $|E|$ total many times throughout, using link-cut trees, the total running time is bounded by $O(mh \log n)$. We repeat this process on $D_T'$, $A_T'$, $G_T'$ to solve the other flow instance. 

Denote the remaining vertices in $A_T$ by $A_{T+1}$ and the remaining vertices in $D_T$ by $D_{T+1}$ after this process. Finally, it remains to show that we can add $D_{T+1}$ back to $A_{T+1}$.

\begin{lemma}
If we did not end in the early termination case, then $\bd |_{A_{T+1} \cup D_{T+1}}$ mixes in $G$ with congestion $44T/\phi$.
\end{lemma}
\begin{proof}
First, by~\cref{lem: small pot implies exp}, we get that $\bd_T|_{A_{T}}$ mixes in $G$ with congestion $7T/\phi$, since $\bd_T|_{A_{T}} \leq \bd_T$. Then, since $\bd|_{A_{T}} \leq 2 \bd_T|_{A_{1}}$, we get that $\bd|_{A_{T}}$ mixes with congestion at most $14T/\phi$ in $G$. 

Finally, let $\bb$ be a demand with $\bb \leq |\bd_{A_{T+1} \cup D_{T+1}}|$. From the first flow instance, for vertices $u \in D_{T+1}$ with $\bb(u) > 0$, we can send their demand into $A_T$ with congestion at most $1/\phi$ (the ``sent'' demand corresponds to negative demand placed on vertices in $A_T$). From the second flow instance, for vertices $u \in D_{T+1}$ with $\bb(u) < 0$, we can send their demand via reverse edges into $A_T$ with congestion at most $1/\phi$ (the ``sent'' demand corresponds to positive demand placed on vertices in $A_T$). Combining these three demands yields a demand $\bb'$ on $A_{T}$ with $|\bb'| \leq 3\bd$. Hence, since $\bd|_{A_{T}}$ mixes with congestion at most $14T/\phi$ in $G$, we get that $\bd |_{A_{T+1} \cup D_{T+1}}$ mixes in $G$ with congestion at most $44T/\phi$, as required.
\end{proof}

%% file: weak-expander-decomp.tex
\section{Weak Expander Decomposition} \label{sec: weak exp}

In this section, we apply~\cref{thm: cut-matching} to obtain a fast near-expander decomposition algorithm. 

\begin{theorem} \label{thm: near ex decomp}
Given a (directed, weighted) graph $G = (V, E, \bc)$ of order $n$ and size $m$ with edge capacities $\bc$ bounded by $W$, a parameter $\phi \in (0,1)$, and a vertex weighting $\bd \in \N_{\geq 0}^V$ bounded by $O(\poly(n,W))$, there is a randomized algorithm that with high probability finds a partition $V = V_1 \sqcup V_2 \cdots \sqcup V_\ell$ and a $E_D \subset E$ such that 
\begin{enumerate}
    \item For each $i \in [\ell]$, $\bd |_{V_i}$ mixes in $G$ with congestion $O(\frac{\log (n) \log (nW)}{\phi})$.
    \item The edge subgraph $D = (V, E_D)$ is acyclic, and $\sum_{i = 1}^\ell \delta_{G \setminus D}(V_i) = O(\phi \bd(V) \log nW)$.
\end{enumerate}
For $\kappa \geq 1$ and $\bd \geq \deg_G / \kappa$, the algorithm runs in $O(\frac{\kappa m \log^2 (n) \log^3 (nW)}{\phi})$ time. For general $\bd$, the algorithm runs in $O(F(n,m) \log (n) \log^2 (nW) +  \log^2 (n) \log^2 (nW) m)$ time, where $F(n, m)$ is the runtime of solving a max-flow instance of order $n$ and size $m$.
\end{theorem}
\begin{proof}
Apply~\cref{thm: cut-matching}. Let $C_1, \ldots C_k \subseteq V$ be the resultant sequence of cuts such that, for each $j \in [k]$, $\Phi_{G[V \setminus C_{<j}], \bd}(C_j) \leq 3 \phi$ and $\bd(C_j) \leq 2 \bd(V)/3$, and $C_i \cap C_j = \emptyset$ for $i \neq j \in [k]$. Now, for each $j \in [k]$, let $\overline{C}_j := V \setminus C_{\leq j}$. Then, since $\Phi_{G[V \setminus C_{<j}], \bd}(C_j) \leq 3 \phi$, we have 
\begin{equation} \label{eqtn: cont to cut edges}
\min(\bc(E(C_j, \overline{C}_j)), \bc(E(\overline{C}_j, C_j))) \leq 3\phi \bd(C_j).    
\end{equation}
Then, for the maximum of the terms in the minimum on the left-hand side, add the corresponding set of edges to $E_D$. 

We then recursively apply~\cref{thm: cut-matching} to each $C_j$, using $\bd|_{C_j}$ as the vertex weighting. If we reached the early termination case in the previously layer of recursion (and only in this case), we also apply~\cref{thm: cut-matching} to $V \setminus C_{\leq k}$. We continue recursing until either reaching subgraphs with $\bd(C_j) = 0$ or decomposing $G$ into subgraphs which reached the near-expander case. These subsets are exactly the $V_i$ in our final decomposition.

Observe that, by~\cref{thm: cut-matching}, at each level we have that the total $\bd$-weight of the largest component decreases by a factor of at least $1 - 1/10^4$ (in the near expander case, it actually decreases by a factor of $2/3$). Hence, since $\bd$ is bounded by $\poly(n,W)$, the algorithm has recursion depth $O(\log nW)$. By the running time of~\cref{thm: cut-matching}, this implies the desired running time of the weak-expander decomposition algorithm.

It remains to conclude correctness of the algorithm. The first property is trivial from when we choose to recurse and the guarantees of~\cref{thm: cut-matching} (notice that any demand on a singleton vertex trivially mixes in $G$). 

For the second property, notice that in the first level of recursion we add an acyclic graph worth of edges to $E_D$. Then, the subsequent levels of recursion add only edges between vertices within components found in the first run of~\cref{thm: cut-matching}. So, there can never be cycles between vertices in different components from the first level of recursion. But, this is then true of components found in the subsequent level of recursion, and, inductively, we get that the fact that each level of recursion adds an acyclic graph of edges to $E_D$ implies that $D = (V, E_D)$ is acyclic.

By \cref{eqtn: cont to cut edges}, the first level of recursion will contribute at most $3 \phi \bd(V)$ total to $\sum_{i = 1}^\ell \delta_{G \setminus D}(V_i)$, and will handle the contributions of all edges between the components in the first level of recursion. All remaining contributing edges are then within the components from the first level of recursion. Then, since the components on each recursive level partition $V$, each level of recursion contributes at most $3 \phi \bd(V)$, and, inductively, we handle the contribution of all edges to $\sum_{i = 1}^\ell \delta_{G \setminus D}(V_i)$. As such, we get
\[
\sum_{i = 1}^\ell \delta_{G \setminus D}(V_i) \leq O(\phi \bd(V) \log nW),
\]
as needed.
\end{proof}

%% file: strong-expander-decomposition.tex
\section{Strong Expander Decomposition} \label{sec: strong exp}
In this section, we prove our main result:
\begin{theorem} \label{thm:strong-expander-decomp}
Given a (directed, weighted) graph $G = (V, E, \bc)$ with edge capacities $\bc$ bounded by $W$ and a parameter $\phi \in (0,1)$, there is a randomized algorithm that with high probability finds a partition $V = V_1 \sqcup V_2 \cdots \sqcup V_\ell$ and $E_D \subset E$ such that 
\begin{enumerate}
    \item For each $i \in [\ell]$, $G[V_i]$ is a $\phi$-expander.
    \item The edge subgraph $D = (V, E_D)$ is acyclic, and $\sum_{i = 1}^\ell \delta_{G \setminus D}(V_i) = O(\phi \deg_G(V) \log^3(n)\log^5(nW))$.
\end{enumerate}
The algorithm runs in $O(\frac{m}{\phi}+m\log^3(n)\log^4(nW))$ time. 
\end{theorem}

For technical reasons outlined in the introduction, we will be constructing a $\mathbf{d}$-expander decomposition with respect to the regularized vertex weighting
\begin{align*}
\mathbf{d}(v)=\deg_G(v)+\frac{\widetilde{\deg}_G(v)}{2m}\cdot \deg_G(V).
\end{align*}
Note that a standard expander decomposition is a $\deg_G$-expander decomposition, so the expanders in the expander decomposition we construct are strictly stronger. Furthermore, since $\bd(V)=2\deg_G(V)$, the volume of the cut edges is still preserved within a factor of two, so the $\bd$-expander decomposition is a valid $\deg_G$-expander decomposition.

Our algorithm for strong expander decompositions follows the general framework introduced in~\cite{saranurak2019expander}. First, we apply our new nonstop directed cut-matching game to obtain either a relatively balanced sequence of sparse cuts or a large $(\phi,\bd)$-near expander. Below, we restate our specific use case of~\Cref{thm: cut-matching} (see also \cref{rem:early-termination}) as a lemma. 
\begin{lemma}[Cut-Matching] \label{lem: cut-matching}
Let $c_0$ be a constant to be chosen later. Given a directed graph $G = (V, E, \bc)$ with edge capacities $\bc(e)$ bounded by $W$ and a parameter $\phi \in (0,1)$, there exists a constant $c_1>0$ and a randomized, Monte Carlo algorithm that outputs:
\begin{itemize}
    \item A sequence of cuts $C_1, \ldots, C_k \subseteq V$ such that, for each $j \in [k]$, $\Phi_{G[V \setminus C_{< j}], \bd}(C_j) \leq c_1\phi\log(n)\log(nW)$, $\bd(C_{j}) \leq 2\bd(V)/3$, and $C_i \cap C_j = \emptyset$ for $i \neq j \in [k]$.
\end{itemize}
 We also have either:
    \begin{enumerate}
        \item \textbf{Early termination:} $\bd(C_{\leq k}) > \bd(V)/(c_0\log(n)\log(nW))$.
        \item \textbf{Large near-expander:}  $\bd|_{V \setminus C_{\leq k}}$ mixes in $G$ with congestion $1/\phi$ with high probability.
    \end{enumerate}
The algorithm runs in $O(m \log(n)\log(nW)/ \phi+m\log^2(n)\log^2(nW))$ time. 
\end{lemma}

If we find a relatively balanced sequence of sparse cuts (the early termination case), we can recurse into each cut in the sequence, which makes sufficient progress. Otherwise, we find a large $\phi$-near expander $A\subseteq V$; in this case, we give a trimming algorithm to find a large subset $A'\subseteq A$ which we certify to be a ${\Omega}(\phi/\log^2n)$-expander, so that we can simply recurse on the remainder of the graph.

\begin{restatable}[Trimming]{lemma}{trimming}\label{lem:trimming} 
    Let $G=(V,E,\Bc)$ be a graph with integral edge capacities $\Bc(e)$ bounded by $W$. Let $A\subseteq V$ be a subset which is $(\phi,\bd)$-nearly expanding in $G$ such that $V\setminus A$ is the disjoint union of a sequence of sparse cuts $C_1,\ldots,C_k$, where $\Phi_{G[V\setminus C_{<j}],\bd}(C_j)\le O(\phi \log (n) \log(nW))$ for each $j\in[k]$. There exists an algorithm which finds $A'\subseteq A$ along with a sequence of cuts $S_1\sqcup S_2\sqcup\cdots=A\setminus A'$ satisfying $\Phi_{G[A\setminus S_{<j}],\bd}(S_j)\le O(\phi)$
    such that either
    \begin{enumerate}
        \item \textbf{Early termination:} $\bd(\bigcup_{j=0}^{t}S_j)\ge \bd(V)/(c_0\log(n)\log(nW))$
        \item \textbf{Certifies expansion:} $A'$ is certified as a $(\Omega(\phi/\log^2(n)\log^3(nW)),\bd)$-expander and $\bd(A')\ge \bd(A)/2$.
    \end{enumerate}
    The algorithm runs in time $O(m\log(nW)/\phi + m \log^2 (n) \log^2 (nW))$.
\end{restatable}

In order to find the certified expander $A'$ in $A$, we use a pair of flow instances on $A$: one to certify out-expansion and one to certify in-expansion. These flow instances, described in \Cref{sec:flow}, have the property that if they are feasible, then we have certified that the near-expanding subset $A$ is actually an expander (i.e., $G[A]$ is an expander). Since any subset of a near-expander is still a near-expander, the goal is to find a large subset $A'\subseteq A$ where both of the flow instances are feasible. We do this iteratively. Specifically, we run a dynamic flow algorithm (\Cref{sec:ppl}) on the two flow problems on $A_t$, starting with $A_0=A$. If either of the flow instances are infeasible, we show that we can find a sufficiently sparse cut $S_t$. We then let $A_{t+1}=A_t\setminus S_t$ and repeat the process. Note that although we may have $\Omega(n)$ different flow instances, since the cuts $S_t$ may be small, we can still bound the runtime since we are applying a dynamic flow algorithm. Eventually, we find $A_t$ where both flow instances are feasible, so we have certified expansion for $A'=A_t$, and we are done. We show that $A'$ is large and the total number of edges cut in $S_t$ is small, so they can be used as part of the expander decomposition. The full algorithm is described in \Cref{sec:trim-algorithm}.

\begin{algorithm}[t]
	\caption{Directed Expander Decomposition}
	\label{alg:strong-ed}
	\fbox{
		\parbox{0.97\columnwidth}{
		  \textsc{Decomp}$(G, A, W, \Pi, \phi)$ \\
            \tab \textbf{Call} non-stop cut-matching on $(G,\phi)$ via \Cref{lem: cut-matching}\\
            \tab \textbf{If} we have early termination via sparse cuts $C_1,\ldots,C_k$\\
            \tab \tab \textbf{Return} \textsc{Decomp}$(G[V\setminus C_{\le k}],\phi)\cup\bigcup_{j=1}^{k}\textsc{Decomp}(G[C_j],\phi)$\\
            \tab \textbf{Else} (i.e., we have a near expander $A$ along with sparse cuts $C_1,\ldots,C_k$)\\
            \tab \tab \textbf{Call} trimming algorithm on $(A,C_1,\ldots,C_k,\phi)$ via \Cref{lem:trimming}\\
            \tab \tab \textbf{If} we have early termination, with $A'$ and sparse cuts $S_1,\ldots,S_t$\\
            \tab \tab \tab \textbf{Return} $\textsc{Decomp}(G[A'],\phi)\cup\bigcup_{j=1}^{k}\textsc{Decomp}(G[C_j],\phi)\cup\bigcup_{j=0}^{t}\textsc{Decomp}(G[S_j],\phi)$\\
            \tab \tab \textbf{Else} (i.e., we have certified an $\Omega(\phi)$-expander $A'$, along with sparse cuts $S_0,\ldots,S_t$)\\
            \tab \tab \tab \textbf{Return} $\{A'\}\cup\bigcup_{j=1}^{k}\textsc{Decomp}(G[C_j],\phi)\cup\bigcup_{j=0}^{t}\textsc{Decomp}(G[S_j],\phi)$
	}}
\end{algorithm}

Assuming the two lemmas above, we can prove~\cref{thm:strong-expander-decomp}. The pseudocode is in \Cref{alg:strong-ed}.
\begin{proof}
    Since our algorithm only stops the recursive process when the subgraph is certified to have conductance $\Phi_{G[V_i],\bd}\ge\Omega(\phi/\log(n)\log(nW))$, the leaves of the recursion $V_1\sqcup\ldots\sqcup V_s=V$ form a valid  $(\Omega(\phi/\log(n)^2\log^3(nW)),\bd)$-expander decomposition. As noted previously, the vertex weighting $\bd$ dominates that of $\deg_G$, so $V_1,\ldots,V_s$ is also a valid standard $\Omega(\phi/\log(n)\log(nW))$-expander decomposition. Now, we bound the runtime of the algorithm and the number of intercluster edges. 

    First, we bound the runtime. After running the cut-matching step on a component, if we get the first case, we obtain a sequence of sparse cuts $C_1,\ldots,C_k$ such that the $\bd$-weight of the largest component decreases by a $1-\Omega(1/\log(n)\log(nW))$ factor. If we get to the second case, we apply a trimming algorithm and either obtain a large expander $A'$ along with a sequence of sparse cuts $S_0,\ldots,S_t$ or just a sequence of sparse cuts $S_0,\ldots,S_t$ such that the $\bd$-weight of largest component decreases by a $1-\Omega(\log(n)\log(nW))$ factor. In summary, the weight of the largest component decreases by a $1-\Omega(1/\log(n)\log(nW))$ factor across each level of the recursion, so the recursion goes on for $O(\log(n)\log^2(nW))$ levels.
    Since the components of each level of the recursion are disjoint, the total runtime of each level of the recursion is $O(m\log(n)\log(nW)/\phi)$, so the total runtime is $O(m\log(n)^2\log^3(nW)/\phi)$. For the runtime term without the $\phi$ factors, the total runtime is $O(m\log^3(n)\log^5(nW))$ time.   
    
    Next, we bound the intercluster edges. In the first case, we have a sequence of $(O(\phi\log(n)\log(nW)),\bd)$-sparse cuts $C_1,\ldots,C_k$. For each cut, charge the edges on the sparse side of the cut to the vertex weight in the smaller side of the cut. Since each vertex is on the smaller side of the cut at most $O(\log{nW})$ times, the volume of intercluster edges from the second case is bounded by $O(\phi\bd(V)\log(n)\log^2(nW))=O(\phi \deg_G(V)\log(n)\log^2(nW))$. In the second case, we always have a sequence of $(O(\phi),\bd)$-sparse cuts $S_1,\ldots,S_t$. Since the depth of the recursion is $O(\log(n)\log^2(nW))$ and the components being recursed on at a single level are disjoint, the volume of intercluster edges from the second case is upper bounded by $O(\phi\bd(V)\log(n)\log^2(nW))=O(\phi \deg_G(V)\log(n)\log^2(nW))$ as well.

    In summary, we can compute a $\Omega(\phi/\log^2(n)\log^3(nW))$-expander decomposition where the number of intercluster edges is $O(\phi \deg_G(V)\log^2(n)\log^3(nW))$ and the runtime is $O(m\log^2(n)\log^3(nW)/\phi)$. Parameterizing $\phi=\Theta(\psi/\log^2(n)\log^3(nW))$, this means we can compute a $\psi$-expander decomposition where the number of intercluster edges is $O(\psi\deg_G(V)\log^3(n)\log^5(nW))$ and the runtime is $O(m/\psi)$. For the runtime term without the $\phi$ factors, this is not affected by the reparameterization, so the total runtime is still $O(m\log^3(n)\log^5(nW))$ time. This matches all the parameters in the theorem statement. 
\end{proof}

\begin{remark}
    In the unweighted case, there is an improved trimming algorithm, given in \Cref{lem:trimming-unweighted}, since path decomposition can be maintained explicitly. 
    By using dynamic unit-flow~\cite{saranurak2019expander} and the analysis of Push-Pull-Relabel from~\cite{SP24}, we can save an additional $\log(n)$ factor in both the runtime of non-stop cut-matching game and the trimming algorithm. Combining all of these, we can match the runtime and intercluster edges of \cite{saranurak2019expander} in the unweighted case. The analysis is the same as the above, with the exponents in the $\log(n)$ and $\log(nW)$ factors suitably changed.
\end{remark}

\subsection{Two Flow Instances: Certifying Directed Expansion}
\label{sec:flow}

We first define our flow instance which certifies out-expansion of a directed near-expander. To certify in-expansion, we simply define the same flow instance on the same graph but with all edges reversed. The two flow instances combined together certifies directed expansion.

Let $A\subseteq V$ denote a $(\phi, \bd)$-near expander in $G$. We assume we also have a witness $(W,\Pi)$ as part of the input. In general, a witness is defined as follows:

\begin{definition}[Expansion Witness]
    A $(\psi, \kappa, \bd)$-\textit{witness of near-expansion on $A$} is a pair $(W,\Pi)$ where
    \begin{enumerate}
        \item $W$ is a graph such that $A$ is a $(\psi,\bd)$-near expander in $W$.
        \item $\Pi:E(W)\to \mathcal{P}(G)$ embeds edges $(u,v)\in E(W)$ into capacitated paths $\Pi(u,v)$ between $u$ and $v$ in $G$ with capacity $\bc_W(u,v)$, such that the congestion of the embedding is at most $\kappa$. Here $\mathcal{P}(G)$ is the set of simple paths in $G$.
    \end{enumerate}
When $A = V(W)$, then $(W, \Pi)$ is a witness of \textit{expansion} on $A$. 
\end{definition}
The choices of parameters $\psi, \kappa, \bd$ will be clear from context throughout, so we suppress specifying these parameters for readability. When the path embedding $\Pi$ is clear from context, we abuse notation to also refer to $W$ itself as the witness. 

\begin{claim} \label{cl: witness implies expansion}
Let $G = (V,E, \bc)$ be a graph. Let $A \subseteq V$. Then, if there exists $(W, \Pi)$ a $(\psi, \kappa, \bd)$-witness of near expansion on $A$, $A$ is a $(\psi/\kappa, \bd)$-near expander in $G$. (If $A = V$, then $G$ is a $(\psi/\kappa, \bd)$-expander.)
\end{claim}
\begin{proof}
We can prove this claim by using either the flow or cut definition of expansion. Although the definitions are equivalent, it is instructive to understand why both proofs hold.

For the flow-based proof, let $\bb$ be a demand on $A$ with $|\bb| \leq \bd$. Then, since $W$ is a $(\psi, \bd)$-expander, there is a routing of this demand (in $W$) with congestion at most $1/\psi$. Then, to route the same demand on $A$ in $G$, just embed this routing directly using $\Pi$, for total congestion at most $1/\psi \cdot \kappa$.

For the cut-based argument, fix $S \subseteq A$. Since $A$ is a $(\psi, \bd)$-near expander in $W$, we know that 
\[
\min(\delta_W(S, V \setminus S), \delta_W(V \setminus S, S)) \geq \psi \min(\bd(S), \bd(A \setminus S)).
\]
Since $\Pi$ embeds with congestion at most $\kappa$, we then immediately get 
\[
\frac{\min(\delta_G(S, V \setminus S), \delta_G(V \setminus S, S))}{\min(\bd(S), \bd(A \setminus S))} \geq \psi/\kappa.\qedhere
\]
\end{proof}

\begin{claim}\label{cl:witness-from-CMG}
The non-stop cut-matching game can be trivially modified to also output a $(1/4, 1/\phi, \bd)$-witness of near-expansion in the case of finding a large near-expander (case 2 of~\cref{lem: cut-matching})
\end{claim}
\begin{proof}
The witness $(W, \Pi)$ is constructed throughout the algorithm implicitly. Indeed, in each matching step, we find a flow path decomposition of the computed flow. For each path $P$ found between source $u$ and sink $v$, we add the edge $(u,v)$ to $W$ with $\bc_W(u,v) = \bc(P)$ and $\Pi(u,v) = P$. 

By choice of parameters of the cut-matching game, $\Pi$  embeds with congestion at most $1/\phi$. All that remains is to show that $A \subset V(W)$ is a $(1/4, \bd)$-near expander in $W$. But, the proof of this exactly follows the proof of~\cref{cla: small pot struct}, except now the flow paths are embedded directly into $W$ instead of into $G$, so we incur a congestion of at most $4$. 
\end{proof}

We assume that we have that $W$ is a graph defined on vertex set $V$ such that $A$ is a $0.1$-near expanding and $\Pi:E(W)\to \mathcal{P}(G)$ is an embedding of $W$ into $G$ with congestion $10/\phi$. Define a flow problem on $G[A]$ as follows. For each embedding path $\Pi(e)$ for some $(u,v)=e\in W$ which crosses the boundary $E(A,V\setminus A)$ of $A$, we add a source of $100 \cdot \bc_W(e)$ units on the endpoints $u$ and $v$ if they lie in $A$. We then add a sink of $\bd(v)$ units at each vertex $v\in A$ and set the capacity of each edge $e$ to be $200\cdot \bc(e)/\phi$. We now prove that if the flow problem is feasible, $A$ is a $(\Omega(\phi), \bd)$-out expander.

\begin{claim}
    Suppose $A\subseteq V$ is a $(\phi, \bd)$-near expander with witness $(W,\Pi)$. If the flow problem defined above has a feasible routing, then $G[A]$ is a $(\phi/10^7, \bd)$-out expander. \label{cl:flow-instance}
\end{claim}
\begin{proof}
    Let $\Bf$ denote the feasible routing for the flow problem. Using $\Bf$ and the witness $(W,\Pi)$ of near expansion, we explicitly construct a new witness $(W',\Pi')$ certifying that $G[A]$ is an $(\Omega(\phi), \bd)$-out expander. 

    Let $W'$ start out as $W$, and remove each edge $e=(u,v)$ such that $\Pi(e)$ crosses the boundary $E(A,V \setminus A)$. Let $\Pi'$ start out as the restriction of $\Pi$ on $W'$. Now, observe that $W'$ embeds into $G[A]$ instead of into $G$. Unfortunately, $W'$ may not be an expander, so we don't immediately certify expansion of $G[A]$. To repair the witness $W'$, consider a path decomposition of $\Bf$. For each flow path $P$ in $\Bf$ between endpoints $u$ and $v$, add an edge $(u,v)$ to $W'$ and set the capacity of $(u,v)$ to be the capacity of $P$. The embedding $\Pi'(u,v)$ is then naturally defined as $P$. For the embedding $\Pi'(u,v)$ of the new edges $(u,v)\in W'$, at most $200\cdot \bc(e)/\phi$ flow is routed through each edge: since $W'$ and its embedding is constructed using $\Bf$, the new edges contribute at most $200/\phi$ toward the overall congestion. Since $W$ embeds into $G$ via $\Pi$ with congestion at most $10/\phi$, we have that $\Pi'$ has congestion at most $210/\phi$. 

    Next, we show that $W'$ is an 0.0001-out expander. Consider any cut $(S,A\setminus S)$ in $W'$ where $\vol_{W'}(S)\le \vol_{W'}(A\setminus S)$; we will show it is not 0.0001-out sparse. To reason about the size of the cut $(S,A\setminus S)$, consider the edges $E_W(S,V\setminus S)$. We know that $\delta_W(S,V\setminus S)\ge 0.1\cdot \vol_{W}(S)$ since $A$ is 0.1-near expanding in $W$. If less than half of the edges in $E_W(S,V\setminus S)$ were removed when constructing $W'$, then we claim that $(S,A\setminus S)$ is not a sparse cut in $W'$. 
    Indeed, this immediately implies that $\delta_{W'}(S, A \setminus S) =\delta_{W'}(S,V\setminus S)\geq \delta_W(S, V \setminus S)/2$. Furthermore, $\vol_{W'}(S)$ is not too much larger than $\vol_W(S)$: the difference is bounded by the total source and sink. The total source is at most $100 \cdot \vol_W(S)$ and the sink is at most $\deg_G(S)\le\vol_W(S)$, so $\vol_{W'}(S)\le 102\cdot\vol_W(S)$. Combining the two bounds, we have $$\delta_{W'}(S,A\setminus S)\ge \delta_{W}(S,V\setminus S)/2\ge 0.0001\cdot 102\cdot\vol_W(S)\ge0.0001\cdot\vol_{W'}(S),$$
    where the middle inequality follows since $W$ is a $0.1$-expander.
     
    Next, suppose that more than half of the edges in $E_W(S,V\setminus S)$ were removed when constructing $W'$. For each edge in $E_W(S,V\setminus S)$ which was removed, we add 100 units of source to the endpoint in $S$ and possibly the endpoint in $V\setminus S$, if it lies in $A$. Since more than half of the edges were removed, we know that the amount of source in $S$ is at least $50\cdot \delta_W(S,V\setminus S)\ge 5\cdot\vol_W(S)$; the inequality follows since $A$ is $0.1$-near expanding in $W$. The total amount of flow which can be absorbed by $S$ can be upper bounded by the sink in $S$, which is at most $\bd_{G}(S)\le \vol_W(S)$. Thus, in the feasible flow $\Bf$, at least $4\cdot\vol_W(S)$ units of flow must be sent across the cut $(S,A\setminus S)$, so $\delta_{W'}(S,A\setminus S)\ge 4\cdot \vol_W(S)$. Combining with the fact that $\vol_{W'}(S)\le 102\cdot\vol_W(S)$, which we have already proven above, we have that
    $$\delta_{W'}(S,A\setminus S)\ge 4\cdot\vol_W(S)\ge 0.01\cdot \vol_{W'}(S).$$ Thus, $W'$ is a $0.0001$-out expander on $A$ which embeds into $G[A]$ with congestion $210/\phi$ via $\Pi'$, certifying that $G[A]$ is a $(\phi/10^7, \bd)$-out expander.
\end{proof}

\subsection{Dynamic Flows: the Push-Pull-Relabel Framework}\label{sec:ppl}

For our trimming procedure, we will need a modified version of the dynamic flow algorithm from~\cite{SP24}, which maintains a flow under source increases and vertex removals. We give a new analysis for our version of the Push-Pull-Relabel algorithm, since the analysis in~\cite{SP24} is limited to unweighted graphs. Specifically, their analysis generalizes that of Dynamic Unit-Flow~\cite{saranurak2019expander,henzinger2020local} such that the runtime scales with the total source of the flow instance. For weighted graphs, this will scale with the total capacity of the edges, which may be very large. Our new analysis enables our trimming algorithm to have running time independent of the edge capacities.

Now, we present the dynamic flow algorithm. The algorithm follows the Push-Relabel framework, where a flow-level pair $(\Bf,\Bell)$ is maintained at all times, with the following properties:
\begin{enumerate}
    \item for all edges $e=(u,v)$ satisfying $\Bell(u)>\Bell(v)+1$, we have $\Bf(e)=\Bc(e)$.
    \item for all nodes $u$ satisfying $\Bell(v)\ge 1$, we have that the sink at $u$ is saturated, i.e., $\abs_{\Bf}(u) = \nabla(u)$.
\end{enumerate} 
Our algorithm needs to support the following operations:
\begin{enumerate}
    \item $\textsc{IncreaseSource}(v, \delta)$: where $\delta\in\mathbb{N}^{V}_{\ge0}$, we set $\Delta(v) \leftarrow \Delta(v) +\delta$; and
    \item $\textsc{RemoveVertices}(S)$: where $S\subseteq \tilde{V}$, we set $\tilde{V}\leftarrow \tilde{V} \setminus S$ (initially $\tilde{V}=V$),
\end{enumerate}
The difficulty arises in the $\textsc{RemoveVertices}(S)$ operation. It is possible that there are flow paths in $\Bf$ going from $S$ into $\tilde{V} \setminus S$, for which the source is removed. When $\textsc{RemoveVertices}(S)$ is called, this may cause some vertex $v$ to have more flow sent out than it has received, since some source is removed. In this case, we say that vertex $v$ has \textit{negative excess}. Formally, $\ex^-_{\Bf}(v)= \max(0, -\Bf(v) - \Delta(v))$, where $-\Bf(v)$ is the net flow out of vertex $v$ (flow out minus flow in). With the introduction of negative excess, we need to generalize the definition of a valid flow state.
\begin{definition}
    The pair $(\Bf,\Bell)$ is said to be a \textit{valid state} if for every vertex $v\in V$, we have (1) $0< \ex^-_{\Bf}(v)$ implies $\Bell(v)=0$ and (2) $0< \ex^+_{\Bf}(v)$ implies $\Bell(v)=h$.
\end{definition}

Motivated by this and the Push-Relabel framework, \cite{SP24} introduced the subroutine \textsc{PullRelabel} to deal with negative excess. Negative excess at a vertex $v$ means that there is a flow path starting at $v$, but $v$ doesn't have extra flow to send out. Informally, we can think that negative excess at $v$ means that $v$ is in debt. In order to resolve this negative excess, we need to find a starting point $w$ for this flow path which has extra flow to send out. In other words, we want to find a vertex $w$ which can pay for $v$'s debt. We find this new starting vertex by ``pulling'' the negative excess around. Suppose there is some node $u$ which can send flow to $v$ (i.e., $\Bc_{\Bf}(u,v)>0$ and $\Bell(u)=\Bell(v)+1$). Even if $u$ doesn't have positive excess, we allow $u$ to send flow to $v$. This will result in some of the negative excess at $v$ being ``pulled'' to $u$, hence the name \textsc{PullRelabel}. The goal is that through these pull operations, we can eventually find a vertex $w$ which gets charged for $v$'s debt. Intuitively, using pull operations, we move the negative excess to vertices with a higher level, which tend to have positive excess. When finally we pull the negative excess to a node with positive excess, the negative excess is resolved. \Cref{alg:valid-state} formalizes this idea in the \textsc{ValidState} data structure.

\begin{algorithm}[h]
	\caption{\textsc{ValidState}$(G = (V,E, \bc) ,\Delta_0,\nabla,h)$}
	\label{alg:valid-state}
	\fbox{
		\parbox{0.97\columnwidth}{
            \textsc{Init}()\\
            \tab $\tilde{V}\leftarrow V$, $\Bc\leftarrow \Bc$, $\Delta\leftarrow\Delta_0$, $\nabla\leftarrow\nabla$, $(\Bf,\Bell)\leftarrow(\mathbf{0},\mathbf{0})$\\
            \tab \textsc{PushRelabel}()\\
            
            \textsc{IncreaseSource}$(v, \delta)$\\
            \tab $\Delta(v) \leftarrow\Delta(v)+\delta$\\
            \tab \textsc{PushRelabel}()\\
            
            \textsc{RemoveVertices}$(S)$\\
            \tab $\tilde{V}\leftarrow\tilde{V} \setminus S$\\
            \tab \textsc{PullRelabel}()\\
            \tab \textsc{PushRelabel}()\\
            
		  \textsc{PushRelabel}()\\
		  \tab \textbf{while} $\exists u$ where $\Bell(u)<h$ and $\ex^+_{\Bf}(u)>0$:\\
            \tab \tab \textbf{if} $\exists (u,v)$ such that $\Bc_f(u,v)>0$ and $\Bell(u)=\Bell(v)+1$\\
            \tab \tab \tab Send $\min(\ex^+_{\Bf}(u),\Bc_f(u,v))$ units of flow from $u$ to $v$ \hfill \textcolor{gray}{// Pushes positive excess from $u$ to $v$.}\\ 
            \tab \tab \textbf{else}:\\
            \tab \tab \tab $\Bell(u)\leftarrow\Bell(u)+1$ \hfill \textcolor{gray}{// Relabel $u$}\\
            
            \textsc{PullRelabel}()\\
            		  \tab \textbf{while} $\exists v$ where $\Bell(v)>0$ and $\ex^-_{\Bf}(v)>0$:\\
            \tab \tab \textbf{if} $\exists (u,v)$ such that $\Bc_f(u,v)>0$ and $\Bell(u)=\Bell(v)+1$\\
            \tab \tab \tab Send $\min(\ex^-_{\Bf}(v),\Bc_f(u,v))$ units of flow from $u$ to $v$ \hfill \textcolor{gray}{// Pulls negative excess from $v$ to $u$.}\\
            \tab \tab \textbf{else}:\\
            \tab \tab \tab $\Bell(v)\leftarrow\Bell(v)-1$ \hfill \textcolor{gray}{// Relabel $v$}
	}}
\end{algorithm}

For our analysis, we introduce two vectors $\Bp_t$ and $\Bn_t$ indexed by $v\in V$ which respectively represent the total amount of positive and negative units of flow located at $v$. These functions characterize the positive and negative excesses: we have $\ex^+_{\Bf}(v)=\max(\bp_t(v)-\bn_t(v),0)$ and $\ex^-_{\Bf}(v)=\max(\bn_t(v)-\bp_t(v)-\nabla(v),0)$ for all $v\in V$ and all timesteps $t$. At initialization, we have $\bp_0(v)=\Delta_0(v)$ and $\bn_0(v)=\nabla(v)$ for each $v\in V$ and evolves based on the operation. If the $t$\textsuperscript{th} action is:
\begin{enumerate}
    \item a push of $\psi$ units from $u$ to $v$, then $\bp_{t+1}(u)=\bp_t(u)-\psi$, $\bp_{t+1}(v)=\bp_t(v)+\psi$;
    \item a pull of $\psi$ units from $v$ to $u$, then $\bn_{t+1}(u)=\bn_t(u)+\psi$, $\bn_{t+1}(v)=\bn_t(v)-\psi$;
    \item $\textsc{IncreaseSource}(v, \delta)$, then $\bp_{t+1}(v)=\bp_{t}(v)+\delta(v)$; and
    \item $\textsc{RemoveVertices}(S)$, then
    \begin{align*}
        \forall v\in \tilde{V}_{t+1}&:\Bp_{t+1}(v)=\Bp_{t}(v)+\Bf(v,S),\Bn_{t+1}=\Bn_t(v)+\Bf(S,v)\\
        \forall v\not\in \tilde{V}_{t+1}&:\Bp_{t+1}(v)=\Bp_{t}(v),\Bn_{t+1}=\Bn_t(v)
    \end{align*}
\end{enumerate}
Let $\Bp_T(v)$, $\Bn_T(v)$ be the states of the vectors at termination. We can now state our algorithm's guarantees:

\begin{lemma}
    Let $(G=(V,E, \bc),\Delta_0,\nabla, h)$ be a flow problem with $\nabla(v)\ge\deg_G(v)$ for $v\in V$ that undergoes updates $\textsc{IncreaseSource}(v, \delta)$ and $\textsc{RemoveVertices}(S)$. Let $k$ be the total number of $\textsc{IncreaseSource}$ calls. There exists a (deterministic) data structure $\textsc{ValidState}$ that explicitly and dynamically maintains a valid state $(\Bf,\Bell)$ for the current flow instance $(G[\tilde{V}],\Delta|_{\tilde{V}},\nabla|_{\tilde{V}}, h)$. The algorithm runs in total time $O((k+n)\log n +h\log{n}\cdot \sum_{v\in V}\frac{\Bn_T(v)}{\nabla(v)}\cdot \widetilde{\deg}_G(v))$.\label{lem:push-pull-relabel}
\end{lemma}

\begin{remark}
    For the dynamic push-relabel setting where there is no negative excess, we simply have $\bn_T(v)=\nabla(v)$ since $\bn_T(v)=\bn_0(v)$ is never updated. This gives a runtime of $k\log{n}+mh\log{n}$, where $m$ is the number of edges. This is what is used in previous work for capacitated expander decompositions (see the discussion in Section 4.1 of~\cite{saranurak2019expander}), and is needed later in cut-matching game algorithm. Our analysis of the push-pull-relabel framework strictly generalizes that of dynamic push-relabel, unlike the previous analysis which only applies to unweighted graphs \cite{SP24}.
\end{remark}

The remainder of the section gives the proof of the lemma. We discuss some (standard) implementation details to formalize the runtime of the algorithm. For each vertex $v\in V$, we maintain a linked-list $L[v]$ containing all non-saturated edges $(v,u)$ satisfying $\Bell(v)=\Bell(u)+1$. With this linked-list, each push operation can be implemented in time $O(1)$. We can maintain $L[v]$ for each vertex $v$ by spending $O(\widetilde{\deg}_G(v))$ time for each relabel and $O(1)$ time for each push. We also maintain an analogous linked-list $L'[v]$ for pull operations.

The main difficulty of analyzing push-pull-relabel is that the labels can go up and down, so it is non-trivial to bound the number of relabels. For each vertex $v$, let $\Lambda(v)$ denote the \textit{flip number} of $v$, defined as the number of times the labels of $v$ flipped from increasing to decreasing, or vice versa. Formally, let $0=s_1<r_1<s_2<\ldots$ (either ending with $s_\lambda=T$ or $r_\lambda=T$) be a minimal partition of time such that $\Bell_t(v)\le \Bell_{t+1}(v)$ for $s_{i}\le t\le r_{i}$ and $\Bell_t(v)\ge \Bell_{t+1}(v)$ for $r_{i}\le t\le s_{i+1}$. Then $\Lambda(v)=2\lambda-2$ if the partition ends with $s_\lambda$ and $\Lambda(v)=2\lambda-1$ if the partition ends with $r_\lambda$. The notation $\Lambda(v)$ is chosen to be reminiscent of the graph of $\Bell_t(v)$ for $t \in \{0,1, \ldots, T\}$. See~\cref{fig: flip-diag}. Note that although the $s_i$'s, $r_i$'s and $\lambda$ technically depend on $v$, the associated $v$ will always be clear from context, so we suppress the connection for notational simplicity. We now bound the flip number of $v$.

\begin{figure}[h]
    \centering
    \input{diagrams/flip-diagram}
    \caption{A graph of the evolution of $\Bell(v)$ over time. The green regions are time intervals $s_i \leq t \leq r_i$ where the $\Bell(v)$ is non-decreasing. The red regions are $r_i \leq t \leq s_{i+1}$ where $\Bell(v)$ is non-increasing. The partition shown is minimal.}
    \label{fig: flip-diag}
\end{figure}
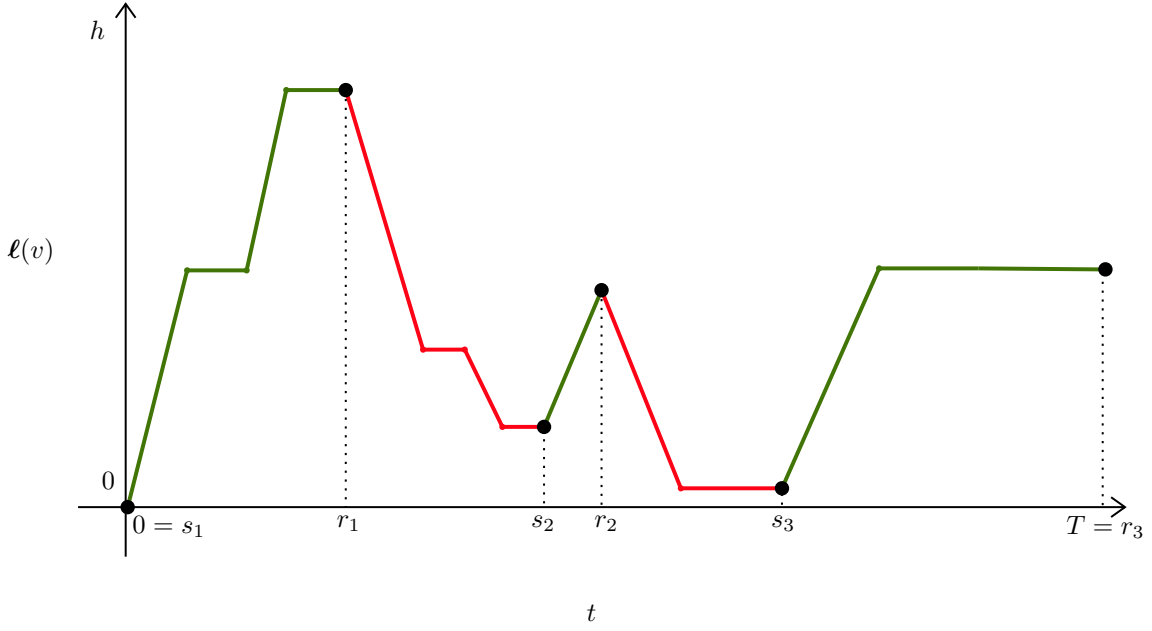

Informally, when a vertex $v$'s label changes from increasing to decreasing, this means that it changed from having positive excess to having negative excess. This implies that at least $\nabla(v)$ units of additional negative excess was added to $v$, so $\Bn_t(v)$ increases by at least $\nabla(v)$. Thus, the flip number should be upper bounded $\Bn_T(v)/\nabla(v)$. However, the claim does not immediately follow because the negative excess may later leave $v$, so $\Bn_t(v)$ may decrease. We now give the formal proof.

\begin{claim}\label{cl:flip-number-bound}
    For all $v \in V$, $\Lambda(v)\le O(\Bn_T(v)/\nabla(v))$. 
\end{claim}
\begin{proof}
    Since the partition is minimal, we are guaranteed that $\Bell_{s_i}(v)<\Bell_{r_i}(v)$ and $\Bell_{r_i}(v)> \Bell_{s_{i+1}}(v)$ for each $i\in[\Lambda(v)]$. 
    The former implies that $v$ was relabeled up at some iteration $t_1\in[s_i,r_i]$ and $t_3\in[s_{i+1},r_{i+1}]$, implying that $\Bn_{t_1}(v)\le \Bp_{t_1}(v)$ and $\Bn_{t_3}(v)\le \Bp_{t_3}(v)$ for such a $t_1,t_3$. The latter similarly implies that $v$ was relabeled down at some iteration $t_2\in[r_i,s_{i+1}]$, implying that $\Bp_{t_2}(v)+\nabla(v)\le \Bn_{t_2}(v)$ for such a $t_2$.
    Since we are only moving (via pushes/pulls) \emph{excess} positive and negative units, the function $\min(\Bp_{t}(v),\Bn_{t}(v))$ is non-decreasing. In particular, we have
    \begin{align*}
         \min(\Bp_{t_1}(v),\Bn_{t_1}(v))\le \min(\Bp_{t_2}(v),\Bn_{t_2}(v)).
    \end{align*}
    We know $\Bn_{t_1}(v)\le \Bp_{t_1}(v)$, and $\Bp_{t_2}(v)+\nabla(v)\le \Bn_{t_2}(v)$, so the left hand side is $\Bn_{t_1}(v)$ and right hand side is $\Bp_{t_2}(v)$, implying that $\Bn_{t_1}(v)\le \Bp_{t_2}(v)$.
    
    Since we only move away negative units when $\Bn_t(v)>\Bp_t(v)+\nabla(v)$, the function $\min(\Bp_t(v)+\nabla(v),\Bn_t(v))$ is also non-decreasing. In particular, we have
    \begin{align*}
        \min(\Bp_{t_2}(v)+\nabla(v),\Bn_{t_2}(v))\le \min(\Bp_{t_3}(v)+\nabla(v),\Bn_{t_3}(v)).
    \end{align*}
    We know that $\Bp_{t_2}(v)+\nabla(v)\le \Bn_{t_2}(v)$ and $\Bn_{t_3}(v)\le \Bp_{t_3}(v)$, so the left hand side is $\Bp_{t_2}(v)+\nabla(v)$ and right hand side is $\Bn_{t_3}(v)$, implying that $\Bp_{t_2}(v)+\nabla(v)\le \Bn_{t_3}(v)$. 

    Combining the two inequalities we concluded gives that 
    \begin{align*}
        \Bn_{t_3}(v)\ge \Bp_{t_2}(v)+\nabla(v)\ge \Bn_{t_1}(v)+\nabla(v).
    \end{align*}
    Inducting on the index of the partition, we know that $\Bn_{t^*}(v)\ge (i-1)\cdot \nabla(v)$ for some $t^*\in[s_i,r_i]$. Since for this $t^*$, we know $\Bp_{t^*}(v)\ge \Bn_{t^*}(v)$, we have that $\Bp_{t^*}(v)\ge (i-1)\cdot\nabla(v)$ as well. Finally, since $\min(\Bp_t(v),\Bn_t(v))$ is a non-decreasing function, we have that $\min(\Bp_T(v),\Bn_T(v))\ge (\lambda-1)\cdot\nabla(v)\ge (\Lambda(v)/2-1)\cdot\nabla(v)$. Hence, we can conclude $\Lambda(v)\le \min(\Bp_T(v),\Bn_T(v))/O(\nabla(v))\le O(\Bn_T(v)/\nabla(v))$.
\end{proof}

We use the bound on the flip number to bound the runtime of the pushes, pulls, and relabels.

\begin{claim}\label{cl:relabel-bound}
    The total number of relabels is upper bounded by $O\left(h\cdot \sum_{v\in V}\frac{n_T(v)}{\nabla(v)} \right)$. Hence, the total runtime of all relabels is upper bounded by $O\left(h\cdot \sum_{v\in V}\frac{n_T(v)}{\nabla(v)}\cdot \widetilde{\deg}_G(v)\right)$.
\end{claim}
\begin{proof}
    Each relabel of vertex $v$ is charged a runtime of $\widetilde{\deg}(v)$, so it suffices to prove that the number of times each vertex $v$ is relabeled is upper bounded by $h\cdot n_T(v)/\nabla(v)$. Since between any two flips in the direction of relabeling, there can be at most $h$ relabels, the total number of relabels for vertex $v$ is upper bounded by $h\cdot\beta(v)\le O(h\cdot n_T(v)/\nabla(v))$, as desired.
\end{proof}

To bound the number of pushes and pulls, we consider whether a push/pull saturates some edge. We say that a push/pull saturates an edge if after sending flow along edge $(u,v)$, we have that $\Bf(u,v)=\bc(u,v)$. We call such a push/pull a \textit{saturating} push/pull, and those that do not an \textit{unsaturating} push/pull. We will bound the number of saturating and unsaturating pushes/pulls separately.

\begin{claim} \label{cl:sat-bound}
    The total number of saturating pushes and pulls, and the total runtime charged to saturating pushes and pulls, is at most $O\left(h\cdot \sum_{v\in V}\frac{n_T(v)}{\nabla(v)}\cdot \widetilde{\deg}_G(v)\right)$.
\end{claim}
\begin{proof}
    Suppose we have a saturating push or a saturating pull, sending flow along an edge $(u,v)$. At this time, we have $\Bell(u)=\Bell(v)+1$. In order for another saturating push along $(u,v)$ to occur, flow must be sent back along $(v,u)$ (in the residual graph). For flow to be sent along $(v,u)$, we would need $\Bell(v)=\Bell(u)+1$. This implies that between two saturating pushes/pulls along edge $(u,v)$, there must be at least two relabels on $u$ and $v$ (in total). Thus, we can charge the number of saturating pushes/pulls along an edge $(u,v)$ to the total number of relabels on its endpoints (up to constant factors). Since each vertex $v$ is charged $O(\widetilde{\deg}_G(v))$ time by each of its incident edges and the number of times $v$ is relabeled is upper bounded by $h\cdot n_T(v)/\nabla(v)$ by \Cref{cl:flip-number-bound}, the desired bound follows.
\end{proof}

Without appealing to link-cut trees~\cite{sleator1981data, goldberg1988pushrel}, we lose an extra factor of $h$ in our bound on the number of unsaturating pushes/pulls. For completeness, we complete the proof without link-cut trees and then sketch how they can be applied in a standard way to optimize the running time.
\begin{claim}
    The total number of unsaturating pushes and pulls, and the total runtime charged to unsaturating pushes and pulls, is $O\left(kh+nh+h^2\cdot \sum_{v\in V}\frac{n_T(v)}{\nabla(v)}\cdot \widetilde{\deg}_G(v)\right)$.    
\end{claim}
\begin{proof}
    Consider the following potential function $\Phi=\Phi_1+\Phi_2$, where
    \begin{align*}
        \Phi_1=\sum_{v:\text{ex}^{+}(v)>0}\Bell(v)\qquad\Phi_2=\sum_{v:\text{ex}^{-}(v)>0}(h-\Bell(v)).
    \end{align*}
    For each unsaturating push along edge $(u,v)$, note that $\Phi_2$ cannot increase since no labels are changed and negative excess on vertices cannot increase via a push. Since the push along $(u,v)$ is unsaturating, we know $\text{ex}^+(u)=0$ after the push, so $\Phi_1$ decreases by $\Bell(u)$. But the push may cause $\text{ex}^+(v)$ to change to 0 to positive, so $\Phi_1$ may increase by at most $\Bell(v)$. In total, $\Phi_1$ decreases by at least $\Bell(u)-\Bell(v)=1$, so $\Phi$ also decreases by at least $1$ every unsaturating push. Similarly, every unsaturating pull along edge $(u,v)$ doesn't change $\Phi_1$ and decreases $\Phi_2$ by at least 1.

    For each relabel in \textsc{PushRelabel}, $\Phi_1$ increases by 1 and $\Phi_2$ remains unchanged. For each relabel in \textsc{PullRelabel}, $\Phi_2$ increases by 1 and $\Phi_1$ is constant. By our bound on the total number of relabels in \Cref{cl:relabel-bound}, we know that relabels increase $\Phi$ by at most $O(h\cdot\sum_{v\in V}\frac{n_T(v)}{\nabla(v)})$ throughout the algorithm. For each saturating push along edge $(u,v)$, again $\Phi_2$ does not increase. But since the push along $(u,v)$ is saturating, we do not necessarily push away all positive excess from $v$, so we may still have $\text{ex}^+(u)>0$. We can still increase the excess at $v$ from $0$ to $\text{ex}^+(v)>0$, so we may increase $\Phi_1$ by $\Bell(v)\le h$. Similarly, a saturating pull along edge $(u,v)$ doesn't increase $\Phi_1$ and may increase by at most $h-\Bell(v)\le h$. Since there are at most $O(h\cdot\sum_{v\in V}\frac{n_T(v)}{\nabla(v)}\cdot \widetilde{\deg}_G(v))$ saturating pushes and pulls total by~\cref{cl:sat-bound}, this increases $\Phi$ by at most $O(h^2\cdot\sum_{v\in V}\frac{n_T(v)}{\nabla(v)}\cdot \widetilde{\deg}_G(v))$. Finally, each increment of the source at a vertex $v$ via \textsc{IncreaseSource} increases $\Phi_1$ by at most $\Bell(v)\le h$. This increases $\Phi$ by at most $kh$.

    At the beginning of the algorithm, $\Phi$ is upper bounded by $O(nh)$. Throughout the algorithm, $\Phi$ can only increase by at most $O(kh+h^2\cdot\sum_{v\in V}\frac{n_T(v)}{\nabla(v)}\cdot \widetilde{\deg}_G(v))$ via relabels, saturating pushes/pulls, and source increases. Furthermore, removing vertices doesn't change the labels of the vertices and can only remove some of the nodes with positive excess, so it never increases the potential. Since the potential is non-negative and each unsaturating push/pull decreases the potential by at least $1$, there can be at most $O(kh+nh+h^2\cdot\sum_{v\in V}\frac{n_T(v)}{\nabla(v)}\cdot \widetilde{\deg}_G(v))$ unsaturating pushes/pulls throughout the algorithm.
\end{proof}

\paragraph{Optimizing the running time using link-cut trees.} 
To optimize the running-time of push-pull relabel further, we appeal to the link-cut tree data structure.

\begin{lemma}[Link-cut trees~\cite{sleator1981data, goldberg1988pushrel}] \label{lem: link-cut}
There exists a data structure which maintains a collection of rooted trees $\mathcal{T}$ on $n$ vertices with integral edge weights $\bw$. Let $\pi(u)$ denote the path from $u$ to its root for each vertex $u$. The data structure supports each of the following operations, all in $O(\log n)$ update time: 
\begin{itemize}
    \item $\link(e)$: for $e = (u,v)$, $u$ a root, and $v$ not in the same tree as $u$, add the edge $e$ and make $v$ the parent of $u$ in the tree of $v$. 
    \item $\cut(e)$: remove the edge $e$ from $\mathcal{T}$.
    \item $\fmin(u)$: returns $(e, \bw(e))$ for edge $e = \arg \min_{e \in \pi(u)} \bw(e)$.  If $u$ is a root, returns nothing. In the case of ties, returns the edge closer to the root. 
    \item $\add(u, \psi)$: For each $e \in \pi(u)$, set $\bw(e) \leftarrow \bw(e) + \psi$.
    \item $\froot(u)$: return the root of the tree that $u$ belongs to.
\end{itemize}
\end{lemma}

The algorithm maintains two link-cut tree data structures, $\mathcal{T}_+$ and $\mathcal{T}_-$, one for push and one for pull operations, with shared edge weights $\bw$. Equivalently, we have one data structure for handling the movement of positive excess and another for handling the movement of negative excess. Each is on vertex set $V$. 

We say that a vertex $v \in V$ is \textit{active} in $\mathcal{T}_+$ if $v$ has not been deleted (in a \textsc{RemoveVertices} step), $\ex_{\Bf}^+(v) > 0$ and $\Bell(v) < h$. Similarly, a vertex $v \in V$ is active in $\mathcal{T}_-$ if $v$ has not been deleted, $\ex_{\Bf}^-(v) > 0$ and $\Bell(v) > 0$. Throughout the algorithm, we maintain the following invariants for $\mathcal{T}_+$ and $\mathcal{T}_-$: 
\begin{enumerate}
    \item Active vertices in $\mathcal{T}_+$ are roots of trees in $\mathcal{T}_+$ and active vertices in $\mathcal{T}_-$ are roots of trees in $\mathcal{T}_-$. (Note: we do not require the converse, i.e., there may be roots that are themselves inactive.)
    \item Each node in $\mathcal{T}_+$ has at most one parent. The same holds for $\mathcal{T}_-$. 
    \item Let $e = (u,v)$ be an edge in $\mathcal{T}_+$ from child to parent. Then, $\Bell(u) = \Bell(v) + 1$. For $e' = (u',v')$ in $\mathcal{T}_-$ from child to parent, we have $\Bell(u') = \Bell(v') - 1$.
    \item $\bw(e) = \bw_{\Bf}(e)$, the residual capacity of the $e$ under the preflow $\Bf$. 
\end{enumerate}

We now sketch how to appropriately modify pushes, pulls, relabels, source increases, and vertex removals to accommodate $\mathcal{T}_+$ and $\mathcal{T}_-$. We modify the $\textsc{PushRelabel}$ as follows. Upon finding active vertex $u$ and admissible edge $(u,v)$, first we call $\link(u,v)$. Then, instead of just sending positive excess along the edge $(u,v)$, we call $\fmin(u)$, returning some pair, $(f, \bw_{\Bf}(f))$, where $f$ is the minimum capacity edge along the path from $u$ to the root of the tree containing $v$. We can then send positive excess from $u$ all the way to the root of the tree containing it by calling $\add(u, -\min(\ex^+_{\Bf}(u), \bw_{\Bf}(f)))$. Then, for edge $e'$ with $\bw_{\Bf}(e') = 0$, we call $\cut(e')$ in both $\mathcal{T}_+$ and $\mathcal{T}_-$. We then repeat this process until $u$ is again a root or $\ex^+_{\Bf} = 0$. Upon running out of admissible edges to push along, $u$ is relabeled. In addition to setting $\Bell(u) \leftarrow \Bell(u) + 1$, to maintain our link-cut tree invariants, we also call $\cut$ on every edge incident to $u$ in $G$ for both $\mathcal{T}_+$ and $\mathcal{T}_-$ (in particular, to preserve the third invariant).

The pull step in $\textsc{PullRelabel}$ is analogous, except in $\mathcal{T}_-$, and involving repeatedly sending negative excess to the root instead of positive excess. In particular, for active vertex $v$, upon finding an admissible edge $(u, v)$, we call $\link(u,v)$ and then repeatedly send negative excess from $v$ to its root, via $u$. This amounts to flow being sent from the root to $v$. Relabels also operate analogously, also cutting every edge incident to the relabeled vertex. 

To handle $\textsc{IncreaseSource}(v, \delta)$ operations, we also cut the at most one edge from $v$ to a parent in $\mathcal{T}_+$ (using our second invariant), since $v$ may now be active in $\mathcal{T}_+$.

To handle $\textsc{RemoveVertices}$ operations, we also cut each incident edge to the removed vertices in both $\mathcal{T}_-$ and $\mathcal{T}_+$.

Finally, we describe how to bound the total running time. The bounds on the number of saturating pushes/pulls and relabels are identical to before. However, we can prove a stronger bound on the number of unsaturating pushes/pulls; this is where our improvement in the running time comes from. Indeed, notice that unsaturating pushes/pulls only arise at the end of sequence of pushes/pulls following a $\link$ or $\textsc{IncreaseSource}$ operation. Hence, the number of unsaturating pushes/pulls is bounded by the total number of links and total number of times a vertex has its source increased, $k$. But, since we maintain a forest in each of $\mathcal{T}_+$ and $\mathcal{T}_-$, the number of total links is at most the number of total cuts plus $2n-2$. But, we can easily bound the number of total cuts: cuts come from saturating pushes, relabels, and calls to $\textsc{RemoveVertices}$. Combining these terms, we can conclude that the number of unsaturating pushes is at most $O(k + n + h\cdot \sum_{v\in V}\frac{n_T(v)}{\nabla(v)}\cdot \widetilde{\deg}_G(v) + m)$. Finally, accounting for the running time of calls to link-cut tree operations, the total running time is bounded by
\[
O((k + n) \log n + h \log n \cdot \sum_{v\in V}\frac{n_T(v)}{\nabla(v)}\cdot \widetilde{\deg}_G(v)), 
\]
as required.

\subsection{Trimming: Finding an Expander in a Near-Expander}\label{sec:trim-algorithm}

In this section, we prove our trimming algorithm.
\trimming*
On a high level, the algorithm maintains two flow instances via the push-pull-relabel framework, one on $G$ and one on $\cev{G}$ ($G$ with all edges reversed), to certify out- and in-expansion, respectively. At each iteration and for each of the flow instances, the dynamic flow algorithms will either output a feasible flow or a sparse cut. If both flow instances have a feasible flow, then we have successfully certified expansion for the subgraph. If one of the flow instances finds a sparse cut $S$, then we remove $S$ from the current subgraph $A$ in \emph{both} flow instances. This will also induce some new sources. We then use push-pull-relabel to dynamically find a flow in the new subgraph $A \setminus S$. We will show that the total volume of vertices cut off is bounded, implying that the expander which we end up finding is indeed large.

\begin{algorithm}[h]
	\caption{Trimming: finding a large expander}
	\label{alg:trimming}
	\fbox{
		\parbox{0.97\columnwidth}{
		  \textsc{Trim}$(G, A, W, \Pi, \phi)$ \\
            \tab Set $h=40000\log{nW}/\phi$ and $A_0=A$\\ 
		  \tab Initialize the data structure \textsc{ValidState} for the two flow instances\\
            \tab \textbf{While} one of the two flow instances is not feasible:\\
            \tab \tab Find a sparse cut $S_t\subseteq A_t$ via \Cref{cl:sparse-level-cut}\\
            \tab \tab Update the two flow states via \Cref{alg:updating-flow-instance}\\
            \tab \tab \textbf{If} $\bd(\bigcup_{j=0}^{t}S_t)\ge \bd(V)/(c_0\log(n)\log(nW))$\hfill \textcolor{gray}{// Early Termination}\\
            \tab \tab \tab Terminate the algorithm and return $(S_j)_{j=0}^{t}$ as the sequence of sparse cuts\\
            \tab Return $A_t$ as the certified expander, and $(S_j)_{j=0}^{t-1}$ as the sequence of sparse cuts
	}}
\end{algorithm}

Now, we can define our flow instance as in \Cref{sec:flow}, with parameter $\phi$. Let $(W,\Pi)$ be the witness obtained in the non-stop cut-matching game (\Cref{cl:witness-from-CMG}). For each $(u,v)=e\in E(W)$ such that its embedding path $\Pi(e)$ crosses the boundary $E(A,V\setminus A)$ of $A$, we add a source of $100 \cdot \bc_W(e)$ units on the endpoints $u$ and $v$ that are in $A$. We then add a sink of $\bd(w)$ units at each vertex $w \in A$ and set the capacity of each edge $e \in G[A]$ to be $200\cdot \bc_G(e)/\phi$. During the trimming process for the near expander $A_0=A$, let the current set be $A_t$; we will always update our flow instance on $A_t$ so that it is of the form from \Cref{sec:flow}.

For our two flow instances, we show that Push-Pull-Relabel either finds a feasible flow or we can find a sparse level cut $S=\{v:\Bell(v)\ge k\}$ for some $k$. This is a Push-Pull-Relabel analogue of~\Cref{lem: bdd ht push relabel} (or, e.g., Lemma B.6 of~\cite{saranurak2019expander}).

\begin{claim}\label{cl:sparse-level-cut}
     Let $(\Bf,\Bell)$ be the valid state maintained by Push-Pull-Relabel. We can either certify that the flow is feasible or find a level cut $S_t=\{v:\Bell(v)\ge k\}$ for some $k\ge1$ such that 
     \[
     \min(\delta_{G}(S_t, A_t \setminus S_t), \delta_{G}(A_t \setminus S_t, S_t)) = O(\phi \log (n) \log (nW) \bd(S_t)).
     \]
\end{claim}
\begin{proof}
We mimic the proof of~\cref{lem: bdd ht push relabel}. Let $H$ denote the current flow instance graph, and let $H_{con}$ be the edge subgraph of the flow instance, only retaining edges between consecutive levels. For $k \in \N$, let $L_k = \{v \,:\, \Bell(v) = k\}$.  If $L_h=\emptyset$, then no vertex has positive excess. Note that there can still be vertices with negative excess, so in this case $\Bf$ is not necessarily a feasible flow.  But, we can easily obtain a feasible flow from $\Bf$. Consider a path decomposition of the flow, treating negative excess as source. Then, simply remove the flow paths corresponding to negative excess sources. 

Now, for any $i$, an edge $(u,v)$ crossing the cut $(L_{\geq i}, L_{<i})$, one of the following cases must hold: (1) $\Bell(u)=i$ and $\Bell(v)=i-1$ or (2) $\Bell(u)-\Bell(v)>1$. Note that by the invariant of Push-Pull-Relabel, all edges of the second type must be saturated in the down-hill direction (i.e., $\Bf(u,v)=200\bc_G(u,v)/\phi$). 
 By the same ball-growing argument as in~\cref{lem: bdd ht push relabel}, we can find $k$ such that
\[
\delta_{H_{con}}(L_{\geq k}) \leq \frac{1}{10} \min( \bd(L_{< k}), \bd(L_{\geq k})).
\]
In particular, then by our invariant of Push-Pull-Relabel that edges skipping levels in the downward direction are saturated, we have that
\begin{align*}
\Bf(L_{\geq k}, L_{< k}) &\geq \delta_{H}(L_{\geq k}, L_{< k}) - \frac{1}{10} \min( \bd(L_{< k}), \bd(L_{\geq k}))  \\
&\geq \frac{200}{\phi} \min(\delta_{G}(L_{\geq k}, L_{< k}), \delta_{G}( L_{< k}, L_{\geq k})) - \frac{1}{10} \min( \bd(L_{< k}), \bd(L_{\geq k})).
\end{align*}
Note that here we use
\[
\delta_{H}(L_{\geq k}, L_{< k}) \geq \frac{200}{\phi} \min(\delta_{G}(L_{\geq k}, L_{< k}), \delta_{G}( L_{< k}, L_{\geq k}))
\]
as opposed to 
\[
\delta_{H}(L_{\geq k}, L_{< k}) \geq \frac{200}{\phi} \delta_{G}(L_{\geq k}, L_{< k}).
\]
since the flow instance we are considering could in fact be in $G$ with its edges reversed.

But, on the other hand,
\[
\Bf(L_{\geq k}, L_{< k}) \leq \Delta(L_{\geq k}) \leq \vol_W(L_{\geq k}) = O(\log (nW) \log (n) \bd(L_{\geq k})),
\]
where $W$ is the witness of near-expansion and the bound on its degrees follows from~\cref{cl:witness-from-CMG}. (Note that negative excess cannot contribute toward $\Bf(L_{\geq k}, L_{< k})$ since all nodes with negative excess are at level $0$.) Combining with the lower bound on $\Bf(L_{\geq k}, L_{< k})$ yields the proof of the first part of the claim.
\end{proof}

In between any two iterations of trimming, we update $(\Bf,\Bell)$ and the source based on the cut $S_t$ (for both flow instances). To update $(\Bf,\Bell)$, we simply take the restriction of the functions on $A_{t+1}=A_t\setminus S_t$; this is done by calling $\textsc{RemoveVertices}(S_t)$. To update the source, for each edge $e\in E(S_t,A_{t+1})$ on the cut $S_t$, we look at all edges $e'\in E(W)$ for which the embedding path $\Pi(e')$ uses $e$, and add a source at the endpoints of $e'$ which are in $A_{t+1}$ of 100 times the capacity of embedding path.\footnote{Implementing this step efficiently is actually a major technical detail in capacitated graphs and is the focus of~\cref{sec: fast-path-decomp}, but we ignore this issue for our current discussion.} We continue this process via Push-Pull-Relabel until we find a feasible flow.

\begin{algorithm}[h]
	\caption{Operations between times $t$ and $t+1$}
	\label{alg:updating-flow-instance}
	\fbox{
		\parbox{0.97\columnwidth}{
		  \textsc{UpdateState}()\\
            \tab Maintain a path decomposition $\mathcal{P}_{\Pi}$ of the embedding paths $\Pi$\\
		  \tab Call $\textsc{RemoveVertices}(S_t)$ on both flow states; $A_{t+1}\leftarrow A_t \setminus S_t$\\
            \tab \textbf{For} each edge $e\in E_G(S_t,A_{t+1})\cup E_G(A_{t+1},S_t)$\\
            \tab \tab \textbf{For} each path $P\in\mathcal{P}_{\Pi}$ which uses $e$:\\
            \tab \tab \tab Let $\delta$ be the amount of flow sent along $P$ and let $u,v$ denote the endpoints of $P$\\
            \tab \tab \tab Call \textsc{IncreaseSource}$(u, 100\delta)$ and \textsc{IncreaseSource}$(v, 100\delta)$ for both flow instances
	}}
\end{algorithm}

Finally, we need to guarantee that none of the $S_t$ we trim is too large. This guarantee will, in turn, imply that the expander we find is relatively large. To achieve this, we need to ensure that our current source  is always relatively small. We do this via another early termination condition: if $\bd(\bigcup_{j=0}^tS_t)$ ever exceeds $\bd(V)/(c_0\log(n)\log(nW))$, we terminate the trimming step and return the sequence of sparse cuts $(S_j)_{j=0}^t$. As in the analysis of cut-matching, this is sufficient for limiting the recursion depth of the algorithm; all recursive components are a $1 - \Omega(1/\log(n)\log(nW))$ factor smaller.

We bound the total amount of source at any iteration of the algorithm.
\begin{claim}\label{cl:source-bound}
    At any iteration $t$, if we haven't reached the early termination condition, we have that the total source every added is upper bounded by $\bd(A_t)/100$.   
\end{claim}
\begin{proof}
    There are two cases for incrementing source: (1) the source is induced by a flow path where one endpoint does not lie in $A_t$, or (2) the source is induced by a flow path where both endpoints lie in $A_t$ (but the flow path crosses out of $A_t$). Let $\Delta^{(1)}_t(A_t)$ and $\Delta^{(2)}_t(A_t)$ denote the amounts of source induced by (1) and (2), respectively. We will show that $\Delta^{(1)}_t(A_t)\le \bd(V)/1000$ and $\Delta^{(2)}_t(A_t)\le \bd(V)/1000$. Since $\bd(A_t)\ge \bd(V)/2$ by the early termination conditions in \Cref{lem: cut-matching} and \Cref{alg:trimming}, we would then have the desired result.

    First, we bound the amount of source added from case (1). Recall that for each $v \in V$, we have $\deg_W(v) \leq O(\log(n)\log(nW))\cdot\bd(v)$. This is because there are $T=O(\log(n)\log(nW))$ rounds of cut-matching and each contributes at most $\bd(v)$ to the degree of $v$ in $W$. This implies that 
    \begin{align*}
        \deg_W(V\setminus A_t)&=O(\log(n)\log(nW))\cdot\bd(V\setminus A_t)\\
        &=O(\log(n)\log(nW))\cdot\left[\bd(\bigcup_{j=1}^{k}C_j)+\bd(\bigcup_{j=0}^{t-1}S_j)\right].
    \end{align*}
    By the early termination conditions in \Cref{lem: cut-matching} and \Cref{alg:trimming}, we know that $\bd(\bigcup_{j=1}^{k}C_j)+\bd(\bigcup_{j=0}^{t-1}S_j)\le \bd(V)/(c_0\log(n)\log(nW))$. If we choose $c_0$ to be a large enough constant, we then have that $\deg_W(V\setminus A_t)\le \bd(V)/1000$. Then, each source in case (1) is induced by an embedding path with an endpoint in $V\setminus A_t$. This implies that $\Delta^{(1)}_t(A_t)\le \deg_W(V\setminus A_t)\le \bd(V)/1000$, as desired.

    Next, we bound the amount of source added due to case (2). Since both endpoints of the embedding paths $P$ inducing this source lie in $A_t$ but the embedding path still crosses out of $A_t$, $P$ crosses in and out of the boundary at some $e\in E(A_t,V\setminus A_t)$ and $e'\in E(V\setminus A_t,A_t)$. Although the cut $(A_t,V \setminus A_t)$ is not necessarily a sparse cut, we know that $V \setminus A_t$ is partitioned into a sequence of sparse cuts $C_1\sqcup\cdots\sqcup C_k$ and $S_0\sqcup\cdots S_{t-1}$, such that for some constant $c_3$, we have
    \[
    \Phi_{G[V\setminus C_{<j}],\bd}(C_j)\le c_1\phi\log(n)\log(nW)
    \]
    and 
    \[
    \Phi_{G[A\setminus S_{<j}],\bd}(S_j)\le c_3\phi\log(n)\log(nW).
    \]
    For simpler notation, let us define $C_{k+j+1}=S_j$ and $k'=k+t$ so that $C_1,\ldots,C_k,C_{k+1},\ldots,C_{k'}$ is our sequence of sparse cuts. We want to charge $e$ and $e'$ to two edges $f$ and $f'$ in one of the sparse cuts such that $f\in E(V\setminus C_{\le j},C_j)$ and $f'\in E(C_j,V\setminus C_{\le j})$ for some $j\in[k']$. If we are able to do this, then we can bound the total amount of source added by the total capacity (measured in the flow instance) of the sparse direction of all the sparse cuts $(C_j,V\setminus C_{\le j})$, which is small.

    To charge $e$ and $e'$, let $P[e,e']$ denote the subpath of $P$ between $e$ and $e'$. If $P[e,e']$ does not cross through any other sparse cut in the sequence of sparse cuts, then $e$ and $e'$ must both lie on the boundary of $A_t$ at $(C_{k'},A_t)=(C_{k'},V\setminus C_{\le k'})$ since $P[e,e']$ must enter and leave the set $C_{k'}$. In this case, we simply charge $e,e'$ to themselves (as edges on the sparse cut $(C_{k'},V\setminus C_{\le k'})$). Otherwise, suppose $P[e,e']$ crosses through at least one sparse cut $(V\setminus C_{\le j},C_j)$ for $j<k'$. We claim that for some $j$, $P[e,e']$ crosses the sparse cut $(V\setminus C_{\le j},C_j)$ in both directions (once going in and once going out). In particular, this is the case for the sparse cut $(V\setminus C_{\le j},C_j)$ corresponding to the smallest $j$ which $P[e, e']$ crosses at least once. This follows since after entering $C_j$, we must also leave $C_j$ because the path $P[e,e']$ ends in $A_t$. We can leave $C_j$ by either crossing through the other direction $E(C_j,V\setminus C_{\le j})$ or crossing through the boundary of $V\setminus C_{<j}$ into an earlier sparse cut. The latter cannot happen since $(V\setminus C_{\le j},C_j)$ is the smallest $j$ which $P[e,e']$ crosses at least once, so the former must happen and we have that $P[e,e']$ crosses the cut $(V\setminus C_{\le j},C_j)$ in both directions. For such a cut $C_j$ where we cross through the sparse cut in both directions, let $f\in E(V\setminus C_{\le j},C_j)$ and $f'\in E(C_j,V\setminus C_{\le j})$ denote two edges in $P[e,e']$ going in and out of the cut, respectively. We charge $e,e'$ to $f,f'$. See~\cref{fig: source-charge}.

    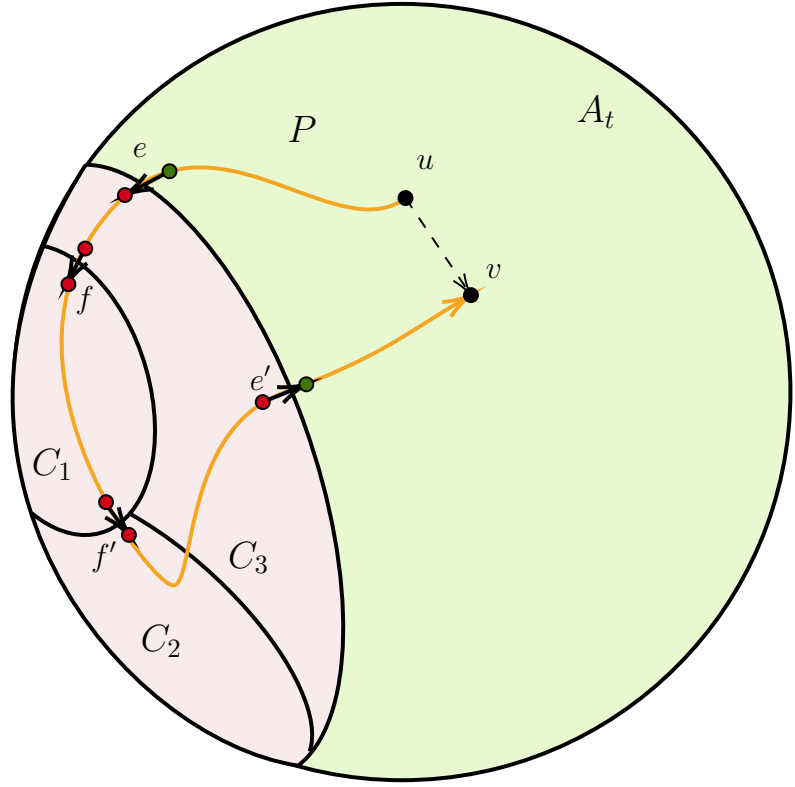
\begin{figure}[h]
        \centering
\input{diagrams/source-charge}
        \caption{Depiction of the edge-charging argument for a Path $P$ leaving $A_t$ via $e$ and returning via $e'$. $P$ crosses in and out of $C_1$ via $f$ and $f'$, and $1$ is the smallest index of a cut the subpath $P[e,e']$ passes through. Although $P$ also crosses in and out of $C_2$, the edge into $C_2$ starts in $C_1$. Hence we cannot charge to edges crossing in and out of $C_2$ since $C_2$ is only necessarily sparse in $G[V \setminus C_1]$.}
        \label{fig: source-charge}
    \end{figure}

    Note that in this charging scheme, each edge $e$ in the sparse cuts is charged at most $10\bc(e)/\phi$. This follows since we only charge to $e$ when some embedding path crosses through it, and the congestion of the embedding $\Pi$ is $10/\phi$. Furthermore, in the charging scheme, we charge the same amount to the in-edges and out-edges of each sparse cut. This implies that the upper bound on the source is a function of the sparse directions of each sparse cut. Formally, we have that 
    \begin{align*}
        \Delta_t^{(2)}(A_t)&\le (20/\phi)\cdot\sum_{j=1}^{k'}\min(\delta_G(C_j,V\setminus C_{<j}),\delta_G(V\setminus C_{<j},C_j))\\
        &\leq (20/\phi)\cdot\sum_{j=1}^{k'}\Phi_{G[V\setminus C_{<j}],\bd}(C_j)\cdot \bd(C_j)\\
        &\le 20(c_1+c_3)\log(n)\log(nW)\cdot \sum_{j=1}^{k'}\bd(C_j)\\
        &\le 40(c_1 + c_3)\cdot\bd(V)/c_0.
    \end{align*}
    The first inequality follows by the charging argument, the first equality follows by the definition of conductance, the second inequality follows by the bound on $\Phi_{G[V\setminus C_{<j}],\bd}(C_j)$ in \Cref{lem: cut-matching} and \Cref{cl:sparse-level-cut}, and the final inequality follows by the early termination condition in \Cref{lem: cut-matching} and \Cref{alg:trimming}. Finally, we choose $c_0$ to be large enough so that both $\Delta_t^{(1)}(A_t)$ and $\Delta_t^{(2)}(A_t)$ are at most $\bd(V)/1000$.
\end{proof}

Finally, we bound the total runtime of the algorithm.

\begin{claim}\label{cl:trim-runtime}
    The runtime of \Cref{alg:trimming} is $O(m \log^2 (n) \log (nW) + m\log(n)\log(nW)/\phi)$.
\end{claim}
\begin{proof}
    In our flow instance, the sink satisfies
    \begin{align*}
        \nabla(v)=\deg_G(v)+\frac{\widetilde{\deg}_G(v)}{2m}\cdot\deg_G(V)\ge \frac{\widetilde{\deg}_G(v)}{2m}\cdot\deg_G(V).
    \end{align*}
    Substituting this bound in and summing over $v\in V$, we have that
    \begin{align*}
        \sum_{v\in V}\frac{\Bn_T(v)}{\nabla(v)}\cdot\widetilde{\deg}_G(v)\le 2m\sum_{v\in V}\frac{\Bn_T(v)}{\deg_G(V)}=2m\cdot\frac{\|\Bn_T\|_1}{\deg_G(V)}.
    \end{align*}
    To bound $\|\Bn_T\|_1$ (the sum of the entries of $\Bn_T$), first observe that 
    \begin{align*}
        \|\Bn_0\|_1=\sum_{v\in V}\nabla(v)=\sum_{v\in V}\frac{\widetilde{\deg}_G(v)}{2m}\cdot\deg_G(V)+\sum_{v\in V}\deg_G(v)=2\deg_G(V),
    \end{align*}
    Throughout the algorithm, $\|\Bn_t\|_1$ only increases after calls to $\textsc{RemoveVertices}(S)$. But throughout the entire process, the total amount of negative excess added is upper bounded by the total amount of source ever added to the flow instances, which we bounded in \Cref{cl:source-bound} by $O(\bd(V))=O(\deg_G(V))$. Applying the runtime bound from~\cref{lem:push-pull-relabel} gives a total runtime of $O(k\log n+mh\log n)$, where $k$ is the total number of $\textsc{IncreaseSource}$ calls. Then, each IncreaseSource call corresponds to adding source at the endpoints of some edge in $W$. But, $W$ has at most $O(m \log (n) \log (nW))$ total edges, since the path decompositions computed in each matching step have at most $m$ edges. This then yields the desired runtime bound.
\end{proof}

We now have all of the necessary ingredients to prove our main result on trimming.

\begin{proof}[Proof of \Cref{lem:trimming}.]
     At any iteration $t$ prior to early termination, the total source ever added to the flow instance is upper bounded by $\bd(V)/100$ by \Cref{cl:source-bound}. Hence, we have that $\bd(S_t)\le\bd(V)/100$, since the level cut in \Cref{cl:sparse-level-cut} contains vertices which have absorbed their entire sink. Since the total source is bounded by $\bd(A_t)/100$, we have that $\bd(S_t)\le \bd(A_t)/100$ as well. Note that this does not trivially follow from the early termination condition since naively, the final cut $S_t$ found which triggered the early termination may have large $\bd(S_t)$; we showed that this is not the case. In the case of finding an expander $A'$, we know that $\bd(A')\ge \bd(V)/2$ by the early termination condition. By \Cref{cl:flow-instance}, we have that $G[A']$ is a $(\phi/10^7,\bd)$-expander since both flow instances are feasible at termination. Finally, \Cref{cl:trim-runtime} bounds the runtime.
\end{proof}

\subsection{Efficient Implementation: Flow Path Decompositions}
\label{sec: fast-path-decomp}

There is a challenge related to source updates that we have thus far ignored. Whenever we find a level cut $S$ in push-pull relabel, for each edge $e = (u,v) \in E(W)$ whose path embedding $\Pi(e)$ crosses through $S$, we want to add $100 \cdot \bc_W(e)$ units of source to $u$ and $v$ by calling $\textsc{IncreaseSource}$ (at least for those of $u,v$ in the remaining graph). However, to make these  $\textsc{IncreaseSource}$ calls, we need to know which $e \in E(W)$ have $\Pi(e)$ crossing into $S$. In the unweighted case, we can maintain our path decompositions from the matching steps of cut-matching explicitly since the paths are edge disjoint, giving the following lemma: 

\begin{lemma}[Trimming for Unweighted Graphs]\label{lem:trimming-unweighted} 
    Let $G=(V,E)$ be an unweighted graph (with binary edge weights). Let $A\subseteq V$ be a subset which is $(\phi,\bd)$-nearly expanding in $G$ such that $V\setminus A$ is the disjoint union of a sequence of sparse cuts $C_1,\ldots,C_k$, where $\Phi_{G[V\setminus C_{<j}],\bd}(C_j)\le c_1\phi \log^2(n)$ for each $j\in[k]$. There exists an algorithm which finds a subset $A'\subseteq A$ along with a sequence of cuts $S_0\sqcup \cdots\sqcup S_t=A\setminus A'$ satisfying $\Phi_{G[A\setminus S_{<j}],\bd}(S_j)\le c_3\phi\log(n)\log(nW)$ such that either
    \begin{enumerate}
        \item \textbf{Early termination:} $\bd(\bigcup_{j=0}^{t}S_j)\ge \bd(V)/(c_0\log^2(n))$
        \item \textbf{Certifies expansion:} $A'$ is certified to be a $(\Omega(\phi),\bd)$-expander and $\bd(A')\ge \bd(A)/2$.
    \end{enumerate}
    The algorithm runs in time $O(m\log^2(n)/\phi)$.
\end{lemma}

However, in the weighted case, our representation must be implicit, since the flow decomposition size lower bound of $\Omega(mn)$ applies. In this section, we show how to extend the analysis from~\Cref{sec:trim-algorithm} to capacitated graphs. In particular, we show how to deduce the following.
\trimming*

\paragraph{Determining when embbeded flow paths cross a cut.} As outlined above, we need to be careful about how we determine whether embedded flow paths pass through $S$. Nonetheless, it turns out that, for a fixed $S$, we can relatively easily determine the $e \in E(W)$ such that $\Pi(e)$ passes through $S$ in $O(m \log^2 (n) \log (nW))$ time. To do this, recall where these path embeddings originated from: we solved some flow problem in a matching step of the cut-matching game and computed a path decomposition of the flow using link-cut trees. 

Suppose that, along with computing the path decomposition, we record a transcript of the link-cut tree algorithm. The total algorithm takes $O(m \log n)$ time, so the transcript is itself succinct. Now, suppose that we add a tag to each edge to mark whether one of the endpoints of the edge is in $S$. One way to view this is as a second weight function on the edges with the weight being $0$ for edges with an endpoint in $S$ and $1$ otherwise. (Say that $\fmin'$ is $\fmin$ for this secondary weight function.) Then, it is easy to determine whether a path from $u \in V$ to its root involves an edge with an endpoint in $S$ using a single $\fmin'(u)$ call (on the secondary weight function). Now, to determine the set of $e \in E(W)$ such that $\Pi(e)$ passes through $S$, all we need to do is determine which paths in the path decomposition pass through $S$. But this is simple: at the cost of only a constant factor loss in the overall-running time, insert a $\froot(u)$ and $\fmin'(u)$ calls prior each to each $\add(u, \psi)$ call. The call to $\froot(u)$ means we know exactly the endpoints of the path in the path decomposition (and hence the corresponding edge in $E(W)$) and the call to $\fmin'(u)$ determines exactly whether the embedding of this edge passes through $S$.

Then, the overall running time is $O(m \log^2 (n) \log (nW))$ since we need to repeat this for each path decomposition algorithm run throughout the cut-matching algorithm. (There are $T = O(\log (n) \log (nW))$ rounds of the cut-matching game and so $O(\log (n) \log (nW))$ path decomposition algorithm calls.) 

\paragraph{Limiting the number of costly $\textsc{IncreaseSource}$ calls.}
Now that we have a way to efficiently determine whether embedded flow paths cross a cut, we need to ensure that we do not have to call this algorithm too many total times; indeed, the running time is $\tilde{O}(m)$ so we can only afford a polylogarithmic number of total calls. 

A simple strategy to attempt to achieve this is to only perform $\textsc{IncreaseSource}$ updates in large batches, and at that time running the aforementioned algorithm for determining whether embedded paths cross the computed cut(s) and increasing sources accordingly. Indeed, consider the following algorithm: instead of calling $\textsc{IncreaseSource}$ after each discovered cut, run push-pull relabel until both flows are feasible (or the early termination case is reached). At that point, handle all of the $\textsc{IncreaseSource}$ updates. Then, repeat this process until both flows are feasible and there are no additional $\textsc{IncreaseSource}$ updates to make. The same analysis as before handles the running-time of the algorithm other than the time to compute $\textsc{IncreaseSource}$ updates. 

To bound the number of batches of $\textsc{IncreaseSource}$ updates, we want to show that the total amount of unabsorbed source in the flow instances is decreasing by a constant factor between each batch of updates. But, we cannot quite show this with our current choice of parameters in the algorithm. E.g.,~\cref{cl:sparse-level-cut} is slightly too weak, showing that, upon finding a level cut, we remove almost $\bd(S_t)$ source, but we might add as much as $\Omega(\log (n) \log (nW) \bd(S_t))$ source back in. 

Instead, we combine two new strategies. First, we want to strengthen the guarantees of~\cref{cl:sparse-level-cut}: we want to ensure that we add less source than we remove. To do this, we scale up the sink of each vertex in the flow instance and the capacities of the edges in the flow instance. We scale up the capacities from $\bd(u)$ to $10C \log (n) \log (nW) \bd(u)$ for each $u \in A$ and from $200 \cdot \bc(e)/ \phi$ to $200 C \cdot  \log (n) \log (nW) \bc(e)/ \phi$ for each $e \in G[A]$, where $C \log (n) \log (nW) \bd(u)$ is an upper bound on $\deg_W(u)$. Then, by following the same analysis as in~\cref{cl:sparse-level-cut}, each cut $S_t$ removes $10 C \log (n) \log (nW) \bd(S_t)$ source while ultimately incurring an addition of at most $2C \log (n) \log (nW) \bd(S_t)$ source addition in the subsequent batched $\textsc{IncreaseSource}$ step. That is, at the end of the round of push-pull relabel (prior to the next batch of $\textsc{IncreaseSource}$ updates), the non-absorbed source decreased by a factor of $5$. 

We are not quite done: we need to be slightly careful here. If after our batched $\textsc{IncreaseSource}$ step we just continue applying push-pull relabel with the same levels and sinks as before, we might find cuts where most of the absorbed source actually comes from previous rounds of $\textsc{IncreaseSource}$ updates. In that case, we are actually adding non-absorbed source without removing any. But~\cite{chen2025parallel} introduce a beautifully simple approach to address this very issue and ensure that the algorithm terminates in few rounds: gradually increase the sink for each vertex across iterations. 

To make this rigorous, we cannot just increase the source for each vertex, apply our batched source increase, and then continue iterating push-pull relabel. Indeed, increasing the sink of vertices clearly disrupts the invariants of push-pull relabel (e.g., there are then vertices at level greater than $0$ without their source fully saturated). Instead, we rely on the fact that these sink increases ensure that the algorithm terminates in $O(\log nW)$ total rounds. As such, we can afford to start from ``scratch'' on each round. To that end, we reset the levels while retaining the same flow, incrementing the sources, and incrementing the sinks by an additional $10C \log (n) \log (nW) \bd(u)$. 

\paragraph{Putting it all together.}
We summarize how all the modifications we described above alter the conclusions of our analysis. Determining when embedded flow paths cross a cut takes $O(m \log^2 (n) \log^2 (nW))$ time in total, since there are $O(\log(nW))$ rounds of $\textsc{IncreaseSource}$. Besides that, we run Push-Pull-Relabel $O(\log nW)$ times, with edge-capacities $O(\log (n) \log (nW) \bc(e)/\phi)$ for each $e \in G[A]$ and sinks scaled up by $O(r\cdot\log(nW) \log n)$ in the $r$\textsuperscript{th} call to Push-Pull-Relabel. Via the same analysis as before, this contributes a total runtime of $O(m \log(nW)/\phi)$. Note that since we scale up the sinks by at least $\Theta(\log(n)\log(nW))$ on each round, $\nabla(v)$ is scaled by the same factor, decreasing the runtime. $\|\bn_T\|_1$ remains bounded by the total source, which is still bounded by $\bd(V)$. The cuts $S_1,\ldots,S_t$ are now $O(\phi)$-sparse because the edge capacities in the flow instance are scaled up by a factor of $\Omega(\log(n)\log(nW))$ and the bound on the total source (and thus the total flow through the cut) is the same as before. Finally, since we scale up the sinks in the flow instance by $\Theta(\log(n)\log^2(nW))$ in the final round of $\textsc{IncreaseSource}$ and the edge capacities are scaled up by $\Omega(\log(n)\log(nW))$ throughout, we only certify that $A'\subseteq A$ is an $(\Omega(\phi/\log^2(n)\log^3(nW)),\bd)$-expander.

%% file: diagrams/flip-diagram.tex
\tikzset{every picture/.style={line width=0.75pt}} 

\begin{tikzpicture}[x=0.75pt,y=0.75pt,yscale=-1,xscale=1]

\draw  (56,270) -- (584,270)(80,16) -- (80,295) (577,265) -- (584,270) -- (577,275) (75,23) -- (80,16) -- (85,23)  ;
\draw [color={rgb, 255:red, 65; green, 117; blue, 5 }  ,draw opacity=1 ][line width=1.5]    (81,270) -- (111,150.5) ;
\draw [color={rgb, 255:red, 65; green, 117; blue, 5 }  ,draw opacity=1 ][line width=1.5]    (111,150.5) -- (141,150.5) ;
\draw [color={rgb, 255:red, 65; green, 117; blue, 5 }  ,draw opacity=1 ][line width=1.5]    (141,150.5) -- (161,59.5) ;
\draw [color={rgb, 255:red, 65; green, 117; blue, 5 }  ,draw opacity=1 ][line width=1.5]    (160,59.5) -- (191,59.5) ;
\draw [color={rgb, 255:red, 252; green, 3; blue, 22 }  ,draw opacity=1 ][line width=1.5]    (191,59.5) -- (230,190.5) ;
\draw [color={rgb, 255:red, 252; green, 3; blue, 22 }  ,draw opacity=1 ][line width=1.5]    (230,190.5) -- (251,190.5) ;
\draw [color={rgb, 255:red, 252; green, 3; blue, 22 }  ,draw opacity=1 ][line width=1.5]    (270,229.5) -- (291,229.5) ;
\draw [color={rgb, 255:red, 252; green, 3; blue, 22 }  ,draw opacity=1 ][line width=1.5]    (251,190.5) -- (270,229.5) ;
\draw [color={rgb, 255:red, 65; green, 117; blue, 5 }  ,draw opacity=1 ][line width=1.5]    (291,229.5) -- (320,160.5) ;
\draw [color={rgb, 255:red, 252; green, 3; blue, 22 }  ,draw opacity=1 ][line width=1.5]    (320,160.5) -- (360,260.5) ;
\draw [color={rgb, 255:red, 252; green, 3; blue, 22 }  ,draw opacity=1 ][line width=1.5]    (360,260.5) -- (411,260.5) ;
\draw [color={rgb, 255:red, 65; green, 117; blue, 5 }  ,draw opacity=1 ][line width=1.5]    (411,260.5) -- (460,149.5) ;
\draw [color={rgb, 255:red, 65; green, 117; blue, 5 }  ,draw opacity=1 ][line width=1.5]    (460,149.5) -- (510,149.5) ;
\draw [color={rgb, 255:red, 65; green, 117; blue, 5 }  ,draw opacity=1 ][line width=1.5]    (510,149.5) -- (571,150) ;
\draw  [fill={rgb, 255:red, 0; green, 0; blue, 0 }  ,fill opacity=1 ] (187.88,59.5) .. controls (187.88,57.77) and (189.28,56.37) .. (191,56.37) .. controls (192.72,56.37) and (194.12,57.77) .. (194.12,59.5) .. controls (194.12,61.23) and (192.72,62.63) .. (191,62.63) .. controls (189.28,62.63) and (187.88,61.23) .. (187.88,59.5) -- cycle ;
\draw  [dash pattern={on 0.84pt off 2.51pt}]  (191,59.5) -- (191,270) ;
\draw  [dash pattern={on 0.84pt off 2.51pt}]  (291,231.26) -- (291,270.26) ;
\draw  [dash pattern={on 0.84pt off 2.51pt}]  (320,160.5) -- (320,270) ;
\draw  [dash pattern={on 0.84pt off 2.51pt}]  (411,258.74) -- (411,269) ;
\draw  [dash pattern={on 0.84pt off 2.51pt}]  (572.75,151.76) -- (572.75,270) ;
\draw  [fill={rgb, 255:red, 0; green, 0; blue, 0 }  ,fill opacity=1 ] (77.88,270) .. controls (77.88,268.27) and (79.28,266.87) .. (81,266.87) .. controls (82.72,266.87) and (84.12,268.27) .. (84.12,270) .. controls (84.12,271.73) and (82.72,273.13) .. (81,273.13) .. controls (79.28,273.13) and (77.88,271.73) .. (77.88,270) -- cycle ;
\draw  [fill={rgb, 255:red, 0; green, 0; blue, 0 }  ,fill opacity=1 ] (287.88,229.5) .. controls (287.88,227.77) and (289.28,226.37) .. (291,226.37) .. controls (292.72,226.37) and (294.12,227.77) .. (294.12,229.5) .. controls (294.12,231.23) and (292.72,232.63) .. (291,232.63) .. controls (289.28,232.63) and (287.88,231.23) .. (287.88,229.5) -- cycle ;
\draw  [fill={rgb, 255:red, 0; green, 0; blue, 0 }  ,fill opacity=1 ] (316.88,160.5) .. controls (316.88,158.77) and (318.28,157.37) .. (320,157.37) .. controls (321.72,157.37) and (323.12,158.77) .. (323.12,160.5) .. controls (323.12,162.23) and (321.72,163.63) .. (320,163.63) .. controls (318.28,163.63) and (316.88,162.23) .. (316.88,160.5) -- cycle ;
\draw  [fill={rgb, 255:red, 0; green, 0; blue, 0 }  ,fill opacity=1 ] (407.88,260.5) .. controls (407.88,258.77) and (409.28,257.37) .. (411,257.37) .. controls (412.72,257.37) and (414.12,258.77) .. (414.12,260.5) .. controls (414.12,262.23) and (412.72,263.63) .. (411,263.63) .. controls (409.28,263.63) and (407.88,262.23) .. (407.88,260.5) -- cycle ;
\draw  [fill={rgb, 255:red, 0; green, 0; blue, 0 }  ,fill opacity=1 ] (571,150) .. controls (571,148.27) and (572.4,146.87) .. (574.12,146.87) .. controls (575.85,146.87) and (577.25,148.27) .. (577.25,150) .. controls (577.25,151.73) and (575.85,153.13) .. (574.12,153.13) .. controls (572.4,153.13) and (571,151.73) .. (571,150) -- cycle ;
\draw  [draw opacity=0][fill={rgb, 255:red, 65; green, 117; blue, 5 }  ,fill opacity=1 ] (458.44,149.5) .. controls (458.44,148.64) and (459.14,147.94) .. (460,147.94) .. controls (460.86,147.94) and (461.56,148.64) .. (461.56,149.5) .. controls (461.56,150.36) and (460.86,151.06) .. (460,151.06) .. controls (459.14,151.06) and (458.44,150.36) .. (458.44,149.5) -- cycle ;
\draw  [draw opacity=0][fill={rgb, 255:red, 252; green, 3; blue, 22 }  ,fill opacity=1 ] (358.44,260.5) .. controls (358.44,259.64) and (359.14,258.94) .. (360,258.94) .. controls (360.86,258.94) and (361.56,259.64) .. (361.56,260.5) .. controls (361.56,261.36) and (360.86,262.06) .. (360,262.06) .. controls (359.14,262.06) and (358.44,261.36) .. (358.44,260.5) -- cycle ;
\draw  [draw opacity=0][fill={rgb, 255:red, 252; green, 3; blue, 22 }  ,fill opacity=1 ] (268.44,229.5) .. controls (268.44,228.64) and (269.14,227.94) .. (270,227.94) .. controls (270.86,227.94) and (271.56,228.64) .. (271.56,229.5) .. controls (271.56,230.36) and (270.86,231.06) .. (270,231.06) .. controls (269.14,231.06) and (268.44,230.36) .. (268.44,229.5) -- cycle ;
\draw  [draw opacity=0][fill={rgb, 255:red, 252; green, 3; blue, 22 }  ,fill opacity=1 ] (249.44,190.5) .. controls (249.44,189.64) and (250.14,188.94) .. (251,188.94) .. controls (251.86,188.94) and (252.56,189.64) .. (252.56,190.5) .. controls (252.56,191.36) and (251.86,192.06) .. (251,192.06) .. controls (250.14,192.06) and (249.44,191.36) .. (249.44,190.5) -- cycle ;
\draw  [draw opacity=0][fill={rgb, 255:red, 252; green, 3; blue, 22 }  ,fill opacity=1 ] (228.44,190.5) .. controls (228.44,189.64) and (229.14,188.94) .. (230,188.94) .. controls (230.86,188.94) and (231.56,189.64) .. (231.56,190.5) .. controls (231.56,191.36) and (230.86,192.06) .. (230,192.06) .. controls (229.14,192.06) and (228.44,191.36) .. (228.44,190.5) -- cycle ;
\draw  [draw opacity=0][fill={rgb, 255:red, 65; green, 117; blue, 5 }  ,fill opacity=1 ] (159.44,59.5) .. controls (159.44,58.64) and (160.14,57.94) .. (161,57.94) .. controls (161.86,57.94) and (162.56,58.64) .. (162.56,59.5) .. controls (162.56,60.36) and (161.86,61.06) .. (161,61.06) .. controls (160.14,61.06) and (159.44,60.36) .. (159.44,59.5) -- cycle ;
\draw  [draw opacity=0][fill={rgb, 255:red, 65; green, 117; blue, 5 }  ,fill opacity=1 ] (139.44,150.5) .. controls (139.44,149.64) and (140.14,148.94) .. (141,148.94) .. controls (141.86,148.94) and (142.56,149.64) .. (142.56,150.5) .. controls (142.56,151.36) and (141.86,152.06) .. (141,152.06) .. controls (140.14,152.06) and (139.44,151.36) .. (139.44,150.5) -- cycle ;
\draw  [draw opacity=0][fill={rgb, 255:red, 65; green, 117; blue, 5 }  ,fill opacity=1 ] (109.44,150.5) .. controls (109.44,149.64) and (110.14,148.94) .. (111,148.94) .. controls (111.86,148.94) and (112.56,149.64) .. (112.56,150.5) .. controls (112.56,151.36) and (111.86,152.06) .. (111,152.06) .. controls (110.14,152.06) and (109.44,151.36) .. (109.44,150.5) -- cycle ;

\draw (19,133) node [anchor=north west][inner sep=0.75pt]   [align=left] {$\displaystyle \Bell( v)$};
\draw (66.5,251) node [anchor=north west][inner sep=0.75pt]   [align=left] {$\displaystyle 0$};
\draw (60.5,23) node [anchor=north west][inner sep=0.75pt]   [align=left] {$\displaystyle h$};
\draw (311,318) node [anchor=north west][inner sep=0.75pt]   [align=left] {$\displaystyle t$};
\draw (553,273) node [anchor=north west][inner sep=0.75pt]   [align=left] {$\displaystyle T=r_{3}$};
\draw (82,273) node [anchor=north west][inner sep=0.75pt]   [align=left] {$\displaystyle 0=s_{1}$};
\draw (185,273) node [anchor=north west][inner sep=0.75pt]   [align=left] {$\displaystyle r_{1}$};
\draw (283,273) node [anchor=north west][inner sep=0.75pt]   [align=left] {$\displaystyle s_{2}$};
\draw (315,273) node [anchor=north west][inner sep=0.75pt]   [align=left] {$\displaystyle r_{2}$};
\draw (404,273) node [anchor=north west][inner sep=0.75pt]   [align=left] {$\displaystyle s_{3}$};

\end{tikzpicture}

%% file: diagrams/source-charge.tex
\tikzset{every picture/.style={line width=0.75pt}} 

\begin{tikzpicture}[x=0.75pt,y=0.75pt,yscale=-1,xscale=1]

\draw  [fill={rgb, 255:red, 232; green, 249; blue, 210 }  ,fill opacity=1 ][line width=1.5]  (105.3,210.92) .. controls (105.3,102.72) and (193.02,15) .. (301.23,15) .. controls (409.43,15) and (497.15,102.72) .. (497.15,210.92) .. controls (497.15,319.13) and (409.43,406.84) .. (301.23,406.84) .. controls (193.02,406.84) and (105.3,319.13) .. (105.3,210.92) -- cycle ;
\draw  [draw opacity=0][fill={rgb, 255:red, 248; green, 235; blue, 236 }  ,fill opacity=1 ][line width=1.5]  (141.71,96.22) .. controls (177.69,96.19) and (224.75,150.88) .. (252.06,227.4) .. controls (279.74,304.96) and (277.51,378) .. (248.55,399.7) -- (185.42,251.18) -- cycle ; \draw  [line width=1.5]  (141.71,96.22) .. controls (177.69,96.19) and (224.75,150.88) .. (252.06,227.4) .. controls (279.74,304.96) and (277.51,378) .. (248.55,399.7) ;  
\draw [fill={rgb, 255:red, 248; green, 235; blue, 236 }  ,fill opacity=1 ][line width=1.5]    (141.71,96.22) .. controls (44,250) and (162,386) .. (248.47,399.44) ;
\draw  [dash pattern={on 4.5pt off 4.5pt}]  (303,113) -- (334.88,160.34) ;
\draw [shift={(336,162)}, rotate = 236.04] [color={rgb, 255:red, 0; green, 0; blue, 0 }  ][line width=0.75]    (10.93,-3.29) .. controls (6.95,-1.4) and (3.31,-0.3) .. (0,0) .. controls (3.31,0.3) and (6.95,1.4) .. (10.93,3.29)   ;
\draw  [draw opacity=0][line width=1.5]  (164.21,272.45) .. controls (182.96,283.03) and (202.88,298.73) .. (220.07,317.45) .. controls (247.56,347.39) and (260.22,376.52) .. (254.65,392.28) -- (177.14,331.06) -- cycle ; \draw  [line width=1.5]  (164.21,272.45) .. controls (182.96,283.03) and (202.88,298.73) .. (220.07,317.45) .. controls (247.56,347.39) and (260.22,376.52) .. (254.65,392.28) ;  
\draw  [draw opacity=0][line width=1.5]  (121.11,137.35) .. controls (131.84,139.79) and (143.11,147.51) .. (153.08,160.38) .. controls (176.23,190.27) and (183.6,236.62) .. (169.54,263.92) .. controls (157.87,286.55) and (135.13,288.96) .. (114.5,272.07) -- (127.62,209.79) -- cycle ; \draw  [line width=1.5]  (121.11,137.35) .. controls (131.84,139.79) and (143.11,147.51) .. (153.08,160.38) .. controls (176.23,190.27) and (183.6,236.62) .. (169.54,263.92) .. controls (157.87,286.55) and (135.13,288.96) .. (114.5,272.07) ;  
\draw [color={rgb, 255:red, 245; green, 166; blue, 35 }  ,draw opacity=1 ][line width=1.5]    (332.92,163.82) .. controls (300.42,183.11) and (290.47,193.07) .. (241,211.5) .. controls (190,230.5) and (197.15,307.87) .. (186,308.5) .. controls (174.85,309.13) and (95,203) .. (147,128.5) .. controls (153.4,119.32) and (160.03,112.6) .. (166.8,107.81) .. controls (214.98,73.68) and (270.29,137.53) .. (303,113) ;
\draw [shift={(336,162)}, rotate = 149.64] [color={rgb, 255:red, 245; green, 166; blue, 35 }  ,draw opacity=1 ][line width=1.5]    (14.21,-4.28) .. controls (9.04,-1.82) and (4.3,-0.39) .. (0,0) .. controls (4.3,0.39) and (9.04,1.82) .. (14.21,4.28)   ;
\draw  [fill={rgb, 255:red, 0; green, 0; blue, 0 }  ,fill opacity=1 ] (299.49,113) .. controls (299.49,111.06) and (301.06,109.49) .. (303,109.49) .. controls (304.94,109.49) and (306.51,111.06) .. (306.51,113) .. controls (306.51,114.94) and (304.94,116.51) .. (303,116.51) .. controls (301.06,116.51) and (299.49,114.94) .. (299.49,113) -- cycle ;
\draw  [fill={rgb, 255:red, 0; green, 0; blue, 0 }  ,fill opacity=1 ] (332.49,162) .. controls (332.49,160.06) and (334.06,158.49) .. (336,158.49) .. controls (337.94,158.49) and (339.51,160.06) .. (339.51,162) .. controls (339.51,163.94) and (337.94,165.51) .. (336,165.51) .. controls (334.06,165.51) and (332.49,163.94) .. (332.49,162) -- cycle ;
\draw [line width=1.5]    (231,216) -- (250.22,208.14) ;
\draw [shift={(253,207)}, rotate = 157.76] [color={rgb, 255:red, 0; green, 0; blue, 0 }  ][line width=1.5]    (14.21,-4.28) .. controls (9.04,-1.82) and (4.3,-0.39) .. (0,0) .. controls (4.3,0.39) and (9.04,1.82) .. (14.21,4.28)   ;
\draw  [fill={rgb, 255:red, 65; green, 117; blue, 5 }  ,fill opacity=1 ] (249.49,207) .. controls (249.49,205.06) and (251.06,203.49) .. (253,203.49) .. controls (254.93,203.49) and (256.5,205.06) .. (256.5,207) .. controls (256.5,208.94) and (254.93,210.52) .. (253,210.52) .. controls (251.06,210.52) and (249.49,208.94) .. (249.49,207) -- cycle ;
\draw  [fill={rgb, 255:red, 208; green, 2; blue, 27 }  ,fill opacity=1 ] (227.5,216) .. controls (227.5,214.06) and (229.07,212.48) .. (231,212.48) .. controls (232.94,212.48) and (234.51,214.06) .. (234.51,216) .. controls (234.51,217.94) and (232.94,219.51) .. (231,219.51) .. controls (229.07,219.51) and (227.5,217.94) .. (227.5,216) -- cycle ;
\draw [line width=1.5]    (141.5,138.49) -- (133.41,153.83) ;
\draw [shift={(132.01,156.49)}, rotate = 297.81] [color={rgb, 255:red, 0; green, 0; blue, 0 }  ][line width=1.5]    (14.21,-4.28) .. controls (9.04,-1.82) and (4.3,-0.39) .. (0,0) .. controls (4.3,0.39) and (9.04,1.82) .. (14.21,4.28)   ;
\draw  [fill={rgb, 255:red, 208; green, 2; blue, 27 }  ,fill opacity=1 ] (137.99,138.49) .. controls (137.99,136.55) and (139.56,134.98) .. (141.5,134.98) .. controls (143.44,134.98) and (145.01,136.55) .. (145.01,138.49) .. controls (145.01,140.43) and (143.44,142) .. (141.5,142) .. controls (139.56,142) and (137.99,140.43) .. (137.99,138.49) -- cycle ;
\draw  [fill={rgb, 255:red, 208; green, 2; blue, 27 }  ,fill opacity=1 ] (129.5,156.49) .. controls (129.5,154.55) and (131.07,152.98) .. (133.01,152.98) .. controls (134.94,152.98) and (136.51,154.55) .. (136.51,156.49) .. controls (136.51,158.43) and (134.94,160) .. (133.01,160) .. controls (131.07,160) and (129.5,158.43) .. (129.5,156.49) -- cycle ;
\draw [line width=1.5]    (153.01,269) -- (161.71,280.61) ;
\draw [shift={(163.51,283.01)}, rotate = 233.15] [color={rgb, 255:red, 0; green, 0; blue, 0 }  ][line width=1.5]    (14.21,-4.28) .. controls (9.04,-1.82) and (4.3,-0.39) .. (0,0) .. controls (4.3,0.39) and (9.04,1.82) .. (14.21,4.28)   ;
\draw  [fill={rgb, 255:red, 208; green, 2; blue, 27 }  ,fill opacity=1 ] (148.5,266.49) .. controls (148.5,264.55) and (150.07,262.98) .. (152.01,262.98) .. controls (153.94,262.98) and (155.51,264.55) .. (155.51,266.49) .. controls (155.51,268.43) and (153.94,270) .. (152.01,270) .. controls (150.07,270) and (148.5,268.43) .. (148.5,266.49) -- cycle ;
\draw  [fill={rgb, 255:red, 208; green, 2; blue, 27 }  ,fill opacity=1 ] (160,283.01) .. controls (160,281.07) and (161.57,279.5) .. (163.51,279.5) .. controls (165.44,279.5) and (167.01,281.07) .. (167.01,283.01) .. controls (167.01,284.95) and (165.44,286.52) .. (163.51,286.52) .. controls (161.57,286.52) and (160,284.95) .. (160,283.01) -- cycle ;
\draw [line width=1.5]    (184.01,99.49) -- (164.15,110.08) ;
\draw [shift={(161.5,111.49)}, rotate = 331.93] [color={rgb, 255:red, 0; green, 0; blue, 0 }  ][line width=1.5]    (14.21,-4.28) .. controls (9.04,-1.82) and (4.3,-0.39) .. (0,0) .. controls (4.3,0.39) and (9.04,1.82) .. (14.21,4.28)   ;
\draw  [fill={rgb, 255:red, 65; green, 117; blue, 5 }  ,fill opacity=1 ] (180.5,99.49) .. controls (180.5,97.55) and (182.07,95.98) .. (184.01,95.98) .. controls (185.94,95.98) and (187.51,97.55) .. (187.51,99.49) .. controls (187.51,101.43) and (185.94,103) .. (184.01,103) .. controls (182.07,103) and (180.5,101.43) .. (180.5,99.49) -- cycle ;
\draw  [fill={rgb, 255:red, 208; green, 2; blue, 27 }  ,fill opacity=1 ] (157.99,111.49) .. controls (157.99,109.55) and (159.56,107.98) .. (161.5,107.98) .. controls (163.44,107.98) and (165.01,109.55) .. (165.01,111.49) .. controls (165.01,113.43) and (163.44,115) .. (161.5,115) .. controls (159.56,115) and (157.99,113.43) .. (157.99,111.49) -- cycle ;

\draw (388,60) node [anchor=north west][inner sep=0.75pt]  [font=\Large] [align=left] {$\displaystyle A_t$};
\draw (307,90) node [anchor=north west][inner sep=0.75pt]  [font=\large] [align=left] {$\displaystyle u$};
\draw (342,145) node [anchor=north west][inner sep=0.75pt]  [font=\large] [align=left] {$\displaystyle v$};
\draw (164,84) node [anchor=north west][inner sep=0.75pt]  [font=\large] [align=left] {$\displaystyle e$};
\draw (113,238) node [anchor=north west][inner sep=0.75pt]  [font=\Large] [align=left] {$\displaystyle C_{1}$};
\draw (168,327) node [anchor=north west][inner sep=0.75pt]  [font=\Large] [align=left] {$\displaystyle C_{2}$};
\draw (212,285) node [anchor=north west][inner sep=0.75pt]  [font=\Large] [align=left] {$\displaystyle C_{3}$};
\draw (242,70) node [anchor=north west][inner sep=0.75pt]  [font=\Large] [align=left] {$\displaystyle P$};
\draw (223,196) node [anchor=north west][inner sep=0.75pt]  [font=\large] [align=left] {$\displaystyle e'$};
\draw (135.01,155.98) node [anchor=north west][inner sep=0.75pt]  [font=\large] [align=left] {$\displaystyle f$};
\draw (143,284.5) node [anchor=north west][inner sep=0.75pt]  [font=\large] [align=left] {$\displaystyle f'$};

\end{tikzpicture}